\documentclass[acmsmall]{acmart}

\setcopyright{cc}
\setcctype{by}
\acmJournal{PACMPL}
\acmYear{2025} \acmVolume{9} \acmNumber{OOPSLA2} \acmArticle{343} \acmMonth{10} \acmPrice{}\acmDOI{10.1145/3763121}

\usepackage{math-setup}
\usepackage{todonotes}
\usepackage{enumitem}
\usepackage{wrapfig}
\usepackage{sidecap}
\usepackage{caption}
\usepackage{subfigure}
\usepackage{mathrsfs}
\usepackage{mdframed}
\usepackage{changepage} 

\begin{document}

\title{Automated Discovery of Tactic Libraries for Interactive Theorem Proving}

\newcommand{\jc}[1]{{\color{red}\{JC: #1\}}}
\newcommand{\jx}[1]{{\color{purple}\{jimmy: #1\}}}
\newcommand{\gp}[1]{{\color{brown}\{GP: #1\}}}
\newcommand{\mx}[1]{{\color{teal}\{mx: #1\}}}

\newcommand{\missing}[1]{{\color{red} #1}}

\author{Yutong Xin}
\email{maxryeery@utexas.edu}
\affiliation{
    \institution{University of Texas at Austin}
    \city{Austin}
    \state{Texas}
    \country{USA}
}

\author{Jimmy Xin}
\email{jxin31415@utexas.edu}
\affiliation{
    \institution{University of Texas at Austin}
    \city{Austin}
    \state{Texas}
    \country{USA}
}

\author{Gabriel Poesia}
\email{poesia@stanford.edu}
\affiliation{
    \institution{Stanford University}
    \city{Stanford}
    \state{California}
    \country{USA}
}

\author{Noah Goodman}
\email{ngoodman@stanford.edu}
\affiliation{
    \institution{Stanford University}
    \city{Stanford}
    \state{California}
    \country{USA}
}

\author{Qiaochu Chen}
\email{qc1127@cs.nyu.edu}
\affiliation{
    \institution{New York University}
    \city{New York}
    \state{New York}
    \country{USA}
}

\author{Isil Dillig}
\email{isil@cs.utexas.edu}
\affiliation{
    \institution{University of Texas at Austin}
    \city{Austin}
    \state{Texas}
    \country{USA}
}
\newcommand{\corpus}{\Pi}
\newcommand{\tname}{\eta}
\newcommand{\targs}{I}
\newcommand{\tret}{O}
\newcommand{\tbody}{\mathcal{E}}
\newcommand{\tinvoke}{t}
\newcommand{\denot}[1]{\llbracket #1 \rrbracket}
\newcommand{\isil}[1]{\textcolor{red}{#1}}
\newcommand{\red}[1]{\textcolor{red}{#1}}
\newcommand{\pscript}{\pi}
\newcommand{\tdg}{\textsf{TDG}\xspace}
\newcommand{\inv}{v_\mathsf{in}}
\newcommand{\outv}{v_\mathsf{out}}
\newcommand{\cp}{\mathsf{CP}}
\newcommand{\tacticlib}{\mathcal{T}}
\newcommand{\tc}{\Psi}
\newcommand{\ev}{\Lambda}
\newcommand{\ws}{\Upsilon}
\newcommand{\freq}{\digamma}
\newcommand{\strength}{\mathscr{E}}
\newcommand{\lib}{\mathscr{L}}
\newcommand{\Ours}{\toolname}
\newcommand{\peano}{{\sc Peano}\xspace}
\newcommand{\nogrammar}{GrammarABL\xspace}
\newcommand{\noub}{PruningABL\xspace}

\newcommand{\changed}[1]{{\color{blue}#1}}
\begin{abstract}
Enabling more concise and modular proofs is essential for advancing formal reasoning using interactive theorem provers (ITPs). Since many ITPs, such as Rocq and Lean, use tactic-style proofs, learning higher-level custom tactics is crucial for proof modularity and automation. This paper presents a novel approach to tactic discovery, which leverages Tactic Dependence Graphs (TDGs) to identify reusable proof strategies across multiple proofs. TDGs capture  logical dependencies between tactic applications while abstracting away irrelevant syntactic details, allowing for both the discovery of new tactics and the refactoring of existing proofs into more modular forms.
We have implemented this technique  in a tool called \toolname and compare it against an anti-unification-based approach ({\sc Peano}) to tactic discovery. Our evaluation demonstrates that \toolname can learn 3$\times$ as many tactics as {\sc Peano} and reduces the size of proofs by 26\% across all benchmarks. Furthermore, our evaluation demonstrates the benefits of learning custom tactics for proof automation, allowing a state-of-the-art proof automation tool to achieve a relative increase of 172\% in terms of success rate.
\end{abstract}



\maketitle

\section{Introduction}\label{sec:intro}

Tactics are essential in interactive theorem proving, particularly for facilitating reuse and modularity in proof development. By encapsulating common proof strategies into reusable components, tactics allow users to apply previously developed solutions to new problems, reducing the effort required to construct proofs from scratch. Tactics also facilitate proof automation, as shorter proofs are typically easier to discover using automated techniques.

While interactive theorem provers (ITPs) such as Rocq (formerly known as Coq) and Lean provide a rich library of built-in tactics, proof engineers typically need to devise \emph{custom tactics} by composing existing tactics into higher-level domain-specific building blocks~\cite{dealmeidaborges_et_al:LIPIcs.ITP.2023.12}. For example, the \texttt{mathlib} library in Lean contains useful tactics for mathematical proofs, and the \texttt{flocq} library in Rocq provides tactics that facilitate proofs about floating point computation. But even if proof assistants enable users to write their own tactics, learning to do so adds an extra layer on top of the already steep learning curve of ITPs.

In this paper, we address the problem of automatically learning \emph{custom tactics} (i.e., compositions of existing tactics) from a given proof corpus.  Learning such tactic libraries can uncover reusable patterns in proofs, making it easier to construct similar proofs within the same domain. Furthermore, tactic  discovery can enhance proof automation by facilitating a form of curriculum learning~\cite{leanagent}, where more complex proofs are synthesized using custom tactics discovered in simpler proofs.

While there has been prior work~\cite{stitch,babble,dreamcoder} on learning software libraries from a given corpus of \emph{programs}, these
methods do not readily apply to tactic-based {proofs}.
    Such proofs are written in an imperative style, with frequent reference to implicit, mutating objects such as the current goal, which is typically not present in the source, but rather visualized interactively. However, existing work on library learning typically assumes functional programs and focuses on {generalizing} concrete program expressions into lambda abstractions.  
In contrast, tactic discovery requires an understanding of the logical relationships between different proof steps and the ability to construct higher-level tactics that capture these relationships.

In this paper, we address this problem by proposing a new proof abstraction called \emph{tactic dependence graph (TDG)} that hides irrelevant syntactic variations between different proofs, while focusing on important \emph{logical} dependencies. In a TDG, nodes represent tactic \emph{applications}, while edges represent \emph{proof state dependencies} between them. For example, Figure~\ref{fig:same-proof} shows two  syntactically different proofs of the same theorem, which have exactly the same TDG abstraction shown in Figure~\ref{fig:tdg-same-proof}. In essence, the TDG abstraction hides minor syntactic variations in the proof, such as how sub-goals are named or in what order tactics are applied, and instead  focuses on \emph{semantic} dependencies in between different tactic applications.

\begin{figure}[t]
    \centering
    \begin{minipage}{0.62\textwidth}
        \centering
\includegraphics[trim= 0 0 0 0, clip, width=\textwidth]{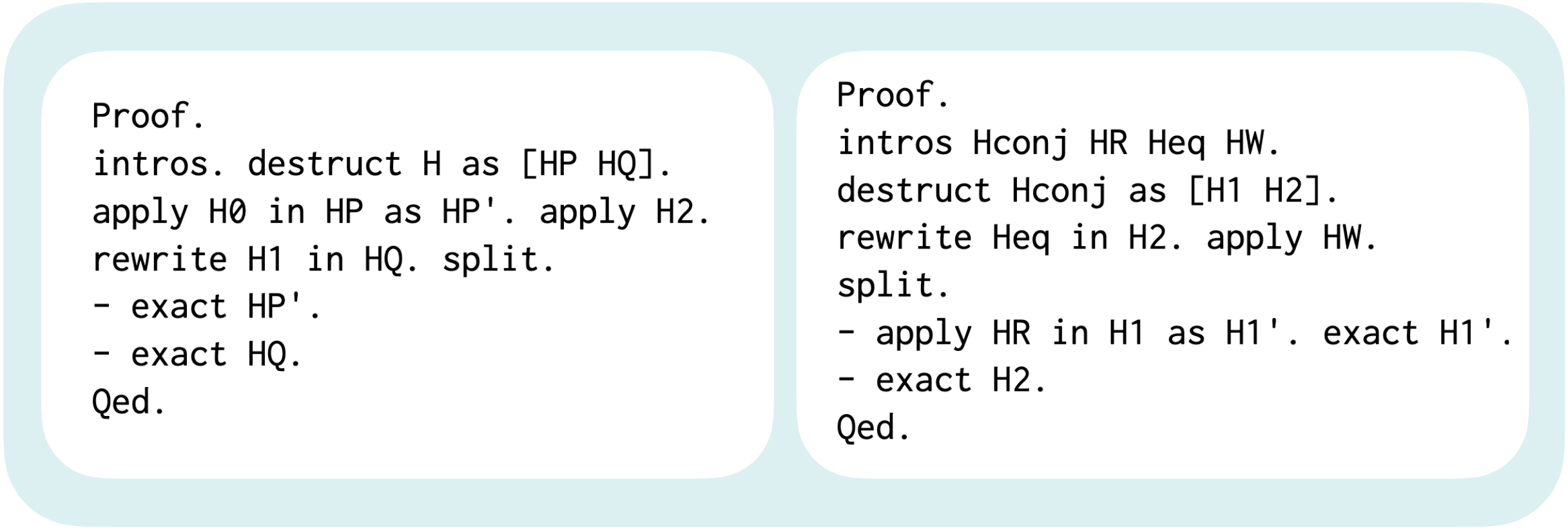}
   \caption{Syntactically different proofs w/ same TDG in Figure~\ref{fig:tdg-same-proof},\\ for Lemma $(P \land Q) \rightarrow (P \rightarrow R) \rightarrow (Q = T) \rightarrow (R \land T \rightarrow W) \rightarrow W$.}
        \label{fig:same-proof}    
    \end{minipage}%
    \hfill
    \begin{minipage}{0.37\textwidth}
        \centering
\includegraphics[width=\textwidth]{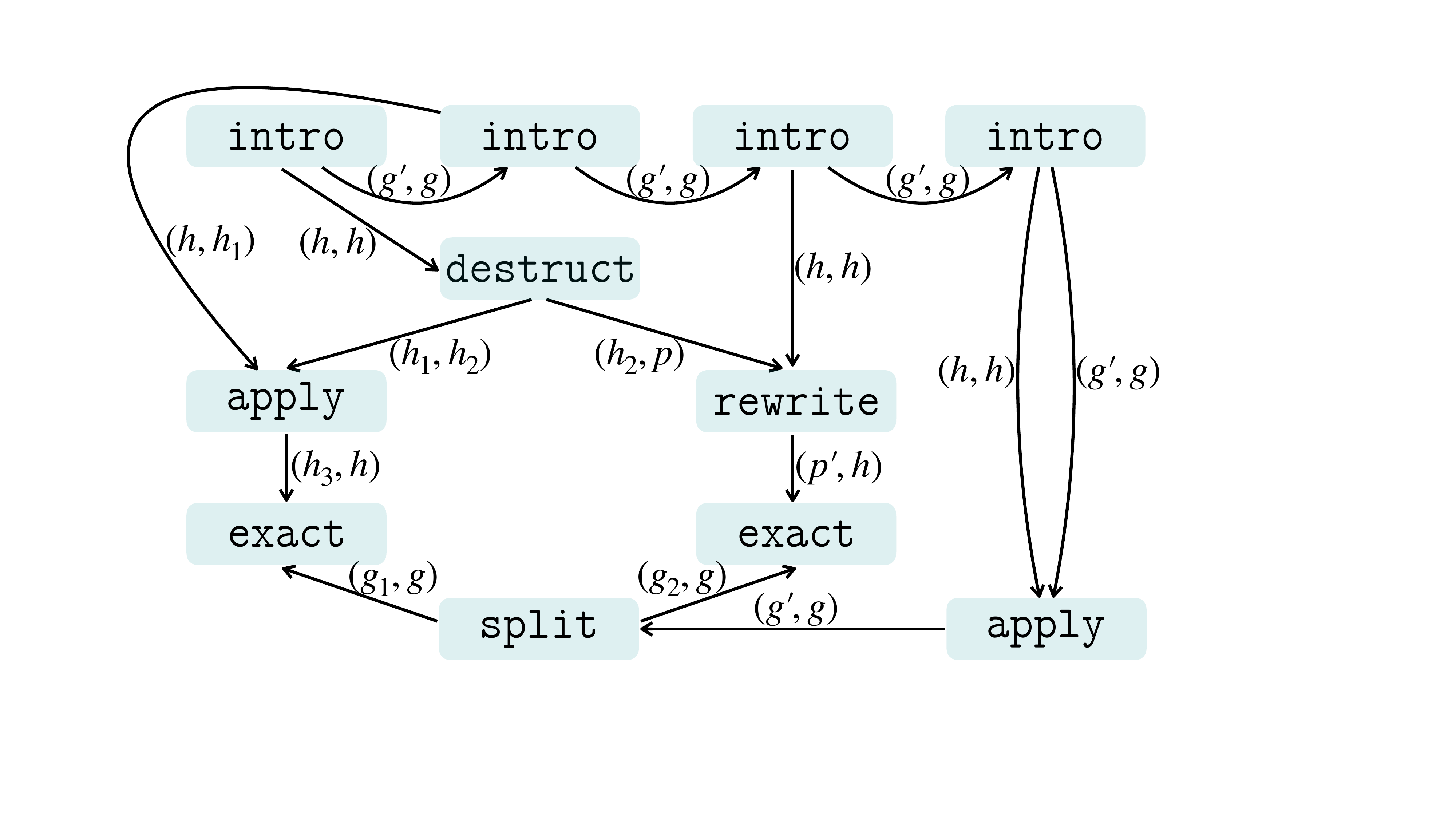}
\vspace{-0.25in}
        \caption{TDG for Figure~\ref{fig:same-proof}.}
        \label{fig:tdg-same-proof}
    \end{minipage}
\vspace{-0.3in}
\end{figure}

{Given a corpus of proofs from the same domain, our approach leverages the TDG abstraction to discover tactics that maximize compression of the proof corpus—a metric widely used in prior work on library learning~\cite{stitch,babble,dreamcoder,peano}. Intuitively, the greater the compression rate of existing proofs, the more broadly applicable the resulting tactics. To achieve this goal, our method identifies \emph{isomorphic subgraphs} within the TDGs that exhibit a special property called \emph{collapsibility}. This property guarantees that (1) each common subgraph can be translated into a valid tactic and (2) the resulting tactics can be applied to refactor the original proofs while preserving their validity.}


Our tactic discovery algorithm is inspired by ideas from \emph{top-down} enumerative program synthesis~\cite{lambda2,myth,neo,stitch}; however, it needs to address additional challenges that are not present in the program synthesis setting. First, since our goal is to discover common isomorphic subgraphs of existing TDGs, we need to perform enumerative search over \emph{graphs} as opposed to abstract syntax \emph{trees}. Second, unlike program synthesis, there is no pre-determined grammar for tactic dependence graphs. Finally, since our goal is to  maximize an optimization objective, the synthesizer cannot stop as soon as it finds an isomorphic embedding  but must keep going until it identifies the highest-scoring tactic. Our proposed tactic discovery algorithm effectively addresses these challenges using two key ideas. First, it extracts a graph grammar from the TDGs of existing proofs. Second, it substantially prunes the search space by deriving upper bounds on the compression power of candidate tactics. 

We have implemented our proposed tactic discovery method in a tool called \toolname and evaluate it on several Rocq projects, including CompCert~\cite{compcert} and proofs involving program logics. We  compare our approach against an anti-unification baseline from prior work ({\sc Peano})~\cite{peano} and show that \toolname can learn $3\times$ as many tactics compared to {\sc Peano}. Furthermore,  the tactics learned by \toolname  make proofs $26\%$ shorter across all benchmarks, reducing the corpus size to 63\% of its original size in some cases. Finally, our evaluation also demonstrates the benefits of our approach for proof automation, allowing a state-of-the-art tool ({\sc Copra})~\cite{copra} to achieve a relative increase of 172\% in the number of theorems proved. 

To summarize, this paper makes the following contributions: 
\begin{itemize}[leftmargin=*]
    \item  We introduce \emph{tactic dependence graphs (TDG)}, a new abstraction that facilitates proof refactoring and tactic discovery.
    \item We define the \emph{semantic proof refactoring} problem and show how to use the TDG abstraction to refactor proofs in a way that preserves their validity.  
    \item We propose a new tactic discovery algorithm for assembling a library of tactics from a  corpus. 
    \item We implement these ideas in a tool called \toolname targeting Rocq proofs and   compare it against a baseline~\cite{peano} that uses anti-unification for tactic discovery. Our method learns around $3\times$ as many tactics compared to the baseline, reducing the size of the proof corpus by 26\% (compared to 9\% for the baseline). Furthermore, the tactics learned by our method enabled {\sc Copra}, a proof automation tool, to increase the number of theorems it can prove by 172\% on a corpus of 50 proofs, providing preliminary evidence of these tactics' utility in automated theorem proving.
    \end{itemize}

\section{Motivating Example}\label{sec:motivation}

\begin{wrapfigure}{r}{0.6\textwidth}
    \centering
    \vspace{-0.2in}
    \includegraphics[width=0.6\textwidth]{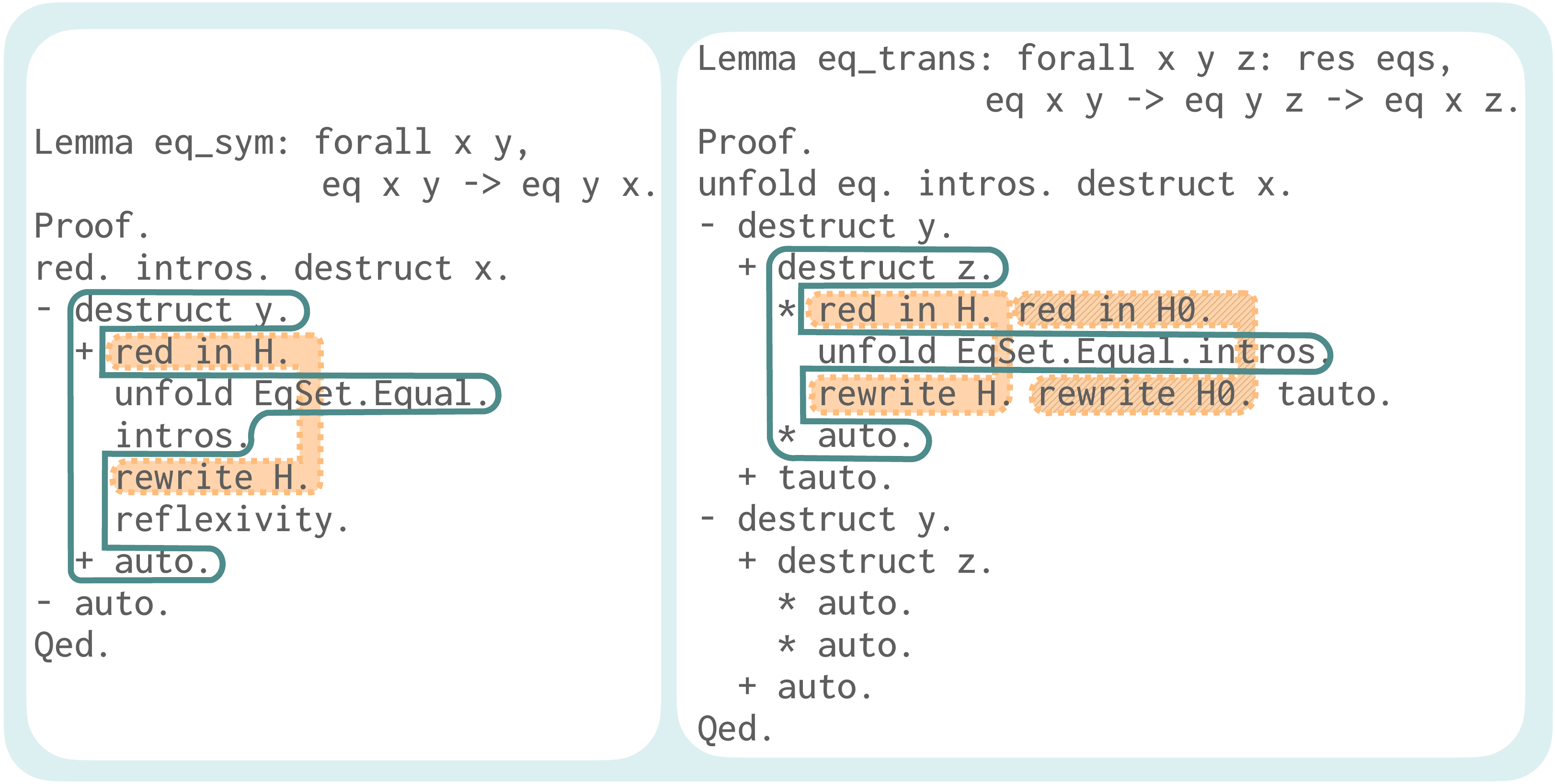}
    \vspace{-0.2in}
    \caption{Input Rocq Proofs.}
    \label{fig:motEx}
    \vspace{-0.1in}
\end{wrapfigure}
In this section, we demonstrate our approach through a simple motivating example shown in Figure~\ref{fig:motEx}. This figure shows two simple Rocq proofs, \texttt{eq\_sym} and \texttt{eq\_trans}, taken from CompCert~\cite{compcert} for establishing the symmetry and transitivity of a predicate defining a suitable notion of equality for dataflow analysis results. As standard in Rocq and many other interactive theorem provers, these are \emph{tactic-style proofs},  where pre-defined tactics  such as \texttt{intros, destruct, rewrite}, etc. are used to transform proof goals into simpler sub-goals. Our goal in this paper is to learn \emph{higher-level (custom) tactics} that can be used to further simplify proof engineering in a specific domain.   

While the two proofs shown in Figure~\ref{fig:motEx} have salient differences, they also have many similarities, as highlighted in color in the figure. In fact, given these two proofs, we can extract the following custom tactics that may be used to simplify several other proofs:

\vspace{0.05in}
\begin{mdframed}
{\small \begin{verbatim}
Ltac simplRewrite H := red in H; rewrite H. 
Ltac destructUnfold n H := destruct n; [unfold H; intros | auto].
\end{verbatim}}
\end{mdframed}
\vspace{0.05in}

The first custom tactic called \texttt{simplRewrite}  is useful in hypothesis \texttt{H} which contains a function or expression that requires some reduction before rewriting. In particular, this tactic first simplifies the hypothesis \texttt{H} by performing one step of reduction, and then uses that simplified hypothesis to rewrite the goal. The second tactic, called \texttt{destructUnfold}, is useful in scenarios that require  handling different cases of an inductive type through unfolding and proof automation, respectively.  Figure~\ref{fig:motEx-rewrite} shows the refactored version of the  proofs from Figure~\ref{fig:motEx} using these custom tactics. 

\begin{wrapfigure}{r}{0.6\textwidth}
    \centering
    \vspace{-0.12in}
    \includegraphics[trim=100 200 340 150, clip,  width=0.6\textwidth]{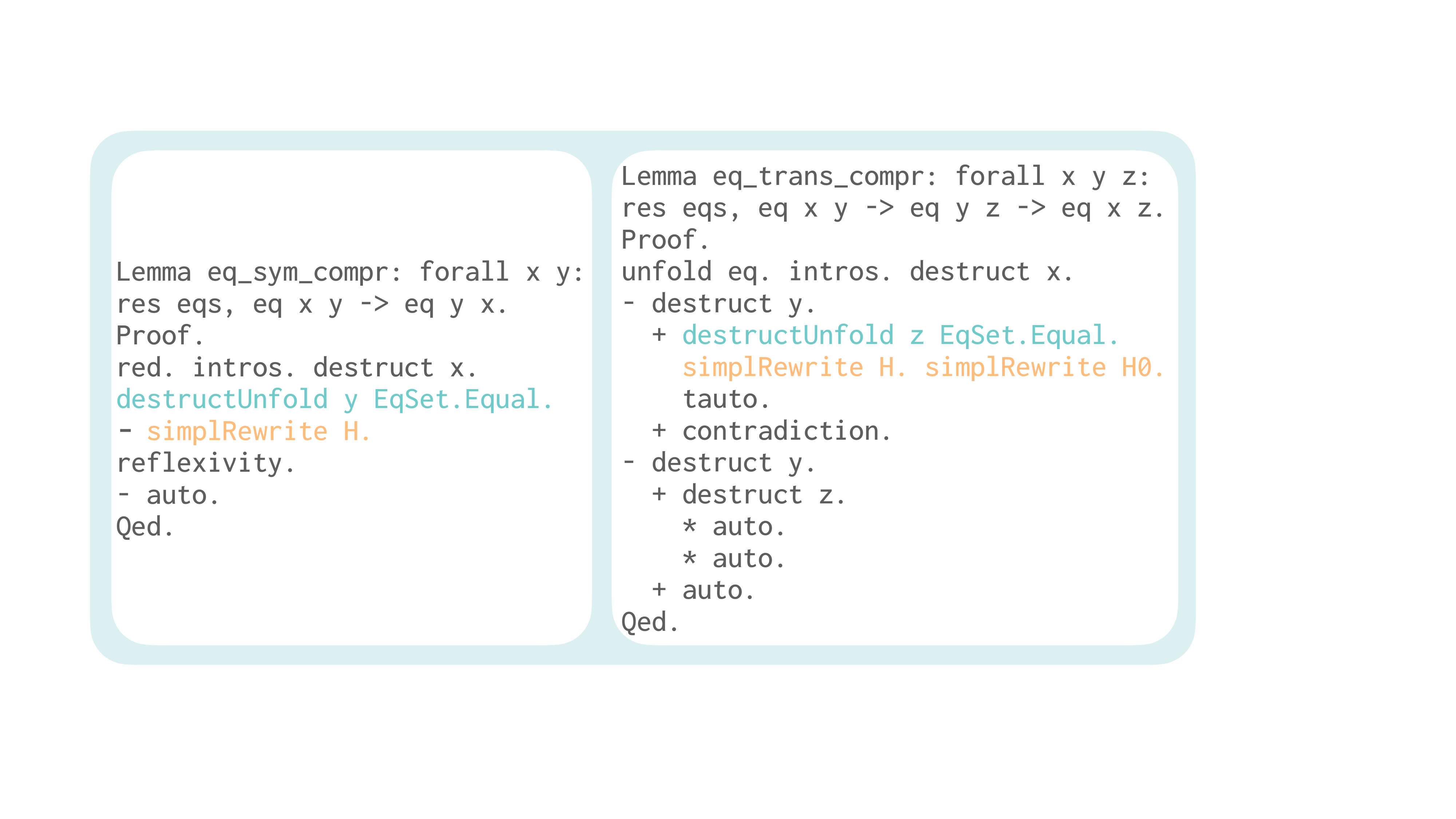}
    \caption{Refactored proofs from Figure~\ref{fig:motEx} using the custom tactics.}
    \label{fig:motEx-rewrite}
    \vspace{-0.12in}
\end{wrapfigure}

\vspace{0.05in}
\noindent
{\bf{\emph{Proof Refactoring.}}} As a first step towards tactic discovery, we pose the question \emph{``Can we, and, if so how do we, refactor proofs using a given custom tactic?''}. At first glance, it is unclear how to refactor the proofs from Figure~\ref{fig:motEx} to use our custom tactics. First, the pre-defined tactics comprising the custom tactic do not appear consecutively in either of the proofs.  Second, the \texttt{destructUnfold} tactic spans multiple cases of the original proof, but  its first branch does not exactly correspond to either of the branches of the original tactic. Third, even if a proof uses the same pre-defined tactics as a custom tactic, this does not necessarily mean that a proof can be refactored using that tactic. In particular, because tactic-style proofs are stateful objects with implicit sub-goals, there can be subtle dependencies between tactic applications that make it impossible to refactor the proof using a given tactic \emph{even when} the proof and the tactic are \emph{syntactically} similar. All of these considerations highlight the need for a suitable {abstraction} that can facilitate  proof refactoring and tactic discovery.
\begin{wrapfigure}{r}{0.5\textwidth}
    \centering
    \vspace{-0.12in}
    \includegraphics[width=0.5\textwidth]{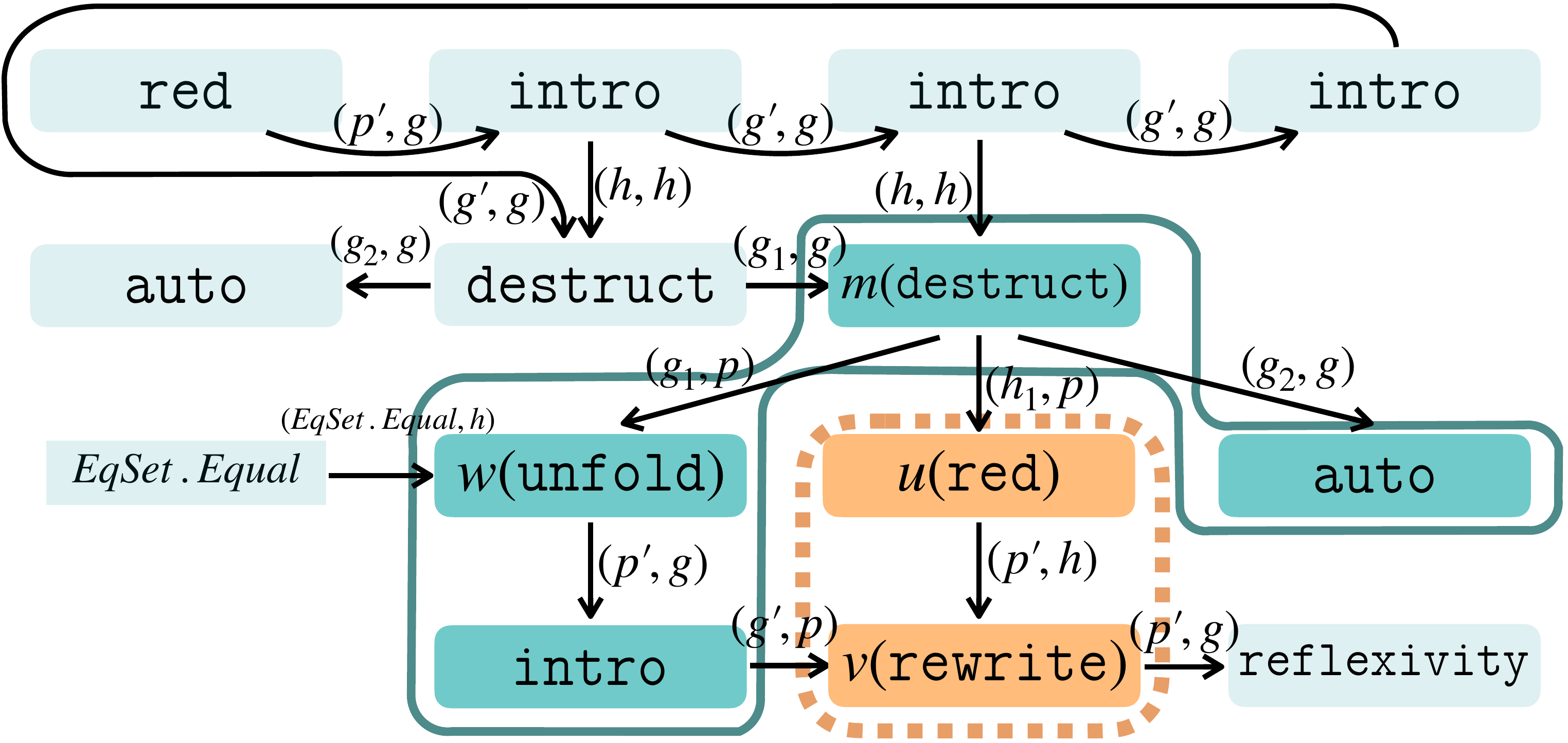}
     \vspace{-0.2in}
    \caption{TDG for {\tt eq\_sym}.} 
    \label{fig:mot-p1}
    \vspace{-0.15in}
\end{wrapfigure}

\vspace{0.05in}
\noindent
{\bf{\emph{Tactic Dependence Graph.}}} To address these problems, we propose representing tactic-style proofs using an abstraction called the \emph{tactic dependence graph (TDG)} that elucidates dependencies between different tactic applications in a proof.  Specifically, nodes in a TDG correspond to tactic applications, and edges encode how the ``outputs'' of one tactic  feed as ``inputs'' to another tactic. For instance, consider the TDG representation of the proof for the \texttt{eq\_sym} lemma shown in Figure~\ref{fig:mot-p1}. Here, an edge from a node $n$ to $n'$ indicates that the tactic application represented by $n'$ depends on the tactic application represented by $n$. For example, looking at this TDG, we see that there is no logical dependence between the second application of \texttt{red} (node $u$) and \texttt{unfold} (node $w$) even though they appear consecutively in the syntax of the proof. On the other hand,  the TDG also makes it clear that there is an immediate dependency between the second application of \texttt{red} ($u$) and \texttt{rewrite} ($v$) even though they are separated in syntax by two other tactic applications. Also, note that TDG edges elucidate not only \emph{whether} there exists a dependency between two tactic applications but also \emph{how} they are dependent. In particular,  an edge  labeled $(\alpha, \beta)$ indicates that the ``formal'' output $\alpha$ of the source tactic application is supplied as the ``formal'' parameter $\beta$ of the target tactic application. As a matter of convention, we use the symbols $h, h_i$ to denote hypotheses, $g, g_i$ to denote goals, and $p, p_i$ to denote propositions that can be either goals or hypotheses. For instance, looking at the same TDG, we see that \texttt{destruct} (node $m$) produces two goals (labeled $g_1, g_2$) and a new hypothesis  (labeled $h_1$).  We also notice that the formal outputs  $g_1, g_2$ of \texttt{destruct} (node $m$) are used by \texttt{unfold} (node $w$) and \texttt{auto} respectively, and that the output hypothesis $h_1$ is used by \texttt{red} (node $u$). 

\vspace{0.05in}
\noindent
{\bf{\emph{TDG for Proof Refactoring.}}} Next, let's consider how we might use the TDG for refactoring the \texttt{eq\_sym} proof from Figure~\ref{fig:motEx} into the form shown in Figure~\ref{fig:motEx-rewrite}. First, observe that the subgraph encircled using solid (green) lines in Figure~\ref{fig:mot-p1} corresponds to the TDG representation of the \texttt{destructUnfold} tactic and the subgraph surrounded by orange dashed lines corresponds to that of  \texttt{simplRewrite}. Hence, we can refactor the original proof by replacing this sequence of tactic invocations with the \texttt{destructUnfold} and \texttt{simplRewrite} tactics. Next, there is the question of \emph{where} in the proof these custom tactic invocations belong. Looking at the TDG, we see that an input of  \texttt{red} (now part of \texttt{simplRewrite}) relies on an output of  \texttt{destruct} (which is now a part of \texttt{destructUnfold}); hence, we deduce that  \texttt{destructUnfold} must occur before \texttt{simplRewrite} in the refactored proof. Using this type of reasoning, we can refactor the original proof from Figure~\ref{fig:motEx} to the version shown in Figure~\ref{fig:motEx-rewrite} while maintaining its syntactic and semantic validity. 

\begin{figure}
    \centering
    \includegraphics[width=0.8\textwidth]{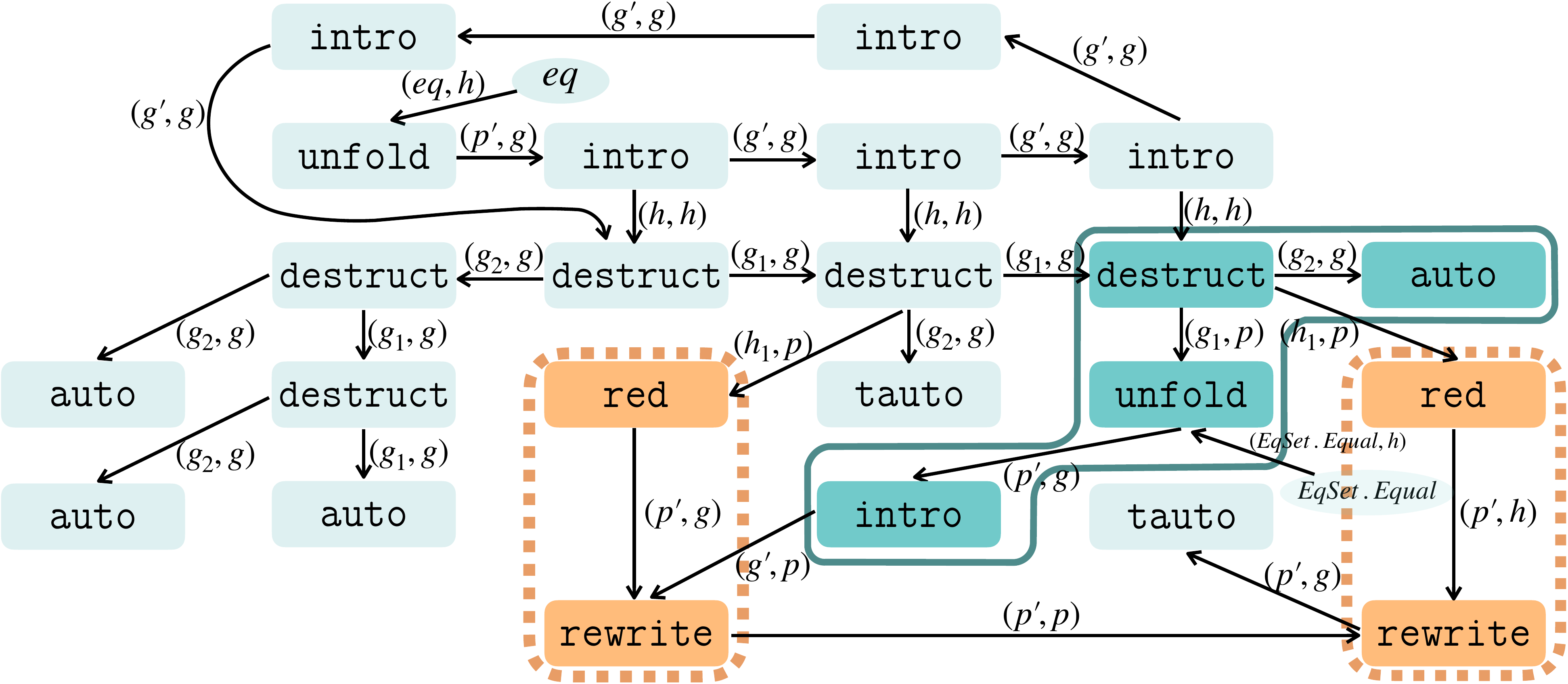}
    \vspace{-0.1in}
    \caption{TDG for {\tt eq\_trans}.}
    \label{fig:mot-p2}
    \vspace{-0.2in}
\end{figure}

\vspace{0.05in}
\noindent
{\bf{\emph{Tactic Discovery.}}} In addition to being useful for proof refactoring, the TDG abstraction is also beneficial for \emph{tactic discovery}: Given a corpus of proofs, we can identify common proof patterns by looking for common isomorphic subgraphs of the TDGs. For instance, comparing the TDGs shown in Figure~\ref{fig:mot-p1} and Figure~\ref{fig:mot-p2} immediately reveals two common isomorphic subgraphs, highlighted in solid (green) and dashed (orange) lines. Hence, the idea behind our tactic discovery algorithm is to find  isomorphic common subgraphs in the corpus and use them to refactor the proofs. However, when learning tactics, there is often a trade-off between  the generality of the tactic and its size: On one hand, the more frequently a proof pattern appears, the more generally applicable it is. On the other hand, the larger the tactic, the more useful it is --- intuitively, large tactics allow us to skip many steps in the proof. To balance this trade-off, our tactic discovery algorithm looks for common isomorphic subgraphs that maximize the total reduction in corpus size, a metric that has also been used in prior work on library learning~\cite{stitch,babble}. For instance, going back to our running example, the \texttt{destructUnfold} tactic reduces the size of the two proofs by $20\%$ whereas the \texttt{simplRewrite} tactic reduces size by $10\%$.  Thus, our learning algorithm first refactors the corpus using the \texttt{destructUnfold} tactic and then looks for a different tactic that maximizes compression, such as the \texttt{simplRewrite} tactic in our example. This process 
continues until no new tactics can be discovered. 
\section{Preliminaries}

This section provides necessary background on interactive theorem proving by developing a formal model of tactic-style proofs.



\subsection{Proof States}

A central concept in tactic-style proofs is the notion of \emph{proof state}, which represents the current status of an ongoing proof, including the goals that remain to be established and the hypotheses available at each step.

 \begin{definition}[{\bf Proof state}]
 A proof state $\pstate$ is a mapping from identifiers (i.e., names like \texttt{H1}) to \emph{proof elements}, which correspond to  hypotheses and goals. A \emph{hypothesis}  is  a proposition $\varphi$ expressed in the formal language of the prover. A \emph{goal} $g$  is a pair $(\varphi, \mathcal{C})$ where $\varphi$ is a proposition to be proven and the \emph{context} $\mathcal{C}$ is a set of hypothesis identifiers that are in scope when proving $g$.
 \end{definition}

Given an identifier $v$, we write $g(v)$ to indicate that it represents a proof goal and $h(v)$ to denote that it corresponds to a hypothesis. {We use the notation $\textsf{Goals}(\pstate)$ to denote the set of goals in the proof state $\pstate$. } Also, given a proof state $\pstate$ and a set of identifiers $I$, we write  $\pstate \backslash I$ to denote the new proof state obtained by removing identifiers $I$ from $\pstate$.

For readability, we omit the context associated with a goal unless necessary. Thus, in most examples, we represent the proof state as a mapping from identifiers to propositions, assuming that all hypotheses are available when proving a goal unless explicitly stated. Additionally, while many ITPs do not explicitly name proof goals—typically operating on a single implicit goal at each step—we choose to explicitly name goals to facilitate disambiguation in multi-goal scenarios.

\subsection{Tactics}
In interactive theorem provers, theorems are typically proven by applying \emph{tactics} to manipulate proof states.  For example, Table~\ref{tab:built-in} lists our representation of some of the built-in  tactics in the Rocq theorem prover, along with an (informal) description of their semantics.\footnote{Note that Table~\ref{tab:built-in} only provides examples; our implementation supports many more, as explained in Section~\ref{sec:impl}.} For instance, consider the \texttt{intro} tactic shown in Table~\ref{tab:built-in} which is listed as having ``signature'' $g \rightarrow h \times g'$. This means that this tactic operates over a specific goal $g$ that is part of the current proof state $\pstate$ and transforms the proof state if $g$ is of the form $p_1 \rightarrow p_2$. In particular, the new proof state is obtained by removing goal $g$ from the current proof state and introducing $p_1$ as a new hypothesis and $p_2$ as a new goal.  Note that some tactics, such as \texttt{split} and \texttt{destruct}, can also introduce multiple new goals, where each goal can have its own unique context (such as in the case of \texttt{destruct}). 

In addition to using built-in tactics, proof engineers can also define their own \emph{custom} tactics, for example, via the {\tt Ltac} construct in Rocq.  We represent  both built-in  and custom tactics through the following formalization:

\begin{definition}[{\bf Tactic definition}]
A \emph{tactic definition}  is a quadruple $(\tname, \targs, \tret, \tbody )$ where $\tname$ is the name of the tactic, $\targs$ and $\tret$ are lists of identifiers representing formal tactic inputs and outputs respectively,  and $\tbody$ is an expression that manipulates a proof state.
\end{definition}


\begin{example}
Consider the following simple tactic in Rocq: \verb|Ltac MyTac h := intro h; simpl.|
This custom tactic first applies the built-in \texttt{intro} tactic to introduce a hypothesis \( h \) and then simplifies the resulting proof goal using \texttt{simpl}. In our representation, this tactic is formalized as:
\[
(\texttt{MyTac}, [g], [h, g''], \texttt{intros} \ [g] \ [h, g']; \texttt{simpl} \ [g'] \ [g''] )
\]
Here, the formal input \( I \) consists of a single goal \( g \), and the formal output is \([h, g'']\), indicating that the resulting proof state includes a new hypothesis \( h \) and a new goal \( g'' \). This explicit representation captures \emph{which parts} of the proof state are modified by a given tactic.
Similarly, when specifying tactic applications (as in the body of the custom tactic), we explicitly define their inputs and outputs. In this case, \texttt{intros} takes the initial goal \( g \), introduces a hypothesis \( h \), and produces an updated goal \( g' \), which is then further transformed by \texttt{simpl} into \( g'' \). This formalization clarifies the sequential modifications to the proof state.
\end{example}


In the remainder of this paper, we assume  uniqueness of tactic names, and, given a tactic name $\tname$, we write $\mathsf{In}(\tname)$ and $\mathsf{Out}(\tname)$ to denote its formal inputs and outputs, and  $\mathsf{Body}(\tname)$ to denote its body. 



\begin{table}[t]
\centering
\small
\vspace{-0.05in}
\caption{Our internal representation of some commonly used  Rocq tactics, where $h$ denotes a hypothesis, $g$ denotes a goal, and $p$ is any proposition (either goal or hypothesis).}
\vspace{-0.05in}
\begin{tabular}{|l|l|l|}
\hline
\textbf{Tactic name} & \textbf{Signature} & \textbf{Semantics} \\ \hline
\texttt{intro}            & $g \rightarrow h \times g'$           & If $g$ matches $p_1 \rightarrow p_2$, produces $p_1$ as hypothesis $h$ \\
& & and $p_2$ as  goal $g'$           \\ \hline
\texttt{apply .. (in) }             & $h \times p \rightarrow p_1 (\times ... \times p_n)$ 
& If $h$ matches $p \rightarrow p_1$ and $p$ is a hypothesis, produces $p_1$ \\ & & as a new hypothesis; if $h$ matches $p_1 \rightarrow \ ..\ \rightarrow p_n \rightarrow p$ \\ & & and $p$ is a goal, produces $p_1 \times \ ..\ \times p_n$ as new goals          \\ \hline
\texttt{exact}           & $h \times g \rightarrow \bot$            & If $h$ matches $g$, discharges the goal \\
& & (i.e., no new goal is produced)            \\ \hline
\texttt{split}           & $g \rightarrow g_1 \times g_2 $            & If $g$ is of the form $(p_1 \land p_2)$,  produces $p_1$ and $p_2$ \\
& & as two new goals $g_1$ and $g_2$, respectively       \\ \hline
\texttt{destruct}$_\lor$           &$h \times g \rightarrow $          & If $h$ matches $(h_1 \lor h_2)$, produces two goals $g_1$ and $g_2$, \\
& $h_1 \times h_2 \times g_1 \times g_2$  & derived from $g$, with corresponding hypotheses $h_1$ and $h_2$        \\ \hline
\texttt{rewrite}           &$h \times p \rightarrow p'$            & If $h$ matches $x = y$ and $p$ is of the form $p(x)$, \\ & &produces $p(y)$   as the new proposition $p'$ \\ \hline
\texttt{left} /   \texttt{right}          &$g \rightarrow g'$            & If $g$ matches $p_1 \lor p_2$, produces $p_1$ (resp. $p_2$ for \texttt{right}) \\
& &  as a new goal $g'$ \\ \hline
\end{tabular}
\label{tab:built-in}
\vspace{-0.13in}
\end{table}


\subsection{Tactic Applications}

A tactic application modifies the proof state by taking a set of actual proof elements as inputs and producing new proof elements as outputs. Formally, a tactic application is represented as a tuple $(\eta, X, Y)$ where $\eta$ is the name of the applied  tactic, $X$ is the list of \emph{actual inputs} (e.g., goals, hypotheses),  $Y$ is the list of \emph{actual outputs} (e.g., new hypotheses, transformed goals). In our representation, a tactic application is exactly akin to \emph{function invocation} with \emph{call-by-value} semantics:  Given a tactic definition $
(\eta, I, O, \tbody)$ 
where $I$ and $O$ denote the formal inputs and outputs, respectively, applying $\eta$ involves replacing formal inputs with actual proof elements and executing $\tbody$.   
More formally, each tactic invocation $\tinvoke = (\tname, X, Y)$ defines a transition relation between proof states, denoted  $\denot{\tinvoke}(\pstate) = \pstate'$, where
$
\pstate' = (\pstate \backslash \textsf{Goals}(X)) \ \uplus  [ Y \mapsto  \big (\denot{\tbody}([\mathsf{In}(\tname) \mapsto \pstate(X) ] \big )(\mathsf{Out}(\tname))]. 
$
In particular, $\denot{\tinvoke}(\pstate)$ modifies the proof state $\sigma$ according to body expression $\tbody$, subject to the usual formal-to-actual renamings. Note that inputs of $\tinvoke$ that correspond to proof goals are transformed by the tactic invocation; hence, identifiers that refer to ``stale'' proof goals are removed when constructing the new input state $\pstate'$. {Importantly, we assume tactics are deterministic with respect to the inputs (i.e., goals, hypotheses, and hint databases) they consume. That is, tactic invocations will always produce the same outputs as long as their inputs in the initial proof state are the same, even if the rest of the state is different.} Given proof state $\pstate$ and expression $\tbody$, we write $\denot{\tbody}(\pstate)$ to denote the resulting  state $\pstate'$ after applying $\tbody$ to~$\pstate$. 

\begin{example}
Consider the following tactic application $(\texttt{apply}, \texttt{[H0, g3]}, \texttt{[g4, g5]})$ (where the semantics of the built-in \texttt{apply} tactic is given in Table~\ref{tab:built-in}) and the proof state $\pstate$: 
\small
\[ [\texttt{h1} \mapsto \texttt{P1}, \ \ \texttt{h2} \mapsto \texttt{P2}, \ \ \texttt{H0} \mapsto (\texttt{P1} \rightarrow \texttt{P2} \rightarrow \texttt{P3}), \ \ \texttt{g3} \mapsto \texttt{P3}] \]
\normalsize
This tactic application results in the following modified proof state $\pstate'$:
\small
\[
[\texttt{h1} \mapsto \texttt{P1}, \ \ \texttt{h2} \mapsto \texttt{P2}, \ \ \texttt{H0} \mapsto (\texttt{P1} \rightarrow \texttt{P2} \rightarrow \texttt{P3}), \ \ \texttt{g4}\mapsto \texttt{P1}, \ \ \texttt{g5} \mapsto \texttt{P2}]
\] \normalsize {Observe that the hypotheses are still the same, but previous goal \texttt{g3} is removed from the old proof state and two new goals, \texttt{g4} and \texttt{g5}, are added to $\pstate'$.}
\end{example}

\subsection{Proof Scripts and Proofs}

We conclude this section by defining \emph{proof scripts}: programs defined by a sequence of tactic invocations.  A \emph{proof} corresponds to a successful \emph{execution} of that program on some initial  state.  

\begin{definition}[{\bf Proof script}]
A \emph{proof script} $\pscript$ is a sequence of tactic invocations.
\end{definition}

\begin{example}\label{ex:script}
{Consider the following lemma:}
\small
\[
\texttt{Lemma implication: (P1} \land \texttt{P2) -> (P1 -> P2 -> P3) -> P3.}
\] 
\normalsize
and its corresponding proof script:
{
\small
\begin{verbatim}
   intros H. destruct H as [h1 h2]. intros H0. apply H0.
   - exact h1.
   - exact h2.
\end{verbatim}
\normalsize
} 
In our representation, this proof script corresponds to the following sequence of tactic invocations:
{
\small
\begin{verbatim}
   intro [g0] [H, g1];  destruct [H, g1] [h1, h2, g2]; 
   intro [g2] [H0, g3]; apply [H0, g3] [g4, g5]; 
   exact [h1, g4] [];   exact [h2, g5] [].
\end{verbatim}
\normalsize
}

Although the original Rocq script uses bullets (``-'') to indicate a structured proof with multiple subgoals, our proof script representation explicitly encodes the logical dependencies between tactics as a linear sequence of invocations. In particular, the \texttt{apply H0} tactic introduces two subgoals, labeled \texttt{g4} and \texttt{g5}, which are subsequently solved by \texttt{exact h1} and \texttt{exact h2}, respectively.
\end{example}


Finally, we define a \emph{proof} as a successful execution of a proof script on some initial state:
\begin{definition}[{\bf Proof}]
Let $\pstate_0$ be a proof state and $\pscript = \tinvoke_1; \ldots; \tinvoke_n$ be a proof script.  Executing $\pi$ on $\pstate_0$ yields a sequence of states (i.e., \emph{trace}) $\pstate_1, \ldots \pstate_n$ where $\pstate_{i+1} = \denot{\tinvoke_i}(\pstate_i)$. We say that $\pi$ is a \emph{proof} of $\pstate_0$ if $\pstate_n$ does not contain any goal identifiers in its domain. 
\end{definition}

\begin{example}
Consider the initial proof state, which corresponds to our lemma to be proven from Example~\ref{ex:script}: 
$ \small \sigma_0: [\texttt{g0} \mapsto ((\texttt{P1} 
 \land \texttt{P2}) \rightarrow (\texttt{P1} \rightarrow \texttt{P2} \rightarrow \texttt{P3}) \rightarrow \texttt{P3}) ] $.
Executing the proof script from Example~\ref{ex:script} results in the following trace:
\[
\footnotesize
\begin{array}{llll}
\sigma_1: & [\texttt{H} \mapsto (\texttt{P1} \land \texttt{P2}), \ \ \texttt{g1} \mapsto ((\texttt{P1} \rightarrow \texttt{P2} \rightarrow \texttt{P3}) \rightarrow\texttt{P3}) ] \\
\sigma_2: & [\texttt{h1} \mapsto \texttt{P1}, \ \ \texttt{h2} \mapsto \texttt{P2}, \ \ \texttt{g2} \mapsto ((\texttt{P1} \rightarrow \texttt{P2} \rightarrow \texttt{P3}) \rightarrow\texttt{P3})] \\ 
\sigma_3: & [\texttt{h1} \mapsto \texttt{P1}, \ \ \texttt{h2} \mapsto \texttt{P2}, \ \ \texttt{H0} \mapsto (\texttt{P1} \rightarrow \texttt{P2} \rightarrow \texttt{P3}), \ \ \texttt{g3} \mapsto \texttt{P3} ] \\
\sigma_4: & [\texttt{h1} \mapsto \texttt{P1}, \ \ \texttt{h2} \mapsto \texttt{P2}, \ \ \texttt{H0} \mapsto (\texttt{P1} \rightarrow \texttt{P2} \rightarrow \texttt{P3}), \ \ \texttt{g4}\mapsto \texttt{P1}, \ \ \texttt{g5} \mapsto \texttt{P2}] \\
\sigma_5: &   [\texttt{h1} \mapsto \texttt{P1},  \ \ \texttt{h2} \mapsto \texttt{P2}, \ \ \texttt{H0} \mapsto ( \texttt{P1} \rightarrow \texttt{P2} \rightarrow \texttt{P3}), \ \ \texttt{g5} \mapsto \texttt{P2}] \\
\sigma_6: &   [\texttt{h1} \mapsto \texttt{P1},  \ \ \texttt{h2} \mapsto \texttt{P2}, \ \ \texttt{H0} \mapsto ( \texttt{P1} \rightarrow \texttt{P2} \rightarrow \texttt{P3}) ] \\
\end{array}
\]
Since $\sigma_6$ does not contain any goals, this constitutes a proof of  our lemma.
\end{example}

\section{Semantic Proof Refactoring}\label{sec:refactor}
{
To define our tactic discovery problem, we first consider: “What makes a custom tactic useful?” In our setting, a necessary requirement is that the tactic can be applied to refactor existing proofs in a way that yields a desirable outcome—be it reduced proof size, improved maintainability, or some other benefit. In this section, we formally introduce the concept of \emph{proof refactoring}. Instead of relying on brittle, {syntactic} notions, our approach takes a {semantic} perspective by representing tactic-style proofs with \emph{tactic dependency graphs (\tdg)}.
}


\begin{definition}[{\bf \tdg of proof script}]
Let $\pscript$ be a proof script. A \emph{tactic dependence graph} for $\pscript$, denoted $\mathsf{TDG}(\pscript)$, is a directed acyclic graph $G = (V, E)$ where each node $v(\tname) \in V$ corresponds to a tactic invocation $\tinvoke = (\tname, X, Y) \in \pscript$ and each arc $(s(\tname), t(\tname'), \beta, \alpha )$ indicates that the formal input $\alpha$ for tactic invocation $t(\tname')$ corresponds to the formal output $\beta$ for tactic invocation $s(\tname)$.
\end{definition}

\begin{example}
Figure~\ref{fig:tdg-ex} shows the TDG for the proof script from Example~\ref{ex:script}. Here, an edge labeled $(a, b)$ indicates the formal output $a$ of the source node corresponds to the formal input $b$ of the target node.  We use the  signatures shown in Table~\ref{tab:built-in} to refer to the formal inputs and outputs of Rocq's built-in tactics. For example, the edge from $\nu_3$ to $\nu_4$ labeled $(g', p)$ indicates that formal output $g'$ of the \texttt{intro} tactic becomes the formal input $p$ of the \texttt{apply} tactic.
\end{example}

\begin{wrapfigure}{r}{0.22\textwidth}
    \centering
    \vspace{-0.15in}  \includegraphics[width=0.22\textwidth]{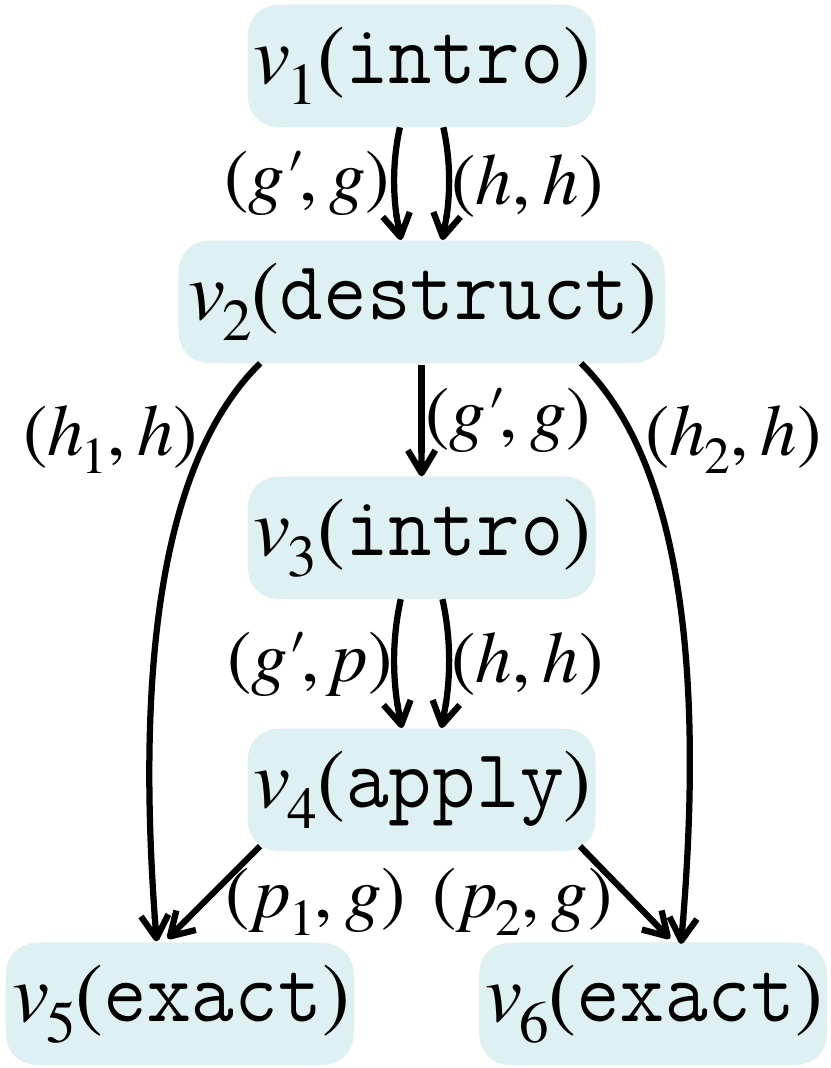}
    \vspace{-0.1in}
    \caption{TDG for Ex~\ref{ex:script}. 
    }
    \vspace{-0.1in}
    \label{fig:tdg-ex}
    \vspace{-0.1in}
\end{wrapfigure}

In general, multiple  syntactically different proof scripts can have the same \tdg. Given a graph $G$, we use $\mathsf{InducedProofs}(G)$ to denote the set of all proof scripts $\Pi = \{ \pscript_1, \ldots, \pscript_k \}$ such that $G$ is isomorphic to $\mathsf{TDG}(\pscript_j)$. Thus, we can view a \tdg as representing a combinatorially large class of syntactically different, but semantically equivalent, proof scripts. We can obtain all induced proofs of $G$ by considering all  topological sorts of $G$ that satisfy  branching constraints of the original proof {(see Section~\ref{sec:impl} for more details)}.

Note that we can also represent a tactic definition as a \tdg as long as its body expression can be expressed as a sequence of other tactic invocations.  However, to simplify the presentation of our algorithms, we augment the \tdg for each tactic definition with two special nodes to represent their formal inputs and outputs.  


\begin{definition}[{\bf Tactic \tdg}] 
Let $\tactic = (\tname, I, O, \pi) $ be a tactic definition where $\pi$ is a sequence of tactic invocations, and let $G_\pi = (V_\pi, E_\pi)$ be the \tdg of $\pi$. The \tdg of $\tactic$ is a directed acyclic graph $G = (V, E)$ such that  $V = V_\pi \cup \{ \inv, \outv  \}$ where $\inv, \outv$ represent the formal inputs and outputs of $\tactic$ respectively. Furthermore, let $I_{\alpha_i, \beta_j}$ denote the set of all tactic invocations of the form $(\tname', \_, \_)$  where the formal input $\alpha_i$ of $\tactic$ corresponds to formal input $\beta_j$ of $\tname'$, and let $O_{\alpha_k, \beta_l}$ denote the set of all tactic invocations of the form $(\tname', \_, \_)$  where the formal output  $\alpha_k$ of $\tau$ corresponds to formal output $\beta_l$ of $\tname'$. Then, $G$ contains edges:
\[
E = E_\pi \cup \bigcup_{ij} I_{\alpha_i, \beta_j} \cup \bigcup_{kl} O_{\alpha_k, \beta_l}
\]
\end{definition}

\begin{wrapfigure}{r}{0.31\textwidth}
 \vspace{-0.8in}
    \centering
    \includegraphics[width=0.31\textwidth]{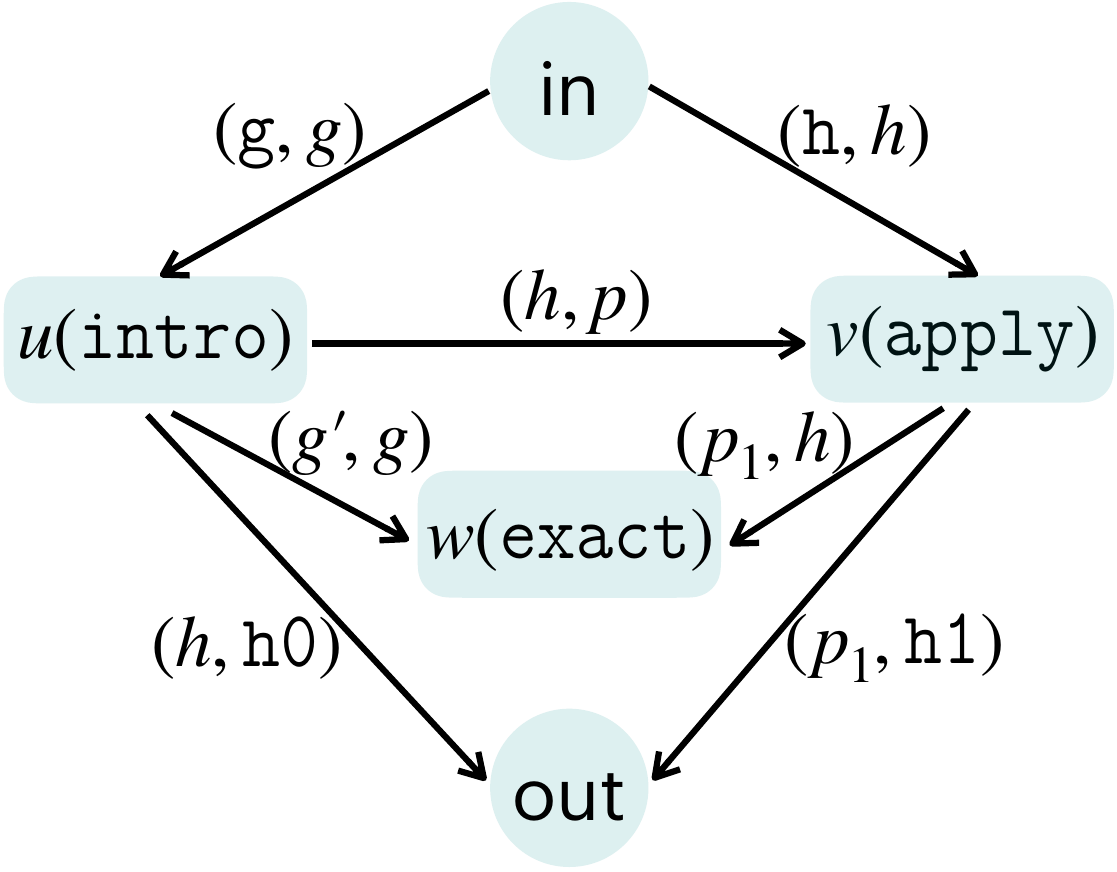}
    \vspace{-0.22in}
    \caption{Tactic TDG for Ex.~\ref{ex:tactic}.}
    \label{fig:tdg-tactic}
    \vspace{-0.15in}
\end{wrapfigure}
In other words, the TDG for a tactic contains \emph{special nodes and edges} that connect formal inputs and outputs in the tactic definition to the formal inputs and outputs of other tactic invocations.

\begin{example}\label{ex:tactic}
Consider the following tactic definition in Rocq: \\
\vspace{-0.1in}
{\small 
\begin{verbatim}
Ltac myTac h h0 h1:= intro h0; apply h0 in h as h1; exact h1. 
\end{verbatim}
}
\noindent In our internal representation, this tactic takes two formal arguments, namely goal $\texttt{g}$ and hypothesis \texttt{h}, and has two formal outputs \texttt{h0} and \texttt{h1}. The body expression of this tactic is represented using the following sequence of tactic invocations:
{\small
\begin{verbatim}
intro [g] [h0, g1]; apply [h, h0] [h1]; exact [h1, g1] []
\end{verbatim}
} \noindent
Figure~\ref{fig:tdg-tactic} shows the TDG representation of this tactic.  Note that there is an edge from node $u$ to the special exit node labeled $(h, {\tt h_0})$ since the formal output $h$ of the {\tt intro} tactic becomes the formal output {\tt h0} of the custom tactic. The edge between $\nu$ and the exit node is labeled $(p_1, {\tt h_1})$ similarly.
\end{example}

Intuitively, a proof script $\pscript$ could be refactored using a tactic $\tactic$ if $\tactic$ is a subgraph of $\pscript$ modulo isomorphism, subject to some extra restrictions. To make this more precise, we define \emph{isomorphic embedding} in the context of \tdg. 

\begin{definition}[{\bf Isomorphic embedding}]\label{def:isomorphic}
 A tactic \tdg $G=(V, E)$ is an isomorphic embedding into a proof  \tdg $G' = (V', E')$ iff there exists an injective function $f: V \backslash \{\inv, \outv\} \rightarrow V'$ such~that:
\begin{enumerate}[leftmargin=*]
\item Every vertex in $V \backslash \{ \inv, \outv \}$ is mapped to a vertex of $G'$ with the same tactic name, i.e., 
\[ \forall v \in V. \exists v' \in V'. f(v(\tname)) = v'(\tname)
\]
\item Every non-special edge in $G$ is mapped to a corresponding edge of $G'$ with the \emph{same edge label}:
\[
\forall (s, t, \beta, \alpha) \in E. \ (s \neq \inv \land t \neq \outv) \rightarrow  \exists (s', t', \beta, \alpha) \in E'.  f(s) = f(s') \land f(t) = f(t')
\]
\end{enumerate}
We refer to this injective function $f$ as the \emph{witness}. 
\end{definition}

\begin{wrapfigure}{r}{0.2\textwidth}
\vspace{-0.2in}
    \centering
    \includegraphics[scale=0.205]{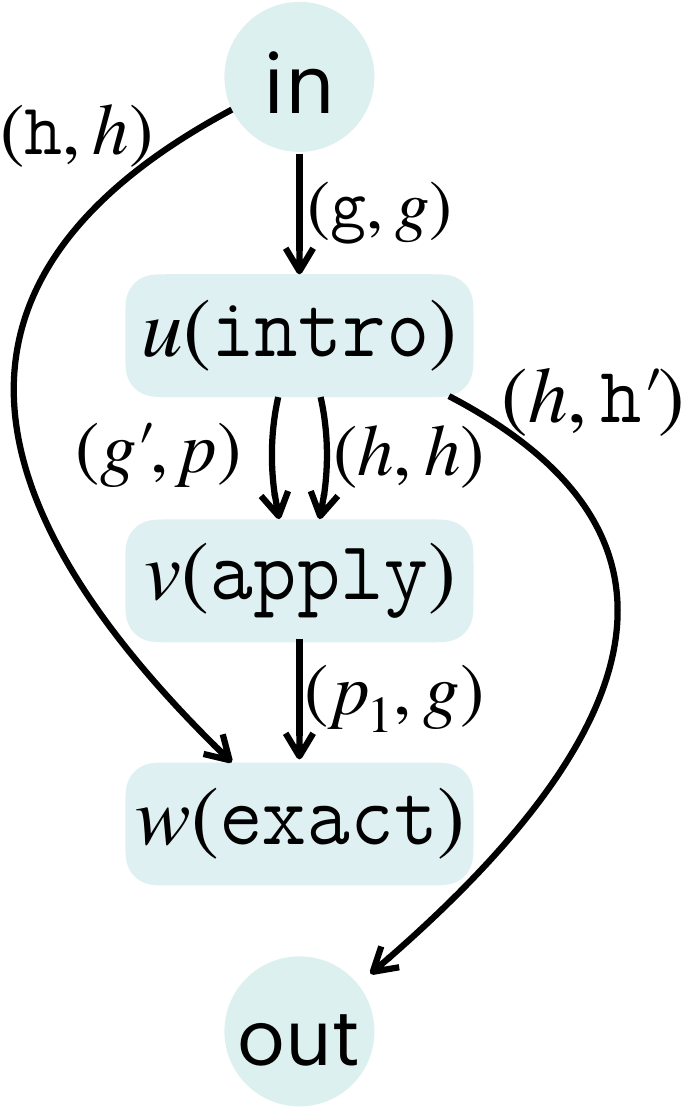}
    \vspace{-0.08in}
    \caption{TDG for Ex~\ref{ex:embedding}.}
    \vspace{-0.13in}
    \label{fig:tdg-tactic2}
\end{wrapfigure}
The two conditions in the above definition ensure agreement between labels of nodes and edges between the tactic TDG $G$ and the proof TDG $G'$ that $G$ is embedded into. 
However, it turns out that finding an isomorphic embedding between a tactic $\tactic$ and a proof script $\pscript$ is a necessary but not sufficient condition for refactoring $\pscript$ using $\tactic$. To see why this is the case, consider two TDGs, one with edges $\{{\tt a} \shortrightarrow {\tt b}, {\tt b} \shortrightarrow {\tt c}, {\tt a} \shortrightarrow {\tt c}\}$, and another with edges $\{{\tt a} \shortrightarrow {\tt d}, {\tt d} \shortrightarrow {\tt c}, {\tt a} \shortrightarrow {\tt c}\}$. 
While the subgraph consisting of nodes \(\texttt{a}\) and \(\texttt{c}\) appears in both proofs, it is not possible to extract a valid tactic that isolates these two steps. This is because node \(\texttt{c}\) depends on a node outside the tactic, which in turn relies on an intermediate result produced by the tactic, preventing it from being used as a standalone argument.
To avoid this issue, we must additionally ensure that the subgraph of $\pscript$ that is isomorphic to $\tactic$ is \emph{collapsible} into a single node, meaning that 
we can replace all incoming edges into the subgraph with the special entry node of the tactic and all outgoing edges with the special exit node.
We formalize this using the definition below:

\begin{definition}[{\bf Collapsible embedding}]\label{def:collapsible} Let $G = (V, E)$,   $G' = (V', E')$ be the \tdg's for a tactic definition  $\tactic$ and proof script $\pscript$ respectively such that $G$ is an isomorphic embedding into $G'$ with witness function $f$. Let $G \vdash u \leadsto v$ denote that node $v$ is reachable from node $u$ in graph $G$. We say that $G'$ is $f$-collapsible iff both of the following conditions hold:
\begin{enumerate}[leftmargin=*]
\item $\forall u, v \in \mathsf{Range}(f). \forall w \in V'. \  (G' \vdash u \leadsto w \ \land G' \vdash w \leadsto v) \Rightarrow w \in \mathsf{Range}(f)
$
\item $\forall u, v \in \mathsf{Dom}(f). \forall e \in E'. \  \ 
e= (f(u), f(v), \alpha, \beta)   \Rightarrow (u, v, \alpha, \beta) \in E $
\end{enumerate}
\end{definition}

Here, the first condition states that, if we have $f(a) = u $ and $f(b) = v$ and there is a path from $u$ to $v$ that includes $w$ in the middle, then $f$ must also map some node $c \in V$ to $w$. On the other hand, the second condition states that if $G'$ includes an edge between any pair of vertices that are in the range of $f$, then the corresponding edge must also exist in $G$.

To see why we require the first  property, note that any valid proof script induced by $G'$ must include  tactic invocations in the order $u, w, v$ since $w$ depends on an output of $u$ and $v$ depends on an output of $w$. Thus, if $f$ had only $u, v$ in its range but not $w$, the refactored proof could not allow the tactic invocations $u, v, w$ in the required order, thereby resulting in an invalid proof script.  Additionally, we  need the second property because a candidate tactic $G$ is not a valid embedding unless it includes \emph{all} required dependencies between a pair of tactic invocations. Intuitively, these conditions enforce that, if we replace the subgraph of $G'$ that is isomorphic to $G$ with a single node representing a new tactic invocation, then the resulting proof is still valid.

\begin{example}\label{ex:embedding}
Consider the proof script and its corresponding TDG from Figure~\ref{fig:tdg-ex}. We give examples and non-examples of collapsible isomorphic embeddings.

\begin{enumerate}[leftmargin=*]
    \item Consider the tactic from Example~\ref{ex:tactic} and its TDG in Figure~\ref{fig:tdg-tactic}. Note that both the proof script and the tactic contain the sequence of tactic invocations {\tt intro}, {\tt apply}, {\tt exact}. However, the function $ [ u \mapsto v_2, v \mapsto v_3, w \mapsto v_4 ]$ does not define an isomorphic embedding because it does not satisfy condition (2) from Definition~\ref{def:isomorphic}.
    \item Now consider the following tactic definition, whose TDG is shown in Figure~\ref{fig:tdg-tactic2}.
    {\small
    \begin{verbatim}
    Ltac myTac2 h0 h := intro h0. apply h0. exact h. 
\end{verbatim}
    }
    In our  representation, this tactic has formal inputs {\tt g} and {\tt h} and  output {\tt h'}. The TDG from Figure~\ref{fig:tdg-tactic2} is an isomorphic embedding into Figure~\ref{fig:tdg-ex} with the witness function $[u \mapsto v_3, v \mapsto v_4, w \mapsto v_5]$.
\end{enumerate}

\begin{minipage}{0.39\textwidth}
\begin{enumerate}[leftmargin=5pt]
    \setcounter{enumi}{2}
    \item Now, consider a subgraph $G_1$ of the TDG in Figure~\ref{fig:tdg-tactic2} that does not contain node $v$ as well as its incoming and outgoing edges. Then, the witness function $f$ from part (2) of this example is an isomorphic embedding but it is not collapsible, as it violates part (1) of Definition~\ref{def:collapsible}. 
    \item Next, consider a modified version $G_2$ of Figure~\ref{fig:tdg-tactic2} that does not contain the edge labeled $(h, h)$ between $u, v$. Then, the same witness function $f$ still defines an isomorphic embedding but it violates part (2) of Definition~\ref{def:collapsible}. Intuitively, $G_2$ is not a valid tactic modulo the original proof because it omits a \emph{required} dependency between {\tt intro} and {\tt apply}.
\end{enumerate}

\begin{adjustwidth}{-7.5pt}{0pt}  
Our proof refactoring procedure is presented in Algorithm~\ref{alg:refactor}: It takes as input a proof script $\pscript$ and a tactic definition $\tactic$, and returns a refactored proof script of $\pscript$.  The algorithm first constructs the \tdg's $G, G_c$ for  the tactic $\tactic$ and proof script $\pscript$ respectively  (line 2)  and then it repeatedly contracts $G_c$ in the while loop (lines 3--12) by finding a collapsible
\end{adjustwidth}
\end{minipage}
\hspace{6pt}
\begin{minipage}{0.55\textwidth}
\fbox{
    \begin{minipage}{0.95\textwidth}
    \small
    \vspace{-0.1in}
    \begin{algorithm}[H]
    \begin{algorithmic}[1]
    \Procedure{Refactor}{$\tactic, \pscript$}
     \Statex \Input{A tactic definition $\tactic$ and a proof script $\pscript$.}
    \Statex \Output{A refactored proof script of $\pscript'$}
    \State $G, G_c \assign {\sf ConstructTDG}(\tactic),  {\sf ConstructTDG}(\pscript)$
    \While{true}
    \State $f \assign {\sf FindEmbedding}(G, G_c)$
    \If{$f \equiv \bot$} {\bf break}
    \EndIf
    \State $V \assign \{v \ | \ v \in {\sf Nodes}(G_c) \wedge  v \in {\sf Range}(f)\}$
    \ForAll{$v \in V$}
    \ForAll {$u \in {\sf Parents}(v) \setminus {\sf Range}(f)$}
    \State $G_c \assign {\sf RewireIn}(u, v, G, G_c)$
    \EndFor
        \ForAll {$u \in {\sf Children}(v) \setminus {\sf Range}(f)$}
    \State $G_c \assign {\sf RewireOut}(u, v, G, G_c)$
    \EndFor
      \State $G_c \assign {\sf Contract}(G_c, V, v(\tau.\eta))$
    \EndFor
    \EndWhile
  
    \State \Return $\pscript' \in ({\sf InducedProofs}(G_c))$
    \EndProcedure
    \end{algorithmic}
    \end{algorithm}
    \vspace{-0.25in}
      \end{minipage}
    }
    \captionsetup{type=figure} 
    \vspace{-0.03in}
    \caption{Procedure for refactoring a proof script $\pscript$ using tactic  $\tactic$. This procedure uses auxiliary procedures ${\sf RewireIn}$ 
    (resp. ${\sf RewireOut}$) to change the labels of the incoming (resp. outgoing) edges of $v$. Given edge $(v', v, \alpha, \beta)$ in the \tdg $G_c$ of $\pscript$ with corresponding edge $(\inv, f^{-1}(v), \gamma, \beta)$ in the \tdg $G$ of $\tactic$,  \textsf{RewireIn} replaces that edge with  $(v', v, \alpha, \gamma)$.  Similarly, given edge $(v', v, \alpha, \beta)$ in $G_c$ with corresponding edge $(f^{-1}(v'), \outv, \alpha, \gamma)$ in $G$, \textsf{RewireOut} replaces that edge with  $(v', v, \gamma, \beta)$. Also, \textsf{FindEmbedding} finds a collapsible isomorphic embedding of $G$ into $G_c$  and \textsf{Contract}  performs standard graph contraction. }
    \label{alg:refactor}
\end{minipage}

\end{example}



{\noindent isomorphic embedding of $G$ into $G_c$ (line 4). If \textsf{FindEmbedding} does not return a witness (line 5),  further refactoring is not possible, so the algorithm returns any one of the induced proofs associated with $G_c$ (line 13). Otherwise, the algorithm proceeds in two steps. First, recall that the identifiers used in the tactic definition $\tactic$ are different from those used in the original tactic invocation, so we need to change the edge labels to ensure that the refactoring is correct. This is done in lines 8-11 using functions called \textsf{RewireIn} and \textsf{RewireOut}. In particular, given a node $v$ whose parent $u$ is not part of the embedding, we need to change the incoming edges to $v$ to use the correct identifier for $\tactic$ using the special entry edges in $G$ from $\inv$ to $f^{-1}(v)$. \textsf{RewireOut} does something very similar but for ``exit'' edges whose children are not part of the embedding. Finally, the second step (call to \textsf{Contract} at line 12) replaces the subgraph of $G_c$ induced by vertices $V$ with a single fresh node $v(\tactic.\tname)$ where $\tactic.\tname$ is the name of tactic $\tactic$. Since the replacement of a subgraph with a single node is the standard \emph{graph contraction} operation~\cite{diestel2024graph}, we do not provide the implementation of \textsf{Contract}.}
\begin{figure}[H]
\vspace{-0.15in}
    \centering
    \subfigure[Original proof TDG]{%
    \includegraphics[width=0.3\textwidth]{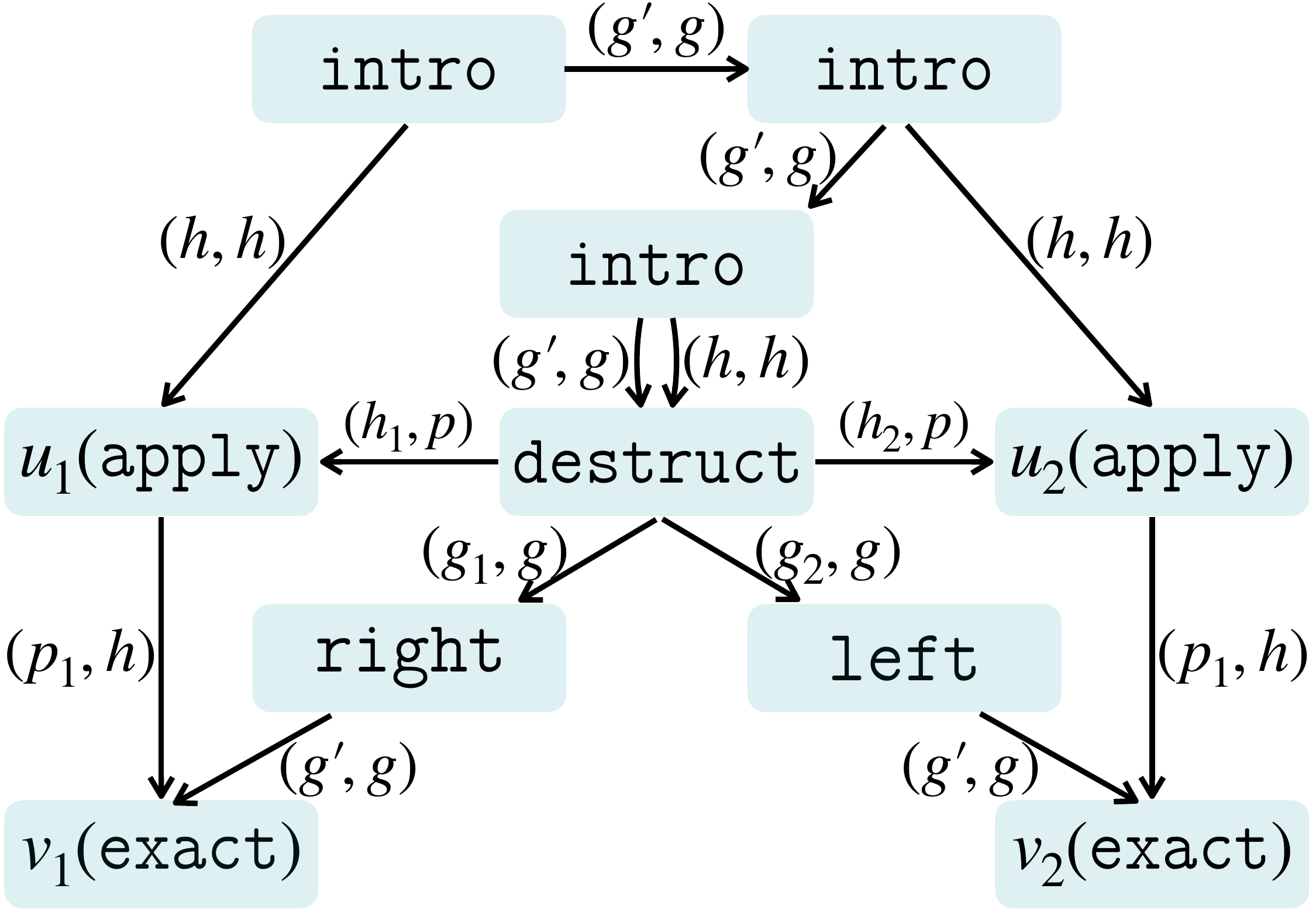}
        \label{fig:refactor-proof}
    }
\hfill
    \subfigure[A new tactic]{%
        \includegraphics[width=0.2\textwidth]{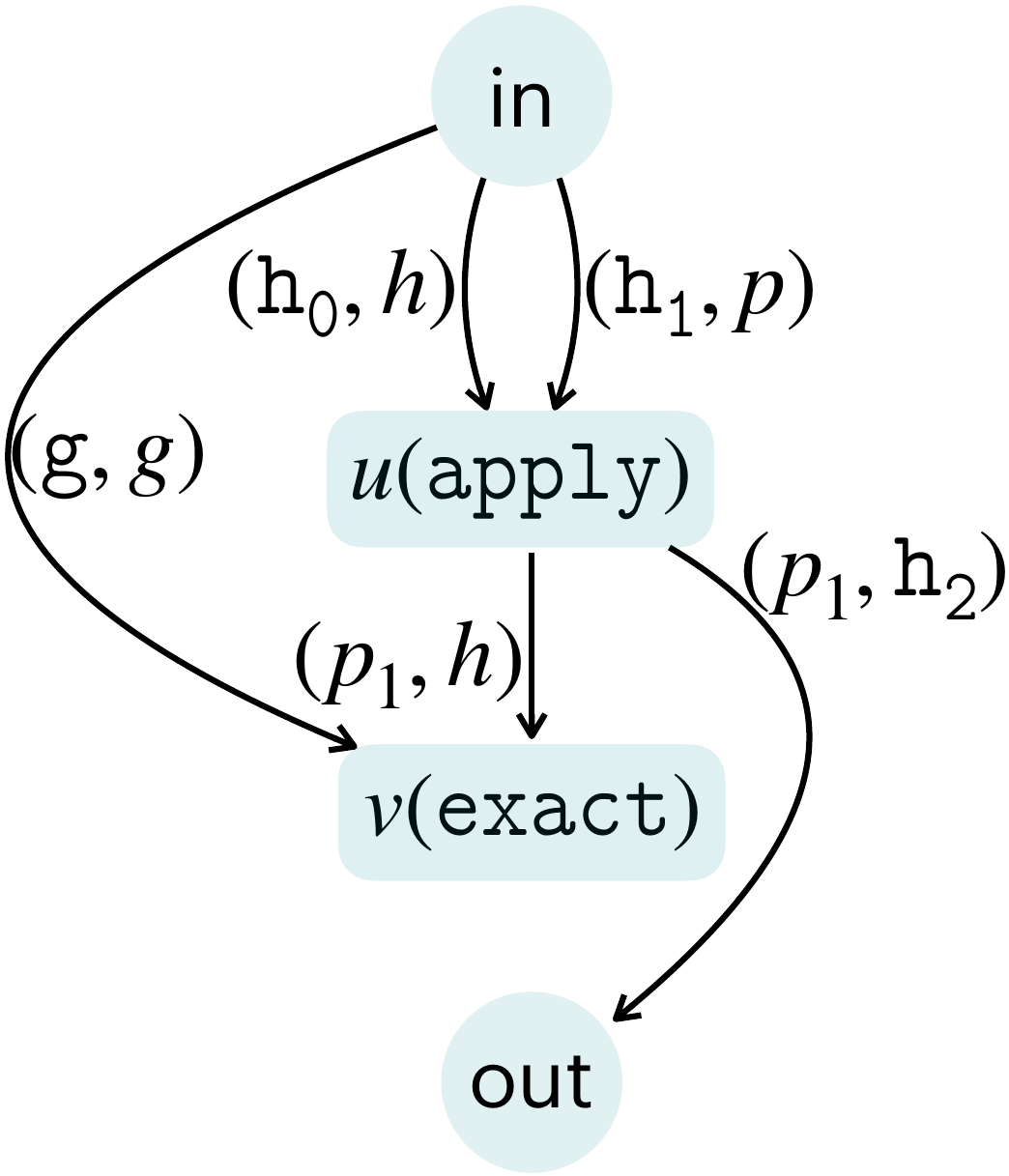}
        \label{fig:refactor-tactic}
    }\hfill
    \subfigure[Refactored proof]{%
     \includegraphics[width=0.45\textwidth]{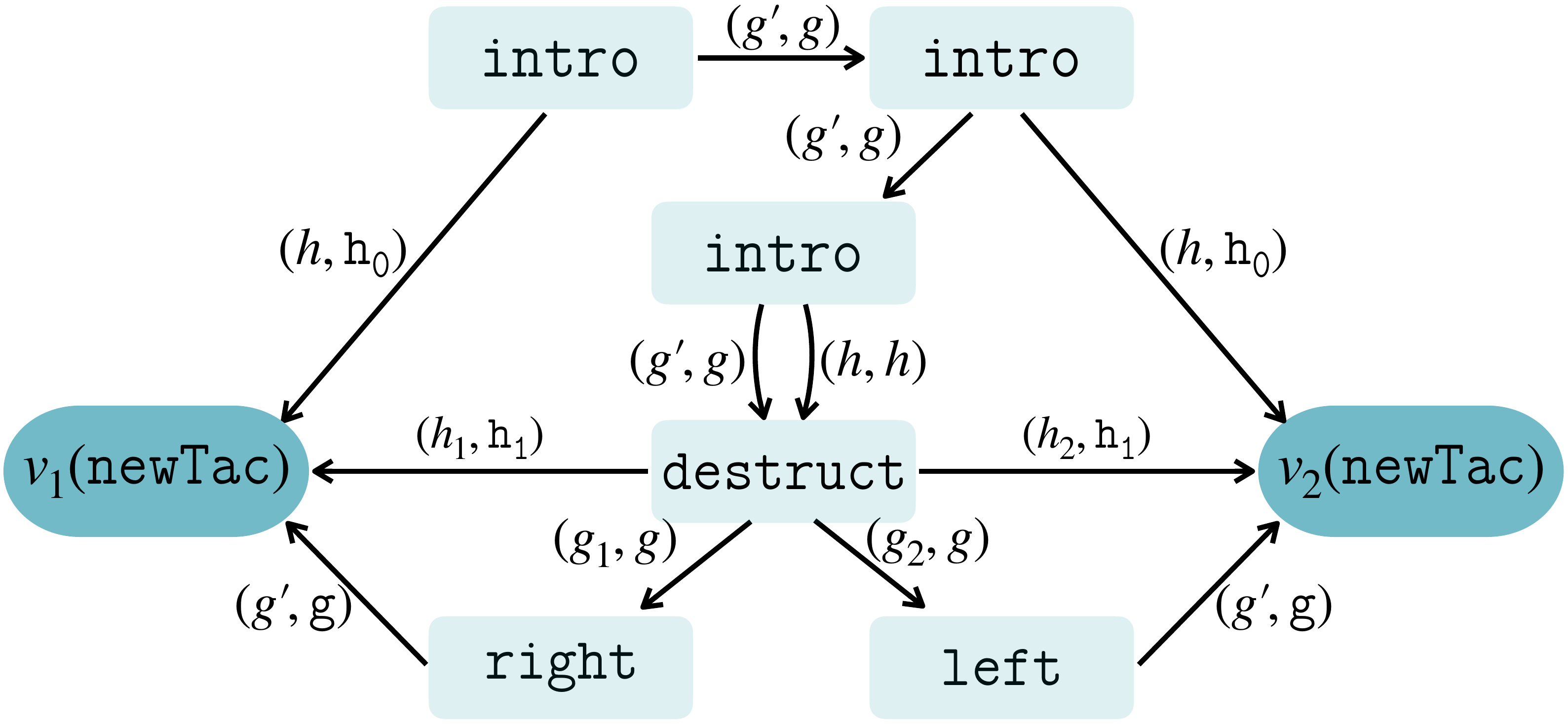}
        \label{fig:refactored}
    }
    \vspace{-0.05in}
    \caption{Refactoring a proof using a new tactic.}
    \label{fig:refactoring}
    \vspace{-0.1in}
\end{figure}

\vspace{0.03in}
\begin{example}
Consider the following Rocq proof  and its  TDG shown in Figure~\ref{fig:refactor-proof}:
\small
\begin{verbatim}
Lemma example:  (A -> D) -> (B -> C) -> ((A \/ B) -> (C \/ D)).
Proof. intro. intro. intro. destruct H1.
    - apply H in H1. right. exact H1.
    - apply H0 in H1. left. exact H1.
\end{verbatim}
\normalsize
Also, consider the following tactic whose TDG is shown in Figure~\ref{fig:refactor-tactic}:
\small
\begin{verbatim}
Ltac newTac h h' := apply h in h'. exact h'.
\end{verbatim}
\normalsize
This tactic can be embedded into the proof using the witness functions $f_1: [u \mapsto u_1, v \mapsto v_1] $ and $f_2: [u \mapsto u_2, v \mapsto v_2] $, and Figure~\ref{fig:refactored} shows the refactored TDG representing the following proof:
\small
\begin{verbatim}
Proof. intro. intro. intro. destruct H1.
    - right. newTac H H1.
    - left.  newTac H0 H1. 
\end{verbatim}
\normalsize
\end{example}

{\bf \emph{Remark.}} If there are multiple collapsible embeddings of a tactic $\tau$ into a proof $\pi$ such that these embeddings overlap, then the result of {\sc Refactor} may not be unique, depending on which embedding is discovered first. For simplicity, the rest of the paper assumes that collapsible embeddings do not overlap; however, our implementation handles this situation (see Section~\ref{sec:impl}).

\section{Tactic Library Synthesis Problem}\label{sec:problem}

In this section, we formalize the \emph{tactic discovery} problem addressed in the rest of this paper.


\begin{definition}[{\bf Tactic discovery problem}]
\label{def:prob-single-tactic}
Let \(\corpus\) be a corpus of proof scripts, and let
$\mathcal{O} \colon (\tactic \times \corpus) \to \mathbb{R}$
be an objective function that evaluates the quality of a tactic
\(\tactic\) on \(\corpus\).
The goal of the tactic discovery problem is to find a tactic \(\tactic\)
that maximizes $\mathcal{O}$ (i.e., 
$\arg \max_{\tactic} \, \mathcal{O}(\tactic, \corpus)$).
\end{definition}

{
While the optimization objective $\mathcal{O}$ could be defined in a number of ways, we mainly consider \emph{compression power} as our primary objective, even though our learning algorithm can also be adapted to  handle other types of measures, as long as $\mathcal{O}(\tactic, \corpus)$ can be defined as:
\small
\[
\mathcal{O}(\tactic, \corpus) \;=\;
f\bigl(\corpus,\;\textsc{Refactor}(\corpus,\tactic)\bigr),
\]
\normalsize
where $f$ is some function $\corpus \times \corpus' \rightarrow \mathbb{R}$ and {\sc Refactor} is the refactoring operation defined in the previous section. 
In particular, prior work on library learning in the program synthesis literature~\cite{babble,dreamcoder,stitch} has argued that reducing the overall size of a corpus---i.e., maximizing how much a learned abstraction ``compresses'' existing code---is an effective proxy for identifying broadly applicable patterns. Hence, we define an adaptation of compression power to the ITP setting. 
}


\begin{definition}[{\bf Compression power}]\label{def:cp} Given proof corpus and tactic $\tactic$, let
$
\corpus' = \{ \pscript' \ | \ \pscript \in \corpus \land \pscript' = \textsc{Refactor}(\pscript, \tactic) \}
$.
Then, the compression power of $\tactic$ modulo $\corpus$, denoted $\cp(\tactic, \corpus)$, is:
\small
\[
\cp(\tactic, \corpus) = \left  ( \frac{\sum_{\pscript \in \corpus} \mathsf{Size}(\pscript) }{\sum_{\pscript' \in \corpus'} \mathsf{Size}(\pscript')} \right )
\]
\normalsize
\end{definition}
Intuitively, 
the larger the compression power, the more effective the tactic is in reducing the corpus size. 
Finally, we can define the tactic library discovery problem as follows:

\begin{definition}[{\bf Library synthesis problem}]\label{def:lib-synth}
Let $\Pi_1$ be a corpus of proof scripts. The \emph{tactic library synthesis}  problem is to find a sequence of tactics $\tactic_1, \ldots, \tactic_n$  such that each $\tactic_i$ is a solution to the tactic discovery problem for corpus $\corpus_i = \textsc{Refactor}(\tactic_{i-1}, \corpus_{i-1})$. 
\end{definition}

{\bf \emph{Remark.}}  One might consider defining the library synthesis problem as finding a set of tactics that \emph{collectively} achieve the maximum compression power for the proof corpus. However, the compression power of a set of tactics depends on the \emph{order} in which those tactics are applied. This makes such a definition either ill-formed or necessitates evaluating all possible permutations of tactic application orders, which is computationally infeasible and, in our small-scale experiments, yields no meaningful gain in compression.  Consequently, we mirror prior library-learning works \cite{stitch, peano} and formulate the library synthesis problem as finding a sequence of tactics that achieve maximal compression at each step.

\section{Learning Tactic Libraries}\label{sec:learning}

We now describe our learning technique for discovering useful tactics from a given corpus. 

\subsection{Preliminary Definitions}
We start this section by presenting some definitions that are useful for describing our algorithm. 

\begin{definition}[\bf Witness set]\label{def:ws}
    Let $G, G'$ be the \tdg's of a tactic and proof script respectively. A witness set of $G$ and $G'$, denoted $\ws(G, G')$,  is the set of all witness functions proving that $G$ is an isomorphic embedding into $G'$.
\end{definition}

In this definition, we deliberately do not require a \emph{collapsible} isomorphic embedding; specifically, \( f \) only needs to satisfy Definition~\ref{def:isomorphic}. As we will demonstrate later in this section, our algorithm incrementally constructs the witness set and subsequently filters out non-collapsible candidates. This approach is necessary because the collapsibility criterion cannot be enforced incrementally.

\begin{definition}[\bf Embedding vector]\label{def:ev}
    An embedding vector $\ev$ for a tactic $\tactic$ and a proof corpus $\corpus$ is a mapping from each proof $\proofsym \in \corpus$  to the witness set $\ws(\tdg(\tactic), \tdg(\proofsym))$.
\end{definition}

That is, an embedding vector for  $\tactic$ maps each proof $\proofsym$ to the witness set between $\tactic$~and~$\proofsym$.

\begin{definition}[\bf Tactic candidate]
    A tactic candidate $\tc$ for a proof corpus $\corpus$ is a pair $(G, \ev)$ where $G$ is a TDG and $\ev$ is an embedding vector of $G$ for $\corpus$. 
\end{definition}

Given a tactic candidate $\tc = (G, \ev)$, we  write $\tc.G$ to denote $G$ and $\tc.\ev$ to denote $\ev$. Intuitively,  $G$ represents a  \tdg that \emph{could} be used to refactor some proofs in the corpus, and its embedding vector $\ev$ allows us to efficiently compute witnesses for \tdg's that are extensions of $G$. As we will see in the next section, our tactic discovery algorithm uses the tactic candidate data structure as a key building block when searching for valid tactics. 

\begin{definition}[\bf Frequency] Let $\tc = (G, \ev)$ be a tactic candidate and $\corpus$ a proof corpus. The \emph{frequency} of $\tc$ in $\corpus$, denoted $\freq(\tc, \corpus)$, is defined as follows:
\small
\[
\freq(\tc, \corpus) = \sum\limits_{\proofsym \in \corpus} \sum\limits_{f \in \ws(G, \tdg(\proofsym))} \mathmybb{1}[{\sf IsCollapsible}(f, G, \tdg(\pi))]
\]
\normalsize
where $\mathmybb{1}[\cdot]$ is the standard indicator function and $\mathsf{IsCollapsible}(f, G, G')$ evaluates to true iff $f$ defines  a collapsible isomorphic embedding of $G'$ into $G$.
\end{definition}

Intuitively, the frequency of a tactic candidate $\tc$ counts the number of times that $\tc$ can be used in refactoring the proofs in the corpus. 
The higher the frequency of a tactic candidate, the more frequently it can be applied when refactoring proofs in the corpus. However, to evaluate how useful a tactic is in compressing the corpus size, we also need to take into account the size of the tactic. To this end, we define the \emph{effectiveness} of a tactic candidate as follows:
\vspace{14pt}
\begin{wrapfigure}{r}{0.63\textwidth}
\vspace{-0.1in}
\fbox{
    \begin{minipage}{0.6\textwidth}
    \small
    \begin{algorithm}[H]
    \begin{algorithmic}[1]
    \vspace{-0.1in}
    \Procedure{LearnTactic}{$\corpus$}
    \Statex \Input{A proof corpus $\corpus$}
    \Statex \Output{A tactic $\tactic$ that achieves maximum compression of $\corpus$}
    \State $\mathcal{G} \assign {\sf ConstructTDGs}(\corpus)$
    \State $R \assign \textsc{LearnGraphGrammar}(\mathcal{G})$
    \State $\worklist \assign \textsc{InitWorklist}(R, \corpus)$
    \State $r \assign \{{\sf Candidate}=\bot, {\sf Eff}=0 \}$
    \While{$\worklist \neq \emptyset$}
    \State $\tc \assign \worklist.{\sf dequeue}()$
    \If{ $\textsc{UpperBound}(\tc, \corpus) < r.{\sf Eff}$} {\bf continue}
    \EndIf
    \If{$\strength(\tc, \corpus) >  r.{\sf Eff}$}
    \State $r \assign \{{\sf Candidate}=\tc, {\sf Eff}=\strength(\tc, \corpus)\}$
    \EndIf
    \State $\worklist.{\sf enqueue}(\textsc{Expand}(\tc, R, \corpus))$
    \EndWhile
    \State \Return ${\sf MakeTactic}(r.{\sf Candidate})$
    \EndProcedure
    \vspace{-0.05in}
    \end{algorithmic}
    \end{algorithm}
    \vspace{-0.23in}
    \end{minipage}
}
\vspace{-0.07in}
\caption{Procedure for learning a tactic that results in maximum compression of proof corpus $\corpus$. Procedure names that are in {\sc SmallCaps} font are defined in separate algorithms and explained in the rest of this section. On the other hand, procedure names that are written in {\sf Sans Serif} font are only explained in text. At a high level, this algorithm iteratively explores tactic candidates, pruning those whose expansions cannot yield a higher compression power than a previously encountered tactic.}
\vspace{-0.1in}
\label{alg:learn_tactic}
\end{wrapfigure}

\vspace{-25pt}
\begin{definition}[\bf Effectiveness]\label{def:strength}
Let $\tc = (G, \ev)$ be a tactic candidate and $\corpus$ a proof corpus. The \emph{effectiveness} of $\tc$ in $\corpus$, denoted $\strength(\tc, \corpus)$, is defined as 
$
\strength(\tc, \corpus) = (\mathsf{Size}(G) - 1) \times \freq(\tc, \corpus)
$.
\end{definition}

In other words, effectiveness takes into account both the frequency and the size of the tactic. \\
{\noindent In this definition, we subtract $1$ from the size of $G$, because when an embedding of the tactic is used}\\
{\noindent for contracting the proof, the size of the proof shrinks by $\mathsf{Size}(G)-1$. Intuitively,  $\strength(\tc, \corpus)$ can be used to accurately characterize the compression power of a tactic -- the higher the value 
 of $\strength(\tc, \corpus)$, the more effective $\tc$ is for compressing the proof corpus $\corpus$. }

As we will see in the next section, our algorithm uses this metric as a pruning criterion during search.\changed{\footnote{To adapt our learning algorithm in Section~\ref{sec:learn-tactic} to optimization objectives other than compression power, one needs to define a suitable instantiation of function $\strength$ for the corresponding objective.}}

\subsection{Tactic Discovery Algorithm}\label{sec:learn-tactic}
In this section, we describe our algorithm for learning a single tactic that {maximizes the desired objective}. Our top-level algorithm, called {\sc LearnTactic}, is presented in Figure~\ref{alg:learn_tactic}: It takes in a proof corpus $\corpus$ and returns a single tactic $\tactic$ that maximizes {$\textsc{CP}(\tactic, \corpus)$}. The algorithm starts by constructing \tdg's for each proof in the corpus and then calls the {\sc LearnGraphGrammar} procedure (described later), for learning a graph grammar $R$ that can be used to construct \tdg's of tactic candidates.  Intuitively, $R$  encapsulates recurring patterns within the proof corpus and provides a set of rules to guide the exploration of potential tactic candidates.

Lines 4-11 of {\sc LearnTactic} systematically explore different tactic candidates. Specifically, the algorithm initializes the worklist $\worklist$ with all the single-node tactics from the grammar $R$ and initializes a record $r$ to track the best tactic discovered so far.  In each iteration of the loop (lines 6--11), the algorithm dequeues an existing tactic candidate $\tc$ and decides whether or not to continue expanding it. In particular, if it can prove that \emph{no expansion} of $\tc$ will result in a higher compression power than the best discovered tactic, it discards $\tc$ from the search space. Otherwise, it proceeds to compute the actual effectiveness of $\tc$ using the function $\strength(\tc, \corpus)$ from Definition~\ref{def:strength}. If $\strength(\tc, \corpus)$ exceeds that of the previous best tactic, the result $r$ is updated to $\tc$. Additionally, since expansions of $\tc$ might have even higher compression power, the algorithm invokes the {\sc Expand} procedure to obtain other candidate tactics that can be produced by expanding $\tc$ using the productions in $R$. Upon termination, $r$ contains a tactic that maximizes compression power for the entire corpus. In the remainder of this section, we elaborate on the auxiliary procedures used in {\sc LearnTactic}.

\vspace{0.05in}
\noindent
{\bf \emph{Learning Graph Grammar.}} While our tactic learning algorithm is inspired by top-down enumerative program synthesis~\cite{lambda2,myth}, a key difference is that we do not have a context-free grammar defining the space of possible TDGs. Hence, our algorithm first learns a graph grammar that can be used to construct \tdg's of possible tactics by analyzing the proof corpus. 

\begin{figure}[H]
\vspace{-0.03in}
\fbox{
    \small
    \begin{mathpar}\resizebox{1\textwidth}{!}{    
    \inferrule{\quad \quad G \in \mathcal{G}  \quad v(\eta) \in {\sf Nodes}(G)  \quad v'(\eta') \in {\sf Nodes}(G) \quad \quad
    \theta = \{(\alpha, \beta) \ | \ (v(\eta), v'(\eta'), \alpha, \beta) \in {\sf Edges}(G) \} \} \quad \quad}{\mathcal{G} \vdash \eta \shortrightarrow (\eta', \theta)}} 
    \end{mathpar}
    }
    \vspace{-0.05in}
    \caption{Inference rule defining the \textsc{LearnGraphGrammar} procedure}
    \label{alg:learn_grammar}
    \vspace{-0.05in}
\end{figure}
\begin{definition}[\bf \tdg grammar]  
A \tdg grammar is defined by a set of production rules of the form $\eta \shortrightarrow (\eta', \theta )$, where each rule specifies that a single-node graph \(G = (\{v(\eta)\})\) can be transformed into a new graph \(G' = (\{v(\eta), v(\eta')\}, \{(v(\eta), v(\eta'), \alpha, \beta) \ | \ (\alpha, \beta) \in \theta \})\). 
\end{definition}

Intuitively, each production $\eta \rightarrow (\eta', \theta)$ states that a node labeled \(\eta\) in a \tdg can be connected to other nodes labeled  \(\eta'\) via arcs whose labels are specified  by $\theta$.  These productions are mined from the proof corpus based on the following observations:
\begin{enumerate}[leftmargin=*]
\item  If a certain tactic never occurs in the proof corpus, it also cannot appear in any learned tactic.
\item If there is no dependency between a pair of tactics in the corpus, there is no point in learning a tactic that involves a spurious dependency between them.
\item If there exists a dependency between a pair of tactics in a proof $\proofsym$ in the  corpus, then the tactic \tdg should either not include that dependency or, else, it should include \emph{all} the input-output dependencies --- otherwise, we cannot find a \emph{collapsible} isomorphic embedding.
\end{enumerate}

Based on these ideas, Figure~\ref{alg:learn_grammar} presents the graph grammar learning procedure as a single inference rule deriving a judgment of the form:
$
\mathcal{G} \vdash \eta \rightarrow (\eta', \theta)
$
The productions $R$ used in the {\sc LearnTactic} algorithm include the set of all rewrite rules $r$ such that $\mathcal{G} \vdash r $ according to Figure~\ref{alg:learn_grammar}.

\vspace{0.05in}
\noindent
{\bf\emph{Worklist Initialization.}} Next, we consider the {\sc InitWorklist} procedure presented in Figure~\ref{alg:init_worklist}. The idea behind this procedure is very simple: It constructs a set of single node \tdg's based on the non-terminals in  graph grammar $R$. Then, for each such single node \tdg, it computes the tactic's embedding vector by going over each proof in the corpus. Hence, the worklist is initialized to tactic candidates that consist of single node \tdg's and their corresponding embedding vector. 

\vspace{0.1in}
\noindent 
{\bf \emph{Generating New Tactic Candidates.}}  Recall that the {\sc LearnTactic} procedure generates new tactic candidates by calling the {\sc Expand} procedure, which is presented in Figure~\ref{alg:expand}. This algorithm takes as input an existing tactic candidate $\tc = (G, \ev)$ and the learned graph grammar $R$ and produces a new set $\Theta$ of tactic candidates, where each $\tc_i \in \Theta$ is an expansion of $\tc$. Specifically, let $\tc_i = (G_i, \ev_i)$ be one of the new tactic candidates produced by {\sc Expand}. Here, every $G_i$ is a strict supergraph of $G$ and always contains  additional edges that are not in $G$. Additionally, some of these $G_i$'s may also contain a \emph{single} additional node that is not in $G$.

As shown in Figure~\ref{alg:expand}, the {\sc Expand} procedure considers each production $\eta \rightarrow (\eta', \theta) \in R$ (line 3) and locates all  nodes $v \in G$ with $\eta$. Then, at line 6, it calls {\sc Apply} (discussed later) to produce a new tactic candidate $G'$  that includes a fresh node $v'$ labeled $\eta'$, along with arcs  labeled $(\alpha, \beta) \in \theta$ between $v$ and $v'$. Additionally, since there may be existing nodes labeled $\eta'$ in  $G$, the inner loop at lines 7--8 also adds additional edges from $v$ to each existing node $v_i$ labeled $\eta'$.  

We now also briefly explain the auxiliary {\sc Apply} procedure defined in Algorithm~\ref{alg:expand-production}. Given a tactic candidate $\tc = (G, \ev)$ where $G$ contains a node $v$, {\sc Apply} first constructs a new \tdg $G'$ that includes a (possibly new) node $v'$ as well  as edges between $v$ and $v'$ with labels $\theta$.  However, since a tactic candidate also contains the embedding vector for the tactic, the loop in lines 7--10 constructs a new embedding vector for $G'$ by extending the existing witness functions in $\ev$.
Finally, if the frequency of the resulting tactic candidate is less than two, this means that the new tactic candidate (or any of its future expansions) are \emph{not} useful for compressing the proof corpus; hence, {\sc Apply} returns the new tactic candidate $\tc'$ only if $\freq(\tc', \corpus)$ is at least 2. 

\begin{figure}[t]
\hspace{-0.17in}
\small
    \centering
    \vspace{-0.04in}
    \begin{minipage}{0.46\textwidth}
        \centering
        \vspace{-0.03in}
        \fbox{
        \begin{minipage}{\textwidth}
        \small
        \begin{algorithm}[H]
        \begin{algorithmic}[1]
        \vspace{-0.06in}
            \Procedure{InitWorklist}{$R, \corpus$}
                \Statex \Input{Graph grammar $R$, proof corpus $\corpus$}
                \Statex \Output{A  worklist of tactic candidates $\worklist$}
                \State $S = \{\eta \ | \ (\eta \shortrightarrow (\eta', E)) \in R \}$; $\worklist \assign \emptyset$ 
                \ForAll{$\eta \in S$}
                    \State $G = (\{v(\eta)\}, \varnothing)$
                    \State $\ev \assign [ \pi \mapsto \bot \ | \ \pi \in \corpus ] $
                    \ForAll{$\proofsym \in \corpus$}
                       \State $\ev[\proofsym] \assign {\sf GetWitness}(G, \tdg(\proofsym))$
                    \EndFor
                    \State $\worklist.{\sf enqueue}((G, \ev))$
                \EndFor
                \State \Return $\worklist$
            \EndProcedure
        \end{algorithmic}
        \end{algorithm}
        \end{minipage}
        }
    \vspace{-0.06in}
    \caption{The \textsc{InitWorklist} procedure for constru-\\cting the initial worklist used in {\sc LearnTactic}. }
    \label{alg:init_worklist}
    \end{minipage}
    \hspace{0.2in}
    \begin{minipage}{0.46\textwidth}
        \centering
        \vspace{-0.03in}
        \fbox{
        \begin{minipage}{\textwidth}
        \small
        \begin{algorithm}[H]
        \begin{algorithmic}[1]
        \vspace{-0.1in}
            \Procedure{Expand}{$\tc, R, \corpus$}
                \Statex \Input{A tactic candidate $\tc = (G, \ev)$, graph grammar $R$, corpus $\corpus$}
                \Statex \Output{A new set of tactic candidates}
                \State $\Theta \assign \emptyset$
                \ForAll{$\eta \shortrightarrow (\eta', \theta) \in R$}
                    \ForAll{$ v \in {\sf FindNodes}(G, \eta)$}
                        \State $v' \gets {\sf Fresh}(\eta')$
                        \State $\Theta \assign \Theta \cup \textsc{Apply}(\tc, v, v', \theta, \corpus)$
                        \ForAll{$v_i \in {\sf FindNodes}(G, \eta')$}
                            \State $\Theta \assign \Theta \cup {\textsc{Apply}}(\tc, v, v_i, \theta, \corpus)$
                        \EndFor
                    \EndFor
                \EndFor
                \State \Return $\Theta$
            \EndProcedure
        \end{algorithmic}
        \end{algorithm}
        \vspace{-0.25in}
        \end{minipage}
        }    
    \vspace{-0.05in}
    \caption{Procedure for expanding a given tactic candidate $\tc$ using graph grammar $R$. }
    \label{alg:expand}
    \end{minipage}
    \vspace{-0.15in}
\end{figure}

\begin{figure}[t]
\hspace{-0.17in}
\small
    \centering
    \begin{minipage}{0.46\textwidth}
    \vspace{-0.02in}
        \centering
        \fbox{ 
        \begin{minipage}{\textwidth}
        \vspace{0.01in}
        \small
        \begin{algorithm}[H]
        \begin{algorithmic}[1]
        \vspace{-0.08in}
  \Procedure{Apply}{$\tc, v, v', \theta, \corpus$}
            \Statex \Input{A tactic candidate $\tc = (G = (V, E), \ev)$, nodes $v, v'$,  edge labels $\theta$, and proof corpus $\corpus$}
            \Statex \Output{Empty set or singleton tactic candidate}
            \State $V' \gets (V \cup \{v'\})$
              \State $E' \gets \{E \cup \{(v, v', \alpha, \beta) \ | \ (\alpha, \beta) \in \theta \})$
            \State $G' \assign (V', E')$
            \State $\ev' \assign [ \proofsym \mapsto \varnothing \ | \ \proofsym \in \mathsf{Dom}(\ev) ];$
            \State $  \tc' \assign (G', \ev')$ 
            \ForAll{$\proofsym \in {\sf Dom}(\ev)$}
                \ForAll{$f \in \ev[\proofsym]$}
                \State $F \gets {\sf Extend}(f, G', \tdg(\proofsym))$
                \State $\ev'[\proofsym] \assign \ev'[\proofsym] \cup F$
                \EndFor
            \EndFor
            \If{$\freq(\tc', \corpus) \geq  2$}  \Return $\{\tc'\}$
            \Else  \ \Return $\varnothing$
            \EndIf
            \EndProcedure
        \end{algorithmic}
        \end{algorithm}
        \vspace{-0.15in}
        \end{minipage}
        }
    \vspace{-0.02in}
    \caption{Procedure for applying a production from the graph grammar on nodes $v, v'$. }
        \label{alg:expand-production}
    \end{minipage}
    \hspace{0.2in}
    \begin{minipage}{0.46\textwidth}
        \centering
        \fbox{
        \begin{minipage}{1\textwidth}
        \small
        \begin{algorithm}[H]
        \begin{algorithmic}[1]
        \vspace{-0.1in}
            \Procedure{UpperBound}{$\tc$, $\corpus$}
            \Statex \Input{A tactic candidate $\tc = (G, \ev)$ and a proof corpus $\corpus$}
            \Statex \Output{Upper bound on the effectiveness of $\tc$}
            \State ${\sf ub} \gets 0$
            \ForAll{$\proofsym \in {\sf Dom}(\ev)$}
            \ForAll{$f \in \ev[\pi]$}
            \State $G' \assign {\sf ApplyWitness}(G, f)$
            \State $G_e \assign {\sf MaxExtend}(G', \tdg(\proofsym))$
            \State ${\sf ub}  \gets {\sf ub} + {\sf size}(G_e) - 1$
            \EndFor
            \EndFor
            \State \Return {\sf ub}
            \EndProcedure
        \end{algorithmic}
        \end{algorithm}
        \vspace{-0.22in}
        \end{minipage}
        }
    \vspace{-0.05in}
    \caption{\small Computes upper bound on the compression power of any extension of $\tc$. \textsf{ApplyWitness} applies function $f$ to  $G$ to obtain a graph $G'$. \textsf{MaxExtend}$(G, G')$ finds a graph $G_e$ such that (1) $G_e$ is subgraph of $TDG(\pi)$ and supergraph of $G'$, (2) \textsf{root}$(G_e) = $ \textsf{root}$(G')$ and (3) $G_e$ has maximum size among all graphs that satisfy (1) and (2). }
        \label{alg:ub}
    \end{minipage}
    \vspace{-0.15in}
\end{figure}

\vspace{0.05in}
\noindent
{\bf \emph{Pruning Non-Optimal Tactic Candidates.}} Finally, we consider the {\sc UpperBound} procedure, presented in Figure~\ref{alg:ub}, for deriving an upper bound on the effectiveness of a given tactic candidate $\tc$ on compressing the size of the proof corpus $\corpus$.  The idea behind {\sc UpperBound} is fairly straightforward: For every embedding of $G$ in a proof script $\proofsym$, we compute the \emph{maximum extension} of $G$ that can still be embedded in $\proofsym$:

\begin{definition}\label{def:max-ext}{\bf (Maximum extension)}  The maximum extension of $G$ in $G'$ is a graph $G_e$ with the following properties: (1) $G$ is a subgraph of $G_e$; (2) $G_e$ is a subgraph of $G'$; (3) \textsf{root}$(G_e) = $ \textsf{root}$(G)$; and (4) For any other graph $G_e'$ satisfying (1) (2) and (3),  we have $\mathsf{size}(G_e') \leq \mathsf{size}(G_e)$.
\end{definition}

The idea is that, by summing up the sizes of all these extensions across all embeddings into the proof corpus, we can  obtain  an upper bound  $\strength(\tc', \corpus)$ for any $\tc'$  that is an extension of $\tc$.  This is precisely what the {\sc UpperBound} procedure in Figure~\ref{alg:ub} computes.

\vspace{0.05in}
\noindent {\bf \emph{Tactic Library Synthesis.}}
Finally, to generate a \emph{library} of tactics, our algorithm  synthesizes a single tactic by calling {\sc LearnTactic}, then refactors the corpus using this tactic,  and repeats this process until  no more tactics  can be learned.

\vspace{0.05in}
\noindent
\textbf{\emph{Discussion: Learning with Other Objectives.}}  
While our algorithm follows standard practice in optimizing \emph{compression power}~\cite{stitch,babble,dreamcoder}, it naturally generalizes to other objectives. For example, one can consider a broader optimization function of the form  
$ 
\alpha \times \textsc{CP}(\tau, \corpus) - \beta \times F(\tau),
$  
where \(F\) penalizes tactics that do not satisfy certain syntactic or semantic criteria, such as requiring too many arguments or being too narrowly applicable. Adapting our method to alternative objectives requires modifying two aspects of the learning algorithm: (1) the objective function itself, captured by \(\strength\), and (2) the \textsc{UpperBound} procedure, which estimates an upper bound on the objective value. For instance, for an objective of the form \(\alpha \times \textsc{CP}(\tau, \corpus) - \beta \times F(\tau)\), a straightforward implementation of \textsc{UpperBound} could reuse our existing technique for over-approximating compression power while conservatively assuming that the penalty term is zero. 
\section{Implementation}\label{sec:impl}
We have implemented the proposed algorithm as a new tool called \toolname, which consists of about 5000 lines of Java code and utilizes the external Coq-SerAPI library~\cite{serapi} to parse proof scripts and extract any information  necessary for constructing TDGs. The current version of \toolname only supports Ltac~\cite{ltac}, as it remains the dominant tactic language in existing Rocq developments.


\vspace{0.05in}
\noindent{\bf \emph{Handling Overlapping Embeddings.}} In the technical presentation, we assume that, if there are multiple isomorphic collapsible embeddings of a tactic into a proof, then they are non-overlapping; however, our implementation does not make this assumption. In particular, when computing an upper bound on effectiveness during the tactic discovery process, we assume that all embeddings can be used for refactoring, giving a conservative upper bound on the effectiveness of a given tactic candidate. Similarly, when computing the actual effectiveness of a tactic, we perform backtracking search to find a disjoint subset of embeddings that maximizes the effectiveness score. 




\vspace{0.05in}
\noindent{\bf{\emph{Converting TDGs to Proof Scripts.}}} Since our approach refactors a proof at the TDG level, we need to convert it back to a valid Rocq proof script. As discussed in Section~\ref{sec:refactor}, we can do this by performing a topological sort of the TDG. However, if a tactic produces multiple sub-goals as its output, we need to ensure that each sub-goal is processed as a separate branch. Our algorithm for converting TDGs to proof scripts performs additional analysis to keep track of this information and ensures that tactics that are used for discharging the same sub-goal appear in the same branch, subject to the topological sorting constraints within that branch.
\vspace{0.05in}

\noindent \textbf{\emph{Supported Tactics.}}
Our method places no restrictions on the tactics used in the input proof scripts. Any tactic that can be parsed using SerAPI, including advanced tactics like \texttt{match goal} and \texttt{eapply}, as well as user-defined custom tactics, is supported. Constructs that bundle combinators (e.g., \texttt{try}, \texttt{first}, \texttt{repeat}) with tactics are treated an atomic tactics, whereas tactics connected by the sequencing operator \texttt{;} are explicitly decomposed into their individual invocations. There are, however, restrictions on what our system can \emph{learn}. We do not synthesize tactics whose \emph{control flow} depends on inspecting the current proof state (e.g., \texttt{match goal})\footnote{In contrast, tactics such as \texttt{intros}, \texttt{simpl}, or \texttt{auto} are supported: although their effects depend on the proof state, their execution does not require branching on it.}, because a TDG records only the executed trace and omits unexecuted branches. Similarly, while our approach can learn tactics involving existential variables (\texttt{evars}), it does so only when their effects remain localized within the proof states that arise from their instantiation.


    
\begin{table}[H]
    \centering
    \footnotesize
    \vspace{-0.05in}
    \caption{Summary of Benchmarks.  ``Total size''  reports total number of pre-defined tactic applications within the domain, and ``Avg. size / proof'' reports average proof size in terms of the number of tactic applications. }
    \vspace{-0.1in}
        \begin{tabular}{cccccc}
        \toprule
        {\bf Domain} &{\bf Topic} & {\bf Description} & {\bf \# of Proof} & {\bf Total size} & {\bf Avg size / proof} \\
        \hline
        & \textsc{IndPred} & Inductive Predicates & {61} & {620} & {10}\\
        {\bf CoqArt} & \textsc{SearchTree} & Search Trees & {48} & {781} & {16} \\
        & \textsc{Reflection} & Proof by Reflection & {33} & {668} & {20} \\
        \hline
        & \textsc{Hoare} & Hoare logic & 65 & 1479 & 23 \\
        {\bf Program} & \textsc{Separation} & Heap properties & 58 & 592 & 10 \\
        {\bf Logics} & \textsc{Seplog} & Separation logic & 70 & 1111 & 16 \\
        & \textsc{CSL} & Concurrent separation logic & 47 & 1282 & 27 \\
        \hline
        & \textsc{RegAlloc} & Register allocation with validation & 31 & 467 & 15\\
        {\bf Comp-} & \textsc{LiveRange} & Proofs for live ranges computation & 32 & 903 & 28 \\
        {\bf Cert} & \textsc{Needness} & Abstract domain for needness analysis & 103 & 1759 & 17 \\
        & \textsc{RTLSpec} & Abstract specification for RTL generation & 55 & 1413 & 26 \\
        \hline
        & \textsc{NMake} & Big natural numbers & {105} & {1465} & {14}\\
        {\bf BigNums} & \textsc{ZMake} & Big integers & {43} & {801} & {19} \\
        & \textsc{QMake} & Big rational numbers & {68} & {1392} & {20} \\
        \bottomrule
        \end{tabular}
        \label{tab:benchmark-stats}
        \vspace{-0.1in}
    \end{table}

\section{Evaluation}
We now describe our experimental evaluation that is designed to answer the following  questions: 

\begin{itemize}[leftmargin=*]
    \item {\bf RQ1: } How effective is our proposed technique in learning tactics compared to a prior approach~\cite{peano} that learns tactics using anti-unification? 
    \item {\bf RQ2: } How useful are the learned tactics in terms of compression rate? 
    \item {\bf RQ3: } Can our learned tactics help with proof automation?
    \item {\bf RQ4: } How data efficient is our learning algorithm? 
    \item {\bf RQ5: } What impact do the key ingredients of our learning algorithm have on  running time?

\end{itemize}

\subsection{Experimental Setup}\label{sec:exp-setup}

{\bf \emph{Benchmarks.}} To evaluate our approach, we collected {918} proofs across {four} different sources:  a textbook called ``Interactive Theorem Proving and Program Development''\cite{bertot2004interactive} (also known as CoqArt\cite{coq-art}), the formally verified C compiler CompCert~\cite{compcert}, the BigNums arbitrary-precision arithmetic library~\cite{bignum}, and formalizations of various program  logics~\cite{leroy2021cdfprogramlogics}. For each of these sources, we choose these sub-domains that satisfy certain criteria, such as containing a threshold number of proofs ($\geq 30$ across all sources). 
Table~\ref{tab:benchmark-stats} provides a summary of our experimental benchmarks, including a description of the sub-domain that the proofs pertain to and the number of proofs in each sub-domain. The last two columns in the table show the total number of tactic applications in each domain and the average number of tactics per proof. Overall, these benchmarks span a wide spectrum of proof domains, ranging from introductory proofs in CoqArt to expert-level developments in large-scale systems like CompCert. Furthermore, they include a mix of programming languages (PL)-specific benchmarks—such as those involving program logics—and formalizations with broader mathematical utility, such as those for infinite-precision arithmetic.

\begin{table}[b]
\centering
\footnotesize
\vspace{-0.16in}
\caption{Results for evaluating the effectiveness of the tactic learning algorithm}
\vspace{-0.1in}
    \begin{tabular}{@{}c@{\hskip 4pt}c@{\hskip 4pt}c@{\hskip 6pt}c@{\hskip 6pt}c@{\hskip 6pt}c@{\hskip 6pt}c}
    \toprule
    & {\multirow{2}{*}{{\bf Topic}}} & {\multirow{2}{*}{{\bf Tool}}} & {\bf \# Tactics} & {\bf Avg Tactic} & {\bf Max Tactic} & {\bf Tactic Usage} \\
    & & & {\bf Learned} & {\bf Size} & {\bf Size} &  {\bf Count} \\
    \hline
    \multirow{6}{*}{\rotatebox{90}{\textbf{CoqArt}}}
    & \multirow{2}{*}{\textsc{IndPred}} & \toolname & {31} & {2.8} & {14} & {83} \\
    & & \peano & {16} & {2} & {2} & {48} \\
    \cdashline{2-7}
    & \multirow{2}{*}{{\textsc{SearchTree}}} & \toolname & {54} & {3.4} & {17} & {161} \\
    & & \peano & {20} & {2.4} & {6} & {31} \\
    \cdashline{2-7}
    & \multirow{2}{*}{{\textsc{Reflection}}} & \toolname & {38} & {3.3} & {11} & {100} \\
    & & \peano & {20} & {2.6} & {9} & {62} \\
    \hline
    \multirow{8}{*}{\rotatebox{90}{\textbf{Program Logics}}}
    & \multirow{2}{*}{\textsc{Hoare}} & \toolname & 87 & 3.2 & 15 & 231 \\
    & & \peano & 18 & 2.2 & 4  & 55 \\
    \cdashline{2-7}
    & \multirow{2}{*}{\textsc{Separation}} & \toolname & 36 & 3.1 & 8 & 93 \\
    & & \peano & 26 & 2.5 & 7 & 68 \\
    \cdashline{2-7}
    & \multirow{2}{*}{\textsc{Seplog}} & \toolname & 69 & 3.4 & 13 & 188 \\
    & & \peano & 16 & 2.3 & 4 & 56 \\
    \cdashline{2-7}
    & \multirow{2}{*}{\textsc{CSL}} & \toolname & {82} & 3.6 & 19 & {214} \\
    & & \peano & 14 & 2.1 & 3 & 44 \\
    \hline
    \multirow{8}{*}{\rotatebox{90}{\textbf{CompCert}}}
    & \multirow{2}{*}{\textsc{RegAlloc}} & \toolname & 18 & 8.8 & 34 & 50 \\
    & & \peano & 8 & 7 & 18 & 13 \\
    \cdashline{2-7}
    & \multirow{2}{*}{\textsc{LiveRange}} & \toolname & 61 & 3.8 & 16 & 162 \\
    & & \peano & 17 & 2.2 & 5 & 42 \\
    \cdashline{2-7}
    & \multirow{2}{*}{\textsc{AbsDomain}} & \toolname & 100 & 4.3 & 23 & 243 \\
    & & \peano & 39 & 2.1 & 3 & 135 \\
    \cdashline{2-7}
    & \multirow{2}{*}{\textsc{RTLSpec}} & \toolname & 93 & 3.4 & 13 & 260\\
    & & \peano & 25 & 2 & 4 & 58 \\
    \hline
    \multirow{6}{*}{\rotatebox{90}{\textbf{BigNums}}}
    & \multirow{2}{*}{{\textsc{NMake}}} & \toolname & {91} & {2.9} & {9} & {263} \\
    & & \peano & {29} & {2.3} & {5} & {84} \\
    \cdashline{2-7}
    & \multirow{2}{*}{{\textsc{ZMake}}} & \toolname & {56} & {4.8} & {19} & {132} \\
    & & \peano & {23} & {3.4} & {10} & {48} \\
    \cdashline{2-7}
    & \multirow{2}{*}{{\textsc{QMake}}} & \toolname & {95} & {3.6} & {21} & {250} \\
    & & \peano & {36} & {2.2} & {4} & {98} \\
    \midrule
    \midrule
    & \multirow{2}{*}{\bf Overall} & \toolname & {918}  & {3.6} & 34 & {2430} \\
    & & \peano & {310} & {2.4} & 10 & {842} \\
    \bottomrule
    \end{tabular}
    \label{tab:ev1}
    \vspace{-0.05in}
\end{table}

\noindent 
{\bf \emph{Baseline.}} Because no prior work exists on tactic library learning for Rocq proofs, we will address our first two research questions and compare our approach against a baseline by adapting the closest relevant prior work: tactic discovery via anti-unification, proposed in {\sc Peano}, a system designed to automate K12 math proofs.  
In particular, {\sc Peano} represents proofs as sequences of actions, and, for each pair of subsequences from different proofs, {\sc Peano} applies anti-unification to find the least general generalization between them. If the two sequences invoke the same actions but on different inputs, anti-unification identifies the shared structure while introducing variables for the differing parts. Similar to our approach, {\sc Peano} also aims to maximally compress the existing proofs. 

Since the implementation of the tactic induction method in {\sc Peano} does not work on Rocq proofs, we re-implemented their technique to extract tactics from a corpus of Rocq proof scripts. Towards this goal, we represent Rocq proof scripts as a sequence of tactics (i.e., ``actions'' in {\sc Peano} terminology) and use the same anti-unification approach to learn custom tactics. However, this syntactic approach does not guarantee that an extracted tactic  can be successfully used to refactor the proof from which it was derived.  Hence, given a tactic library learned using this baseline, we perform post-processing to discard tactics that cannot be used to re-factor \emph{any} proof in the corpus.
\vspace{0.05in}

\noindent
{\bf \emph{Computational Resources}}
All of our experiments are conducted on a machine with an Apple M2 CPU and 24 GB of physical memory, running on the macOS operating system.
\vspace{-0.1in}

\vspace{0.05in}
\subsection{Evaluation of Learned Tactics}\label{sec:effectiveness}
To answer our first research question, we use both \toolname and \peano to extract tactics from each domain. This evaluation, summarized in Table~\ref{tab:ev1}, treats the entire proof corpus of each benchmark as the training data. The figure provides statistics about the number of tactics learned by each technique and quantitative metrics to assess their quality, including average size of learned tactics, the maximum size among all learned tactics, and the total number of times the learned tactics can be applied in the proof corpus. As shown in this table, our proposed method learns substantially more tactics compared to {\sc Peano} (around 3$\times$ more across all benchmarks) and the learned tactics tend to be larger on average (1.5$\times$). Furthermore, across all benchmarks, the tactics learned by {\sc Peano} can be applied a total of 842 times, whereas those learned by \toolname can be used 2430 times. We believe these results demonstrate that our proposed TDG abstraction and learning algorithm allow for more effective tactic  discovery compared to a baseline that extracts tactics from a syntactic representation of proofs.  \\
\begin{mdframed}
{\bf Results for RQ1:} \toolname extracts around 3$\times$ more tactics compared to \peano. Furthermore, the tactics learned by \toolname are both larger and more frequently applicable in the corpus.
\end{mdframed}
\vspace{-1pt}

\subsection{Effectiveness of the Learned Tactics for  Proof Refactoring}\label{sec:usefulness}

\begin{wrapfigure}{r}{0.58\textwidth}
    \vspace{-0.2in}
    \definecolor{peachcrayola}{RGB}{247,197,159}
\definecolor{beige}{RGB}{239, 239, 208}
\definecolor{darkbluegray}{RGB}{102,106,134}
\definecolor{shadowblue}{RGB}{120,138,163}
\definecolor{opal}{RGB}{146,182,177}
\definecolor{laurelgreen}{RGB}{178,201,171}
\definecolor{dutchwhite}{RGB}{232,221,181}
\begin{tikzpicture}
\begin{axis}[
    width=8cm,
    height=4cm,
    ybar=0pt,  
    bar width=0.2cm,  
    ylabel={\footnotesize Compression Rate per Topic},
    ymin=1,
    ymax=1.7,
    xtick=data,
    xticklabels={IndPred, SearchTree, Reflection, Hoare,Separation,Seplog,CSL,RegAlloc,LiveRange,AbsDomain,RTLSpec, NMake, QMake, ZMake},
    y tick label style={font=\footnotesize},
    x tick label style={rotate=40,align=center,text width=0.7cm, font=\footnotesize},
    legend style={
        at={(0.9,0.9)},
        anchor=east,
        legend columns=2,
        /tikz/every even column/.append style={column sep=0.2cm, font=\footnotesize}
    },
    legend image code/.code={
    \draw [#1] (0cm,-0.1cm) rectangle (0.2cm,0.25cm); },
    symbolic x coords={IndPred, SearchTree, Reflection, Hoare,Separation,Seplog,CSL,RegAlloc,LiveRange,AbsDomain,RTLSpec, NMake, QMake, ZMake},
    nodes near coords align={vertical},
    ymajorgrids=false,
    axis x line*=bottom,
    axis y line=left,
    nodes near coords style={font=\tiny},
    enlarge x limits=0.1
]

\addplot[   
    bar shift=-0.1cm,  
    black,
    fill=dutchwhite,
    postaction= {
        pattern=north east lines,  
    }
] coordinates {
    (IndPred,1.08)
    (SearchTree,1.04)
    (Reflection,1.06)
    (Hoare,1.14)
    (Separation,1.08)
    (Seplog,1.16)
    (CSL,1.18)
    (RegAlloc,1.05)
    (LiveRange,1.08)
    (AbsDomain,1.19)
    (RTLSpec,1.11)
    (NMake,1.09)
    (QMake,1.15)
    (ZMake,1.03)
};

\addplot[
    bar shift=0.1cm,  
    black,
    fill=opal,
    postaction = {
        pattern=north west lines,
    }
] coordinates {
    (IndPred,1.13)
    (SearchTree,1.28)
    (Reflection,1.18)
    (Hoare,1.6)
    (Separation,1.17)
    (Seplog,1.52)
    (CSL,1.52)
    (RegAlloc,1.47)
    (LiveRange,1.28)
    (AbsDomain,1.46)
    (RTLSpec,1.42)
    (NMake,1.24)
    (QMake,1.29)
    (ZMake,1.17)
};

\draw[dotted, black, thick] (axis description cs:0,0.5) -- (axis description cs:0.97,0.5) node[above] {\footnotesize {\sc TacMiner} \qquad \qquad \qquad};
\draw[dashed, darkbluegray] (axis description cs:0,0.129)  node[above] {\footnotesize \quad \qquad {\sc Peano}} -- (axis description cs:0.97,0.129);
\legend{\peano,\toolname}

\end{axis}
\end{tikzpicture}
    \vspace{-0.15in}
    \caption{Average compression power per topic. The dotted lines denote the total across topics for each tool.}\label{fig:ev2}
    \vspace{-0.12in}
\end{wrapfigure}

To answer our second research question, we evaluate how useful the learned tactics are in  refactoring previously unseen proofs. To perform this experiment, we split the benchmarks into two separate training and test sets. We use the training set for tactic discovery but evaluate the usefulness of learned tactics \emph{only} on the test set. For this experiment, we use 65\% of the proofs for training (selected via an automated sampling script), and the remaining 35\% as the test set. We further evaluate the effectiveness of our approach for different training vs. test set ratios in Section~\ref{sec:data-eff}.

The results of this evaluation are presented in Figure~\ref{fig:ev2} where the y-axis shows the compression power (Def~\ref{def:cp}) of the learned tactics on the test set. For each category of benchmarks, we show the compression power achieved by both \toolname and \peano. As we can see from this figure, \toolname results in significantly higher compression power compared to \peano in \emph{all} categories. Across all benchmarks, the compression power of \toolname is 1.35 for \toolname vs 1.1 for \peano --- this means that, using the tactics learned by \toolname, the size of the proof corpus can be reduced by 26\%, whereas  the tactics learned by \peano result in a reduction of only 9\%. Furthermore, for the Hoare logic proofs, \toolname achieves the maximum compression of 1.6, indicating that the learned tactics for this domain are particularly effective at simplifying other proofs in the same domain. That is, for this domain, the refactored proofs are  approximately 63\% of their old size. \\

\begin{mdframed}
{\bf Results for RQ2: } 
Using the tactics learned by \toolname, the size of the proof corpus can be reduced by $26\%$. In contrast, {\sc Peano}'s tactics can reduce the corpus size by only 9\%. 
\end{mdframed}

\subsection{Using Learned Tactics for Proof Automation}
\label{sec:copra}

\begin{wraptable}{r}{0.52\textwidth} 
    \centering
    \small
    \vspace{-0.16in}
    \caption{Comparison of theorem proving success rates.}
    \label{tab:copra}
    \vspace{-0.08in}
    \begin{tabular}{|l|c|}
        \hline
        {\bf Method}                & {\bf Theorems proved} \\
        \hline
        {\sc Copra} with built-in tactics & 11/50 (22\%)           \\
        {\sc Copra} with \peano tactics    & 18/50 (36\%)           \\
        {\sc Copra} with \toolname tactics & 30/50 (60\%)         \\
        \hline
    \end{tabular}
    \vspace{-0.05in}
\end{wraptable}
Next, we evaluate whether our learned custom tactics can help a proof automation tool. To perform this investigation, we modify {\sc Copra}~\cite{copra}, a state-of-the-art proof automation tool based on Large Language Models (LLMs), to leverage our learned tactics. To adapt {\sc Copra} to use custom tactics, we modify the prompt provided to the LLM for in-context learning~\cite{icl}. Specifically, for each custom tactic learned by \toolname, we add the tactic definition as {\sc Copra}'s context as well as one example showing a proof state where that tactic was used (after being refactored by \toolname) along with the tactic's invocation (i.e., its arguments) that took place in the refactored proof. These two pieces of information should, in principle, allow {\sc Copra} to leverage custom tactics. For this experiment, we use OpenAI's \texttt{gpt-4o-2024-10-06} \cite{hurst2024gpt} as the underlying LLM.

\begin{wrapfigure}{r}{0.4\textwidth}
    \vspace{-0.14in}
    \resizebox{1.0\linewidth}{!}{\definecolor{opal}{RGB}{146,182,177}
\definecolor{darkbluegray}{RGB}{102,106,134}
\begin{tikzpicture}
\begin{axis}[
    width=7cm,
    height=5cm,
    ymajorgrids=false,
    axis x line*=bottom,
    axis y line=left,
    xlabel={\small \% of Training Data},
    ylabel={\small Average Compression Rate},
    legend pos=north west,
    xmin=0.1,
    xmax=1,
    ymin=1,
    ymax=1.5,
    error bars/y dir=both,
    error bars/y explicit,
    enlarge x limits=0.02
]

\addplot[
    darkbluegray,
    mark=*,
    thick,
    error bars/.cd,
    y dir=both,
    y explicit,
] coordinates {
    (0.20,1.10) +- (0,0.01)
    (0.25,1.12) +- (0,0.01)
    (0.30,1.13) +- (0,0.01)
    (0.35,1.15) +- (0,0.02)
    (0.40,1.17) +- (0,0.02)
    (0.45,1.19) +- (0,0.02)
    (0.50,1.21) +- (0,0.02)
    (0.55,1.23) +- (0,0.03)
    (0.60,1.25) +- (0,0.03)
    (0.65,1.26) +- (0,0.03)
    (0.70,1.27) +- (0,0.03)
    (0.75,1.29) +- (0,0.03)
    (0.80,1.30) +- (0,0.03)
    (0.85,1.31) +- (0,0.03)
    (0.90,1.32) +- (0,0.03)
    (0.95,1.32) +- (0,0.03)
    (1.00,1.35) +- (0,0.04)
};

\end{axis}
\end{tikzpicture}}
    \vspace{-0.2in}
    \caption{Data efficiency for \toolname. The error bar reflects the standard error.}
    \label{fig:data-eff}
    \vspace{-0.05in}
\end{wrapfigure}

To evaluate the utility of learned tactics for downstream proof automation, we selected a set of 50 theorems from the same sources listed in Table~\ref{tab:benchmark-stats}. These theorems were chosen based on two criteria. First, they are theorems for which {\sc Copra}, with a smaller computational budget than our official evaluation, was able to make partial progress—successfully solving some sub-goals but failing to complete the entire proof. This ensures that the benchmarks are neither trivial nor intractable, providing a meaningful setting for evaluating the added value of tactic learning. Second, we required that the corresponding refactored proof makes use of at least one tactic discovered by our system, ensuring that the learned tactics are relevant and can plausibly aid automation. Together, these criteria yield a benchmark set well-suited for measuring the practical benefits of tactic reuse in proof automation workflows.



The results of this evaluation are presented in Table~\ref{tab:copra}. The  vanilla {\sc Copra} baseline using only built-in tactics can prove 11 out of the 50 theorems, resulting in a success rate of $22\%$. We then independently evaluate {\sc Copra} augmented with learned tactics from {\sc Peano} and {\sc TacMiner} on the same 50 theorems. {\sc Copra} supplied with the tactics learned by {\sc Peano} can prove 7 additional theorems, increasing success rate to 36\%. Finally, using the custom tactics learned by our method, {\sc Copra}  can prove 19 additional  theorems over the baseline, increasing the success rate to $60\%$. These results demonstrate that (1) custom tactics can be useful for improving proof automation, and (2) tactics learned by our method are more useful for proof automation compared to the {\sc Peano} baseline. \\

\begin{mdframed}
{\bf Results for RQ3: } Using the custom tactics learned by \toolname, we are able to increase the success rate of an LLM-based proof automation tool from 22\% to 60\%.
\end{mdframed}

\subsection{Data efficiency of proposed tactic learning technique}\label{sec:data-eff}

The results of this evaluation are presented in Figure~\ref{fig:data-eff}. As expected, the larger the number of training benchmarks, the more effective the learning technique. However, even if we only use 25\% of the benchmarks for training, \toolname still achieves a compression power of around 1.13, which is higher than that achieved by \peano with a much larger training set. \\

\begin{mdframed}
{\bf Results for RQ4: } While \toolname benefits from a larger training corpus, it can still achieve significant compression on the test set as we decrease the size of the training data.  
\end{mdframed}

\subsection{Ablation Studies for Tactic Discovery Algorithm}
To answer our final research question, we present the results of an ablation study in which we disable some of the key components of our tactic learning algorithm. In particular, we consider the following two ablations of \toolname:

\begin{itemize}[leftmargin=*]
\item {\bf \nogrammar:} This ablation performs a limited form of grammar learning.\footnote{We also tried switching off grammar learning entirely, but since it performs \emph{extremely poorly}, we only report the results of a limited form of grammar learning.} In particular, it learns  graph grammar rules of the form $\eta \rightarrow \eta'$ instead of $\eta \rightarrow (\eta', \theta)$, meaning that it does not have prior knowledge about the argument dependencies between tactics. 
\item {\bf \noub:} This ablation does not use the pruning procedure ({\sc UpperBound}) from Section~\ref{sec:learn-tactic}. 
\end{itemize}

The results of this ablation study  are presented in Figures \ref{fig:eval4-cdf} (for Program Logics), \ref{fig:eval4-compcert} (for CompCert), \ref{fig:eval4-coqart} (for CoqArt), and \ref{fig:eval4-bignum} (for BigNums).  Here, the $x$-axis shows the number of total tactics learned, and the $y$-axis shows the cumulative learning time in seconds. Note that the $y$-axis uses log scale. As we can see from these plots, both  graph grammar learning and pruning via upper bound estimation have a huge impact on the running time of the tactic discovery algorithm. While \toolname can terminate on the entire corpus in about 13 minutes, both of the ablations fail to terminate within a 30-minute time limit.  \\

\begin{mdframed}
{\bf Results for RQ5: } Both of our algorithmic optimizations (namely,  grammar learning and pruning method) significantly reduce the learning algorithm's running time.
\end{mdframed}

\begin{figure}[H]
    \centering
    \begin{minipage}{0.47\textwidth} 
        \centering
        \resizebox{\linewidth}{!}{\definecolor{peachcrayola}{RGB}{247,197,159}
\definecolor{beige}{RGB}{239, 239, 208}
\definecolor{darkbluegray}{RGB}{102,106,134}
\definecolor{shadowblue}{RGB}{120,138,163}
\definecolor{opal}{RGB}{146,182,177}
\definecolor{laurelgreen}{RGB}{178,201,171}
\definecolor{dutchwhite}{RGB}{232,221,181}
\begin{tikzpicture}
\begin{semilogyaxis}[
    width=7cm,
    height=4cm,
    ymajorgrids=false,
    axis x line*=bottom,
    axis y line=left,
    xlabel={\small Number tactics learned},
    ylabel={\small Cumulative Time (s)},
    legend columns=-1,
    legend style={at={(0,1.05)}, anchor=south west, font=\footnotesize},
    legend cell align={left},
    ymin=0.01,
    ymax=2000,
    xmin=0,
    xmax=280,
]

\addplot[very thick, densely dotted, opal] table[x index=0, y index=1] {
1.0 0.012
2.0 0.025
3.0 0.038
4.0 0.052
5.0 0.066
6.0 0.08
7.0 0.095
8.0 0.11
9.0 0.126
10.0 0.142
11.0 0.158
12.0 0.174
13.0 0.19
14.0 0.207
15.0 0.224
16.0 0.242
17.0 0.26
18.0 0.278
19.0 0.297
20.0 0.317
21.0 0.337
22.0 0.357
23.0 0.378
24.0 0.399
25.0 0.42
26.0 0.442
27.0 0.465
28.0 0.488
29.0 0.511
30.0 0.534
31.0 0.559
32.0 0.585
33.0 0.611
34.0 0.637
35.0 0.663
36.0 0.69
37.0 0.717
38.0 0.744
39.0 0.771
40.0 0.798
41.0 0.826
42.0 0.854
43.0 0.882
44.0 0.911
45.0 0.94
46.0 0.969
47.0 0.999
48.0 1.03
49.0 1.061
50.0 1.092
51.0 1.123
52.0 1.155
53.0 1.188
54.0 1.221
55.0 1.255
56.0 1.291
57.0 1.327
58.0 1.364
59.0 1.401
60.0 1.438
61.0 1.476
62.0 1.514
63.0 1.552
64.0 1.591
65.0 1.63
66.0 1.67
67.0 1.711
68.0 1.752
69.0 1.794
70.0 1.836
71.0 1.879
72.0 1.922
73.0 1.965
74.0 2.009
75.0 2.053
76.0 2.097
77.0 2.141
78.0 2.187
79.0 2.236
80.0 2.285
81.0 2.335
82.0 2.387
83.0 2.442
84.0 2.497
85.0 2.555
86.0 2.614
87.0 2.673
88.0 2.733
89.0 2.794
90.0 2.857
91.0 2.921
92.0 2.985
93.0 3.052
94.0 3.12
95.0 3.191
96.0 3.262
97.0 3.333
98.0 3.405
99.0 3.478
100.0 3.551
101.0 3.624
102.0 3.698
103.0 3.772
104.0 3.847
105.0 3.923
106.0 4.0
107.0 4.08
108.0 4.162
109.0 4.244
110.0 4.33
111.0 4.416
112.0 4.503
113.0 4.593
114.0 4.684
115.0 4.777
116.0 4.871
117.0 4.965
118.0 5.059
119.0 5.155
120.0 5.252
121.0 5.352
122.0 5.455
123.0 5.56
124.0 5.672
125.0 5.786
126.0 5.9
127.0 6.016
128.0 6.134
129.0 6.253
130.0 6.372
131.0 6.494
132.0 6.618
133.0 6.747
134.0 6.877
135.0 7.007
136.0 7.144
137.0 7.289
138.0 7.435
139.0 7.59
140.0 7.745
141.0 7.912
142.0 8.09
143.0 8.268
144.0 8.463
145.0 8.666
146.0 8.87
147.0 9.081
148.0 9.294
149.0 9.536
150.0 9.778
151.0 10.033
152.0 10.293
153.0 10.563
154.0 10.837
155.0 11.114
156.0 11.4
157.0 11.691
158.0 11.997
159.0 12.309
160.0 12.622
161.0 12.94
162.0 13.289
163.0 13.64
164.0 13.999
165.0 14.361
166.0 14.748
167.0 15.153
168.0 15.559
169.0 15.973
170.0 16.419
171.0 17.01
172.0 17.618
173.0 18.257
174.0 18.98
175.0 19.708
176.0 20.471
177.0 21.314
178.0 22.164
179.0 23.054
180.0 24.09
181.0 25.15
182.0 26.402
183.0 27.708
184.0 29.602
185.0 36.606
186.0 43.714
187.0 53.037
188.0 69.221
189.0 112.099
190.0 155.382
191.0 198.721
192.0 255.89
193.0 316.737
194.0 556.266
};
\addlegendentry{PruningABL}

\addplot[very thick, densely dashdotted, peachcrayola] table[x index=0, y index=1] {
1.0 0.292
2.0 0.591
3.0 0.895
4.0 1.203
5.0 1.576
6.0 1.96
7.0 2.347
8.0 2.737
9.0 3.128
10.0 3.521
11.0 3.916
12.0 4.311
13.0 4.71
14.0 5.112
15.0 5.516
16.0 5.921
17.0 6.326
18.0 6.745
19.0 7.17
20.0 7.601
21.0 8.032
22.0 8.467
23.0 8.923
24.0 9.384
25.0 9.849
26.0 10.334
27.0 10.82
28.0 11.312
29.0 11.813
30.0 12.328
31.0 12.857
32.0 13.423
33.0 13.999
34.0 14.582
35.0 15.169
36.0 15.764
37.0 16.376
38.0 17.036
39.0 17.794
40.0 18.748
41.0 19.713
42.0 20.701
43.0 21.737
44.0 22.83
45.0 23.948
46.0 25.069
47.0 26.191
48.0 27.316
49.0 28.461
50.0 29.607
51.0 30.757
52.0 31.937
53.0 33.141
54.0 34.542
55.0 35.953
56.0 37.412
57.0 39.364
58.0 41.319
59.0 43.523
60.0 45.842
61.0 48.52
62.0 51.209
63.0 53.906
64.0 56.607
65.0 59.374
66.0 62.148
67.0 65.029
68.0 68.243
69.0 71.773
70.0 75.595
71.0 79.585
72.0 83.776
73.0 88.35
74.0 94.92
75.0 101.989
76.0 109.354
77.0 117.722
78.0 126.119
79.0 135.274
80.0 144.672
81.0 155.223
82.0 166.779
83.0 178.431
84.0 190.901
85.0 204.219
86.0 218.169
87.0 232.275
88.0 247.027
89.0 262.386
90.0 278.247
91.0 294.316
92.0 310.657
93.0 327.467
94.0 344.446
95.0 361.7
96.0 379.375
97.0 397.205
98.0 415.294
99.0 433.804
100.0 452.322
101.0 471.257
102.0 490.394
103.0 509.738
104.0 529.392
105.0 549.128
106.0 569.271
107.0 589.446
108.0 610.068
109.0 631.15
110.0 652.651
111.0 674.572
112.0 696.921
113.0 719.692
114.0 742.875
115.0 766.459
116.0 790.377
117.0 814.582
};
\addlegendentry{GrammarABL}

\addplot[very thick, darkbluegray] table[x index=0, y index=1] {
1.0 0.006
2.0 0.013
3.0 0.02
4.0 0.027
5.0 0.034
6.0 0.041
7.0 0.049
8.0 0.057
9.0 0.065
10.0 0.073
11.0 0.081
12.0 0.091
13.0 0.102
14.0 0.113
15.0 0.124
16.0 0.136
17.0 0.149
18.0 0.162
19.0 0.176
20.0 0.19
21.0 0.205
22.0 0.22
23.0 0.235
24.0 0.25
25.0 0.266
26.0 0.283
27.0 0.3
28.0 0.318
29.0 0.336
30.0 0.354
31.0 0.372
32.0 0.391
33.0 0.41
34.0 0.429
35.0 0.448
36.0 0.468
37.0 0.488
38.0 0.508
39.0 0.529
40.0 0.55
41.0 0.571
42.0 0.592
43.0 0.613
44.0 0.635
45.0 0.657
46.0 0.679
47.0 0.701
48.0 0.723
49.0 0.746
50.0 0.769
51.0 0.793
52.0 0.817
53.0 0.841
54.0 0.865
55.0 0.89
56.0 0.915
57.0 0.94
58.0 0.965
59.0 0.991
60.0 1.017
61.0 1.043
62.0 1.07
63.0 1.097
64.0 1.124
65.0 1.151
66.0 1.178
67.0 1.205
68.0 1.232
69.0 1.259
70.0 1.287
71.0 1.315
72.0 1.343
73.0 1.371
74.0 1.4
75.0 1.429
76.0 1.458
77.0 1.487
78.0 1.517
79.0 1.547
80.0 1.577
81.0 1.607
82.0 1.638
83.0 1.669
84.0 1.7
85.0 1.731
86.0 1.762
87.0 1.793
88.0 1.824
89.0 1.856
90.0 1.888
91.0 1.92
92.0 1.952
93.0 1.984
94.0 2.017
95.0 2.05
96.0 2.084
97.0 2.118
98.0 2.153
99.0 2.189
100.0 2.225
101.0 2.262
102.0 2.299
103.0 2.337
104.0 2.376
105.0 2.416
106.0 2.457
107.0 2.499
108.0 2.541
109.0 2.583
110.0 2.626
111.0 2.669
112.0 2.712
113.0 2.755
114.0 2.798
115.0 2.842
116.0 2.886
117.0 2.93
118.0 2.975
119.0 3.02
120.0 3.066
121.0 3.112
122.0 3.159
123.0 3.207
124.0 3.255
125.0 3.303
126.0 3.351
127.0 3.399
128.0 3.448
129.0 3.497
130.0 3.547
131.0 3.597
132.0 3.647
133.0 3.698
134.0 3.75
135.0 3.803
136.0 3.857
137.0 3.911
138.0 3.965
139.0 4.021
140.0 4.077
141.0 4.133
142.0 4.19
143.0 4.247
144.0 4.305
145.0 4.364
146.0 4.423
147.0 4.482
148.0 4.541
149.0 4.6
150.0 4.66
151.0 4.72
152.0 4.781
153.0 4.843
154.0 4.905
155.0 4.967
156.0 5.029
157.0 5.091
158.0 5.154
159.0 5.217
160.0 5.28
161.0 5.343
162.0 5.407
163.0 5.471
164.0 5.535
165.0 5.6
166.0 5.666
167.0 5.734
168.0 5.803
169.0 5.874
170.0 5.948
171.0 6.023
172.0 6.098
173.0 6.173
174.0 6.248
175.0 6.323
176.0 6.399
177.0 6.475
178.0 6.553
179.0 6.631
180.0 6.712
181.0 6.794
182.0 6.877
183.0 6.961
184.0 7.046
185.0 7.132
186.0 7.218
187.0 7.305
188.0 7.392
189.0 7.48
190.0 7.568
191.0 7.658
192.0 7.75
193.0 7.843
194.0 7.936
195.0 8.033
196.0 8.131
197.0 8.232
198.0 8.333
199.0 8.435
200.0 8.541
201.0 8.653
202.0 8.768
203.0 8.887
204.0 9.007
205.0 9.133
206.0 9.259
207.0 9.385
208.0 9.514
209.0 9.643
210.0 9.774
211.0 9.905
212.0 10.038
213.0 10.174
214.0 10.311
215.0 10.451
216.0 10.592
217.0 10.736
218.0 10.881
219.0 11.028
220.0 11.175
221.0 11.322
222.0 11.471
223.0 11.624
224.0 11.782
225.0 11.94
226.0 12.102
227.0 12.266
228.0 12.43
229.0 12.595
230.0 12.76
231.0 12.926
232.0 13.095
233.0 13.27
234.0 13.452
235.0 13.636
236.0 13.827
237.0 14.018
238.0 14.21
239.0 14.404
240.0 14.599
241.0 14.806
242.0 15.018
243.0 15.235
244.0 15.456
245.0 15.712
246.0 15.977
247.0 16.261
248.0 16.561
249.0 16.866
250.0 17.177
251.0 17.489
252.0 17.821
253.0 18.16
254.0 18.503
255.0 18.875
256.0 19.262
257.0 19.659
258.0 20.073
259.0 20.521
260.0 21.01
261.0 22.248
262.0 23.622
263.0 25.369
264.0 27.209
265.0 29.755
266.0 32.745
267.0 56.201
268.0 80.879
269.0 106.071
270.0 133.369
271.0 164.33
272.0 200.746
};
\addlegendentry{Ours}

    \end{semilogyaxis}
\end{tikzpicture}} 
        \vspace{-0.27in}
        \caption{Learning curve for ProgramLogic.}
        \label{fig:eval4-cdf}
        \vspace{0.15in}
    \end{minipage}%
    \hfill
    \begin{minipage}{0.45\textwidth} 
        \centering
        \resizebox{\linewidth}{!}{\definecolor{peachcrayola}{RGB}{247,197,159}
\definecolor{beige}{RGB}{239, 239, 208}
\definecolor{darkbluegray}{RGB}{102,106,134}
\definecolor{shadowblue}{RGB}{120,138,163}
\definecolor{opal}{RGB}{146,182,177}
\definecolor{laurelgreen}{RGB}{178,201,171}
\definecolor{dutchwhite}{RGB}{232,221,181}
\begin{tikzpicture}
\begin{semilogyaxis}[
    width=7cm,
    height=4cm,
    ymajorgrids=false,
    axis x line*=bottom,
    axis y line=left,
    xlabel={\small Number tactics learned},
    ymin=0.01,
    ymax=2000,
    xmin=0,
    xmax=280,
    legend columns=-1,
    legend style={at={(0,1.05)}, anchor=south west, font=\footnotesize},
    legend cell align={left},
    ]
    
\addplot[
    very thick,
    opal,
    densely dotted
    ] table [x index = 0, y index = 1] {
1.0 0.02
2.0 0.041
3.0 0.064
4.0 0.088
5.0 0.112
6.0 0.136
7.0 0.161
8.0 0.188
9.0 0.215
10.0 0.243
11.0 0.272
12.0 0.303
13.0 0.334
14.0 0.365
15.0 0.396
16.0 0.427
17.0 0.458
18.0 0.49
19.0 0.522
20.0 0.554
21.0 0.587
22.0 0.62
23.0 0.653
24.0 0.686
25.0 0.72
26.0 0.754
27.0 0.788
28.0 0.823
29.0 0.858
30.0 0.895
31.0 0.932
32.0 0.97
33.0 1.009
34.0 1.048
35.0 1.087
36.0 1.127
37.0 1.167
38.0 1.207
39.0 1.247
40.0 1.288
41.0 1.329
42.0 1.371
43.0 1.413
44.0 1.455
45.0 1.499
46.0 1.543
47.0 1.587
48.0 1.632
49.0 1.677
50.0 1.723
51.0 1.77
52.0 1.817
53.0 1.865
54.0 1.913
55.0 1.961
56.0 2.01
57.0 2.06
58.0 2.11
59.0 2.161
60.0 2.212
61.0 2.263
62.0 2.314
63.0 2.366
64.0 2.418
65.0 2.47
66.0 2.525
67.0 2.582
68.0 2.64
69.0 2.701
70.0 2.763
71.0 2.827
72.0 2.892
73.0 2.958
74.0 3.025
75.0 3.093
76.0 3.165
77.0 3.239
78.0 3.314
79.0 3.393
80.0 3.476
81.0 3.562
82.0 3.649
83.0 3.737
84.0 3.828
85.0 3.919
86.0 4.011
87.0 4.104
88.0 4.202
89.0 4.303
90.0 4.406
91.0 4.512
92.0 4.623
93.0 4.735
94.0 4.849
95.0 4.975
96.0 5.101
97.0 5.228
98.0 5.359
99.0 5.499
100.0 5.64
101.0 5.784
102.0 5.929
103.0 6.077
104.0 6.23
105.0 6.395
106.0 6.566
107.0 6.739
108.0 6.937
109.0 7.145
110.0 7.364
111.0 7.586
112.0 7.815
113.0 8.044
114.0 8.276
115.0 8.516
116.0 8.766
117.0 9.026
118.0 9.294
119.0 9.586
120.0 9.883
121.0 10.236
122.0 10.612
123.0 10.997
124.0 11.458
125.0 11.922
126.0 12.423
127.0 12.971
128.0 13.546
129.0 14.137
130.0 14.74
131.0 15.371
132.0 16.023
133.0 16.767
134.0 17.525
135.0 18.36
136.0 19.199
137.0 20.274
138.0 21.665
139.0 23.275
140.0 24.905
141.0 27.601
142.0 30.327
143.0 33.48
144.0 36.866
145.0 40.541
146.0 44.318
147.0 50.819
148.0 61.28
149.0 84.583
150.0 111.789
151.0 140.324
152.0 174.881
153.0 210.274
154.0 247.77
155.0 317.674
156.0 387.814
157.0 472.815
158.0 655.905
    };
    
\addplot[
    very thick,
    densely dashdotted, peachcrayola
    ] table [x index=0, y index=1] {
1.0 0.015
2.0 0.031
3.0 0.048
4.0 0.066
5.0 0.089
6.0 0.112
7.0 0.137
8.0 0.165
9.0 0.195
10.0 0.51
11.0 1.657
12.0 2.984
13.0 4.342
14.0 5.711
15.0 7.095
16.0 8.49
17.0 9.948
18.0 11.418
19.0 12.894
20.0 14.388
21.0 15.885
22.0 17.39
23.0 18.916
24.0 20.532
25.0 22.154
26.0 23.793
27.0 25.439
28.0 27.086
29.0 28.737
30.0 30.389
31.0 32.051
32.0 33.726
33.0 35.402
34.0 37.08
35.0 38.771
36.0 40.528
37.0 42.292
38.0 44.065
39.0 45.844
40.0 47.628
41.0 49.458
42.0 51.333
43.0 53.23
44.0 55.137
45.0 57.054
46.0 58.985
47.0 60.923
48.0 62.886
49.0 64.874
50.0 66.879
51.0 68.898
52.0 70.942
53.0 73.122
54.0 75.314
55.0 77.528
56.0 79.748
57.0 81.968
58.0 84.335
59.0 86.783
60.0 89.312
61.0 91.841
62.0 94.379
63.0 96.923
64.0 99.472
65.0 102.063
66.0 104.655
67.0 107.554
68.0 110.457
69.0 113.368
70.0 116.31
71.0 119.311
72.0 122.411
73.0 125.525
74.0 128.727
75.0 132.217
76.0 135.91
77.0 139.609
78.0 143.552
79.0 147.77
80.0 152.94
81.0 158.323
82.0 163.789
83.0 169.425
84.0 175.566
85.0 181.967
86.0 188.375
87.0 194.94
88.0 201.668
89.0 210.281
90.0 220.04
91.0 232.443
92.0 245.835
93.0 259.529
94.0 273.262
95.0 287.003
96.0 300.891
97.0 315.852
98.0 332.664
99.0 350.348
100.0 368.662
101.0 387.776
102.0 407.891
103.0 428.339
104.0 451.65
105.0 481.039
106.0 522.159
107.0 570.686
108.0 620.236
109.0 671.414
110.0 727.038
111.0 799.718
112.0 876.909
113.0 973.289
114.0 1123.786
115.0 1713.98
    };
    
    \addplot[
        very thick,
        darkbluegray,
        ] table [x index = 0, y index = 1] {
1.0 0.003
2.0 0.006
3.0 0.011
4.0 0.016
5.0 0.022
6.0 0.028
7.0 0.036
8.0 0.052
9.0 0.07
10.0 0.088
11.0 0.106
12.0 0.124
13.0 0.143
14.0 0.162
15.0 0.181
16.0 0.2
17.0 0.219
18.0 0.238
19.0 0.257
20.0 0.277
21.0 0.297
22.0 0.317
23.0 0.337
24.0 0.357
25.0 0.377
26.0 0.397
27.0 0.417
28.0 0.437
29.0 0.458
30.0 0.479
31.0 0.5
32.0 0.521
33.0 0.543
34.0 0.565
35.0 0.587
36.0 0.609
37.0 0.631
38.0 0.653
39.0 0.675
40.0 0.697
41.0 0.719
42.0 0.742
43.0 0.765
44.0 0.788
45.0 0.811
46.0 0.834
47.0 0.857
48.0 0.88
49.0 0.903
50.0 0.927
51.0 0.951
52.0 0.975
53.0 1.0
54.0 1.025
55.0 1.05
56.0 1.075
57.0 1.101
58.0 1.127
59.0 1.153
60.0 1.179
61.0 1.206
62.0 1.233
63.0 1.261
64.0 1.289
65.0 1.318
66.0 1.347
67.0 1.376
68.0 1.405
69.0 1.434
70.0 1.464
71.0 1.494
72.0 1.524
73.0 1.554
74.0 1.585
75.0 1.616
76.0 1.647
77.0 1.678
78.0 1.709
79.0 1.741
80.0 1.773
81.0 1.805
82.0 1.837
83.0 1.87
84.0 1.903
85.0 1.936
86.0 1.97
87.0 2.004
88.0 2.038
89.0 2.072
90.0 2.107
91.0 2.142
92.0 2.178
93.0 2.214
94.0 2.25
95.0 2.286
96.0 2.323
97.0 2.36
98.0 2.398
99.0 2.436
100.0 2.474
101.0 2.512
102.0 2.55
103.0 2.59
104.0 2.63
105.0 2.67
106.0 2.71
107.0 2.75
108.0 2.791
109.0 2.833
110.0 2.875
111.0 2.917
112.0 2.96
113.0 3.004
114.0 3.049
115.0 3.094
116.0 3.139
117.0 3.184
118.0 3.23
119.0 3.277
120.0 3.325
121.0 3.375
122.0 3.425
123.0 3.476
124.0 3.528
125.0 3.58
126.0 3.632
127.0 3.685
128.0 3.738
129.0 3.791
130.0 3.845
131.0 3.902
132.0 3.96
133.0 4.02
134.0 4.081
135.0 4.145
136.0 4.21
137.0 4.278
138.0 4.349
139.0 4.423
140.0 4.499
141.0 4.575
142.0 4.652
143.0 4.73
144.0 4.808
145.0 4.888
146.0 4.968
147.0 5.048
148.0 5.128
149.0 5.209
150.0 5.291
151.0 5.373
152.0 5.455
153.0 5.538
154.0 5.621
155.0 5.706
156.0 5.791
157.0 5.879
158.0 5.967
159.0 6.056
160.0 6.147
161.0 6.239
162.0 6.332
163.0 6.426
164.0 6.52
165.0 6.621
166.0 6.723
167.0 6.826
168.0 6.929
169.0 7.032
170.0 7.135
171.0 7.239
172.0 7.343
173.0 7.452
174.0 7.562
175.0 7.677
176.0 7.793
177.0 7.91
178.0 8.029
179.0 8.148
180.0 8.268
181.0 8.388
182.0 8.509
183.0 8.63
184.0 8.754
185.0 8.889
186.0 9.028
187.0 9.167
188.0 9.311
189.0 9.455
190.0 9.6
191.0 9.745
192.0 9.891
193.0 10.037
194.0 10.196
195.0 10.361
196.0 10.527
197.0 10.7
198.0 10.885
199.0 11.075
200.0 11.273
201.0 11.475
202.0 11.678
203.0 11.882
204.0 12.087
205.0 12.292
206.0 12.5
207.0 12.709
208.0 12.921
209.0 13.133
210.0 13.356
211.0 13.584
212.0 13.823
213.0 14.074
214.0 14.331
215.0 14.588
216.0 14.846
217.0 15.137
218.0 15.43
219.0 15.731
220.0 16.044
221.0 16.36
222.0 16.682
223.0 17.029
224.0 17.383
225.0 17.751
226.0 18.133
227.0 18.545
228.0 18.966
229.0 19.416
230.0 19.881
231.0 20.369
232.0 20.908
233.0 21.455
234.0 22.004
235.0 22.564
236.0 23.141
237.0 23.74
238.0 24.383
239.0 25.056
240.0 25.744
241.0 26.564
242.0 27.426
243.0 28.31
244.0 29.283
245.0 30.357
246.0 31.457
247.0 32.567
248.0 33.922
249.0 35.399
250.0 36.959
251.0 38.533
252.0 40.224
253.0 42.065
254.0 43.916
255.0 45.799
256.0 47.901
257.0 50.348
258.0 52.82
259.0 55.379
260.0 57.946
261.0 60.863
262.0 64.227
263.0 67.718
264.0 71.355
265.0 75.224
266.0 80.137
267.0 85.394
268.0 93.272
269.0 103.133
270.0 115.989
271.0 133.249
272.0 516.192
    };
    
\legend{PruningABL, GrammarABL, Ours}
\end{semilogyaxis}
\end{tikzpicture}} 
        \vspace{-0.27in}
        \caption{Learning curve for CompCert.}
        \label{fig:eval4-compcert}
        \vspace{0.15in}
    \end{minipage}
    
    \begin{minipage}{0.47\textwidth} 
        \centering
        \resizebox{\linewidth}{!}{\definecolor{peachcrayola}{RGB}{247,197,159}
\definecolor{beige}{RGB}{239, 239, 208}
\definecolor{darkbluegray}{RGB}{102,106,134}
\definecolor{shadowblue}{RGB}{120,138,163}
\definecolor{opal}{RGB}{146,182,177}
\definecolor{laurelgreen}{RGB}{178,201,171}
\definecolor{dutchwhite}{RGB}{232,221,181}
\begin{tikzpicture}
\begin{semilogyaxis}[
    width=7cm,
    height=4cm,
    ymajorgrids=false,
    axis x line*=bottom,
    axis y line=left,
    xlabel={\small Number tactics learned},
    ylabel={\small Cumulative Time (s)},
    legend columns=-1,
    legend style={at={(0,1.05)}, anchor=south west, font=\footnotesize},
    legend cell align={left},
    ymin=0.01,
    ymax=2000,
    xmin=0,
    xmax=130,
]

\addplot[very thick, densely dotted, opal] table[x index=0, y index=1] {
1.0 0.01
2.0 0.02
3.0 0.03
4.0 0.04
5.0 0.05
6.0 0.061
7.0 0.072
8.0 0.083
9.0 0.094
10.0 0.106
11.0 0.118
12.0 0.13
13.0 0.143
14.0 0.157
15.0 0.171
16.0 0.185
17.0 0.2
18.0 0.216
19.0 0.233
20.0 0.251
21.0 0.269
22.0 0.288
23.0 0.308
24.0 0.328
25.0 0.349
26.0 0.37
27.0 0.391
28.0 0.412
29.0 0.434
30.0 0.456
31.0 0.478
32.0 0.5
33.0 0.523
34.0 0.546
35.0 0.571
36.0 0.596
37.0 0.622
38.0 0.648
39.0 0.675
40.0 0.702
41.0 0.729
42.0 0.756
43.0 0.784
44.0 0.812
45.0 0.84
46.0 0.869
47.0 0.898
48.0 0.927
49.0 0.957
50.0 0.988
51.0 1.021
52.0 1.054
53.0 1.087
54.0 1.122
55.0 1.157
56.0 1.192
57.0 1.228
58.0 1.264
59.0 1.301
60.0 1.338
61.0 1.376
62.0 1.414
63.0 1.453
64.0 1.492
65.0 1.531
66.0 1.572
67.0 1.615
68.0 1.659
69.0 1.703
70.0 1.748
71.0 1.794
72.0 1.84
73.0 1.887
74.0 1.934
75.0 1.981
76.0 2.029
77.0 2.079
78.0 2.13
79.0 2.182
80.0 2.235
81.0 2.289
82.0 2.344
83.0 2.399
84.0 2.456
85.0 2.517
86.0 2.578
87.0 2.64
88.0 2.703
89.0 2.769
90.0 2.838
91.0 2.91
92.0 2.982
93.0 3.059
94.0 3.142
95.0 3.226
96.0 3.315
97.0 3.405
98.0 3.496
99.0 3.589
100.0 3.687
101.0 3.8
102.0 3.913
103.0 4.031
104.0 4.166
105.0 4.308
106.0 4.453
107.0 4.609
108.0 4.777
109.0 4.971
110.0 5.204
111.0 5.451
112.0 5.701
113.0 5.973
114.0 6.286
115.0 6.628
116.0 6.973
117.0 7.611
118.0 8.623
119.0 10.69
120.0 14.929
121.0 31.4
122.0 54.47
};
\addlegendentry{PruningABL}

\addplot[very thick, densely dashdotted, peachcrayola] table[x index=0, y index=1] {
1.0 0.072
2.0 0.145
3.0 0.219
4.0 0.295
5.0 0.371
6.0 0.455
7.0 0.54
8.0 0.625
9.0 0.71
10.0 0.796
11.0 0.882
12.0 0.969
13.0 1.056
14.0 1.144
15.0 1.232
16.0 1.32
17.0 1.409
18.0 1.499
19.0 1.59
20.0 1.681
21.0 1.773
22.0 1.865
23.0 1.958
24.0 2.051
25.0 2.145
26.0 2.242
27.0 2.341
28.0 2.442
29.0 2.547
30.0 2.654
31.0 2.763
32.0 2.872
33.0 2.981
34.0 3.092
35.0 3.205
36.0 3.318
37.0 3.432
38.0 3.546
39.0 3.661
40.0 3.776
41.0 3.892
42.0 4.009
43.0 4.127
44.0 4.246
45.0 4.365
46.0 4.484
47.0 4.604
48.0 4.728
49.0 4.852
50.0 4.98
51.0 5.108
52.0 5.236
53.0 5.365
54.0 5.494
55.0 5.625
56.0 5.76
57.0 5.896
58.0 6.032
59.0 6.171
60.0 6.311
61.0 6.452
62.0 6.593
63.0 6.736
64.0 6.88
65.0 7.026
66.0 7.173
67.0 7.321
68.0 7.472
69.0 7.627
70.0 7.784
71.0 7.942
72.0 8.101
73.0 8.263
74.0 8.426
75.0 8.589
76.0 8.753
77.0 8.919
78.0 9.085
79.0 9.251
80.0 9.417
81.0 9.584
82.0 9.756
83.0 9.943
84.0 10.141
85.0 10.353
86.0 10.567
87.0 10.788
88.0 11.012
89.0 11.236
90.0 11.463
91.0 11.702
92.0 11.943
93.0 12.192
94.0 12.47
95.0 12.75
96.0 13.051
97.0 13.352
98.0 13.662
99.0 13.977
100.0 14.298
101.0 14.636
102.0 14.976
103.0 15.323
104.0 15.674
105.0 16.068
106.0 16.468
107.0 16.876
108.0 17.285
109.0 17.728
110.0 18.221
111.0 18.751
112.0 19.349
113.0 19.979
114.0 20.622
115.0 21.868
116.0 23.153
117.0 24.712
118.0 26.326
119.0 29.111
120.0 44.884
121.0 74.026
122.0 578.706
123.0 1090.551
};
\addlegendentry{GrammarABL}

\addplot[very thick, darkbluegray] table[x index=0, y index=1] {
1.0 0.004
2.0 0.009
3.0 0.014
4.0 0.019
5.0 0.024
6.0 0.029
7.0 0.034
8.0 0.039
9.0 0.044
10.0 0.049
11.0 0.054
12.0 0.059
13.0 0.065
14.0 0.071
15.0 0.077
16.0 0.083
17.0 0.089
18.0 0.095
19.0 0.102
20.0 0.109
21.0 0.116
22.0 0.123
23.0 0.13
24.0 0.137
25.0 0.144
26.0 0.152
27.0 0.16
28.0 0.168
29.0 0.176
30.0 0.184
31.0 0.192
32.0 0.201
33.0 0.21
34.0 0.219
35.0 0.228
36.0 0.237
37.0 0.246
38.0 0.255
39.0 0.264
40.0 0.273
41.0 0.282
42.0 0.292
43.0 0.302
44.0 0.312
45.0 0.322
46.0 0.332
47.0 0.342
48.0 0.353
49.0 0.364
50.0 0.375
51.0 0.386
52.0 0.397
53.0 0.409
54.0 0.421
55.0 0.433
56.0 0.445
57.0 0.457
58.0 0.47
59.0 0.483
60.0 0.496
61.0 0.51
62.0 0.525
63.0 0.54
64.0 0.556
65.0 0.572
66.0 0.588
67.0 0.605
68.0 0.622
69.0 0.64
70.0 0.658
71.0 0.676
72.0 0.695
73.0 0.714
74.0 0.733
75.0 0.752
76.0 0.771
77.0 0.791
78.0 0.811
79.0 0.831
80.0 0.853
81.0 0.875
82.0 0.898
83.0 0.921
84.0 0.945
85.0 0.969
86.0 0.994
87.0 1.019
88.0 1.045
89.0 1.071
90.0 1.097
91.0 1.124
92.0 1.152
93.0 1.18
94.0 1.209
95.0 1.239
96.0 1.269
97.0 1.3
98.0 1.332
99.0 1.365
100.0 1.398
101.0 1.432
102.0 1.466
103.0 1.501
104.0 1.536
105.0 1.575
106.0 1.619
107.0 1.668
108.0 1.723
109.0 1.778
110.0 1.836
111.0 1.895
112.0 1.959
113.0 2.028
114.0 2.102
115.0 2.178
116.0 2.265
117.0 2.393
118.0 2.547
119.0 2.718
120.0 3.165
121.0 3.758
122.0 6.92
123.0 10.732
};
\addlegendentry{Ours}

    \end{semilogyaxis}
\end{tikzpicture}} 
        \vspace{-0.27in}
        \caption{Learning curve for CoqArt.}
        \label{fig:eval4-coqart}
    \end{minipage}%
    \hfill
    \begin{minipage}{0.45\textwidth} 
        \centering
        \resizebox{\linewidth}{!}{\definecolor{peachcrayola}{RGB}{247,197,159}
\definecolor{beige}{RGB}{239, 239, 208}
\definecolor{darkbluegray}{RGB}{102,106,134}
\definecolor{shadowblue}{RGB}{120,138,163}
\definecolor{opal}{RGB}{146,182,177}
\definecolor{laurelgreen}{RGB}{178,201,171}
\definecolor{dutchwhite}{RGB}{232,221,181}
\begin{tikzpicture}
\begin{semilogyaxis}[
    width=7cm,
    height=4cm,
    ymajorgrids=false,
    axis x line*=bottom,
    axis y line=left,
    xlabel={\small Number tactics learned},
    ylabel={\small Cumulative Time (s)},
    legend columns=-1,
    legend style={at={(0,1.05)}, anchor=south west, font=\footnotesize},
    legend cell align={left},
    ymin=0.01,
    ymax=2000,
    xmin=0,
    xmax=250,
]

\addplot[very thick, densely dotted, opal] table[x index=0, y index=1] {
1.0 0.008
2.0 0.017
3.0 0.026
4.0 0.035
5.0 0.045
6.0 0.055
7.0 0.065
8.0 0.075
9.0 0.086
10.0 0.099
11.0 0.113
12.0 0.127
13.0 0.141
14.0 0.156
15.0 0.171
16.0 0.187
17.0 0.203
18.0 0.219
19.0 0.235
20.0 0.251
21.0 0.268
22.0 0.285
23.0 0.302
24.0 0.319
25.0 0.336
26.0 0.353
27.0 0.37
28.0 0.388
29.0 0.406
30.0 0.424
31.0 0.442
32.0 0.461
33.0 0.48
34.0 0.499
35.0 0.518
36.0 0.538
37.0 0.559
38.0 0.58
39.0 0.601
40.0 0.622
41.0 0.643
42.0 0.664
43.0 0.686
44.0 0.708
45.0 0.73
46.0 0.752
47.0 0.774
48.0 0.797
49.0 0.82
50.0 0.843
51.0 0.867
52.0 0.891
53.0 0.915
54.0 0.939
55.0 0.963
56.0 0.987
57.0 1.012
58.0 1.037
59.0 1.062
60.0 1.087
61.0 1.113
62.0 1.139
63.0 1.165
64.0 1.191
65.0 1.217
66.0 1.244
67.0 1.271
68.0 1.298
69.0 1.326
70.0 1.354
71.0 1.382
72.0 1.412
73.0 1.442
74.0 1.472
75.0 1.502
76.0 1.532
77.0 1.563
78.0 1.594
79.0 1.625
80.0 1.656
81.0 1.688
82.0 1.72
83.0 1.753
84.0 1.787
85.0 1.821
86.0 1.856
87.0 1.891
88.0 1.926
89.0 1.962
90.0 2.0
91.0 2.039
92.0 2.078
93.0 2.118
94.0 2.158
95.0 2.198
96.0 2.239
97.0 2.28
98.0 2.321
99.0 2.362
100.0 2.404
101.0 2.446
102.0 2.488
103.0 2.532
104.0 2.576
105.0 2.62
106.0 2.665
107.0 2.711
108.0 2.757
109.0 2.804
110.0 2.852
111.0 2.901
112.0 2.95
113.0 2.999
114.0 3.049
115.0 3.099
116.0 3.152
117.0 3.206
118.0 3.26
119.0 3.317
120.0 3.374
121.0 3.431
122.0 3.49
123.0 3.552
124.0 3.617
125.0 3.683
126.0 3.751
127.0 3.825
128.0 3.902
129.0 3.982
130.0 4.064
131.0 4.147
132.0 4.232
133.0 4.318
134.0 4.407
135.0 4.499
136.0 4.591
137.0 4.684
138.0 4.78
139.0 4.876
140.0 4.975
141.0 5.074
142.0 5.173
143.0 5.276
144.0 5.38
145.0 5.485
146.0 5.591
147.0 5.699
148.0 5.809
149.0 5.919
150.0 6.03
151.0 6.142
152.0 6.258
153.0 6.378
154.0 6.501
155.0 6.627
156.0 6.757
157.0 6.889
158.0 7.022
159.0 7.16
160.0 7.299
161.0 7.441
162.0 7.586
163.0 7.733
164.0 7.881
165.0 8.03
166.0 8.179
167.0 8.329
168.0 8.484
169.0 8.643
170.0 8.802
171.0 8.963
172.0 9.135
173.0 9.31
174.0 9.487
175.0 9.664
176.0 9.85
177.0 10.037
178.0 10.228
179.0 10.447
180.0 10.673
181.0 10.906
182.0 11.142
183.0 11.383
184.0 11.638
185.0 11.896
186.0 12.167
187.0 12.455
188.0 12.744
189.0 13.033
190.0 13.327
191.0 13.633
192.0 13.951
193.0 14.293
194.0 14.635
195.0 14.978
196.0 15.341
197.0 15.751
198.0 16.177
199.0 16.625
200.0 17.109
201.0 17.697
202.0 18.291
203.0 18.961
204.0 19.674
205.0 20.39
206.0 21.121
207.0 21.878
208.0 22.685
209.0 23.517
210.0 24.44
211.0 25.433
212.0 26.441
213.0 27.463
214.0 28.491
215.0 29.568
216.0 30.671
217.0 31.832
218.0 33.032
219.0 34.291
220.0 35.652
221.0 37.21
222.0 39.048
223.0 41.27
224.0 43.976
225.0 46.701
226.0 49.605
227.0 52.53
228.0 56.025
229.0 60.729
230.0 65.957
231.0 105.101
232.0 151.746
233.0 245.759
234.0 340.63
235.0 453.139
};
\addlegendentry{PruningABL}

\addplot[very thick, densely dashdotted, peachcrayola] table[x index=0, y index=1] {
1.0 0.07
2.0 0.142
3.0 0.215
4.0 0.289
5.0 0.367
6.0 0.45
7.0 0.535
8.0 0.62
9.0 0.706
10.0 0.794
11.0 0.888
12.0 0.984
13.0 1.083
14.0 1.183
15.0 1.284
16.0 1.386
17.0 1.499
18.0 1.614
19.0 1.729
20.0 1.845
21.0 1.961
22.0 2.08
23.0 2.2
24.0 2.321
25.0 2.443
26.0 2.565
27.0 2.687
28.0 2.81
29.0 2.934
30.0 3.058
31.0 3.183
32.0 3.309
33.0 3.436
34.0 3.564
35.0 3.694
36.0 3.825
37.0 3.956
38.0 4.088
39.0 4.22
40.0 4.355
41.0 4.491
42.0 4.63
43.0 4.769
44.0 4.91
45.0 5.052
46.0 5.195
47.0 5.343
48.0 5.492
49.0 5.642
50.0 5.792
51.0 5.945
52.0 6.098
53.0 6.252
54.0 6.407
55.0 6.566
56.0 6.726
57.0 6.892
58.0 7.061
59.0 7.231
60.0 7.403
61.0 7.577
62.0 7.755
63.0 7.933
64.0 8.115
65.0 8.301
66.0 8.494
67.0 8.694
68.0 8.896
69.0 9.098
70.0 9.302
71.0 9.507
72.0 9.712
73.0 9.919
74.0 10.127
75.0 10.336
76.0 10.548
77.0 10.76
78.0 10.973
79.0 11.186
80.0 11.4
81.0 11.614
82.0 11.83
83.0 12.046
84.0 12.263
85.0 12.48
86.0 12.703
87.0 12.927
88.0 13.152
89.0 13.379
90.0 13.608
91.0 13.837
92.0 14.068
93.0 14.3
94.0 14.533
95.0 14.767
96.0 15.004
97.0 15.241
98.0 15.478
99.0 15.716
100.0 15.956
101.0 16.197
102.0 16.439
103.0 16.684
104.0 16.93
105.0 17.178
106.0 17.426
107.0 17.676
108.0 17.927
109.0 18.18
110.0 18.433
111.0 18.687
112.0 18.945
113.0 19.203
114.0 19.467
115.0 19.731
116.0 19.998
117.0 20.267
118.0 20.536
119.0 20.812
120.0 21.096
121.0 21.381
122.0 21.675
123.0 21.972
124.0 22.271
125.0 22.574
126.0 22.879
127.0 23.201
128.0 23.524
129.0 23.851
130.0 24.179
131.0 24.508
132.0 24.839
133.0 25.17
134.0 25.505
135.0 25.84
136.0 26.187
137.0 26.541
138.0 26.9
139.0 27.261
140.0 27.627
141.0 27.993
142.0 28.368
143.0 28.746
144.0 29.126
145.0 29.507
146.0 29.89
147.0 30.274
148.0 30.662
149.0 31.054
150.0 31.451
151.0 31.849
152.0 32.258
153.0 32.669
154.0 33.08
155.0 33.492
156.0 33.905
157.0 34.325
158.0 34.755
159.0 35.189
160.0 35.625
161.0 36.07
162.0 36.52
163.0 36.974
164.0 37.458
165.0 37.949
166.0 38.448
167.0 38.959
168.0 39.495
169.0 40.033
170.0 40.572
171.0 41.12
172.0 41.687
173.0 42.269
174.0 42.854
175.0 43.447
176.0 44.048
177.0 44.649
178.0 45.253
179.0 45.859
180.0 46.477
181.0 47.103
182.0 47.731
183.0 48.367
184.0 49.029
185.0 49.691
186.0 50.367
187.0 51.073
188.0 51.793
189.0 52.513
190.0 53.25
191.0 54.155
192.0 55.157
193.0 56.234
194.0 57.344
195.0 58.692
196.0 60.099
197.0 61.508
198.0 62.918
199.0 64.417
200.0 65.938
201.0 67.473
202.0 69.028
203.0 70.687
204.0 72.363
205.0 74.056
206.0 75.824
207.0 77.592
208.0 79.463
209.0 81.45
210.0 83.451
211.0 85.786
212.0 88.225
213.0 90.669
214.0 93.239
215.0 95.85
216.0 98.48
217.0 101.119
218.0 103.954
219.0 106.838
220.0 110.294
221.0 114.197
222.0 118.153
223.0 122.198
224.0 126.245
225.0 130.387
226.0 134.79
227.0 139.237
228.0 143.689
229.0 148.233
230.0 153.034
231.0 158.559
232.0 164.406
233.0 170.591
234.0 176.947
235.0 183.318
236.0 189.783
237.0 197.123
238.0 204.609
239.0 213.618
240.0 223.053
241.0 251.147
242.0 286.48
};
\addlegendentry{GrammarABL}

\addplot[very thick, darkbluegray] table[x index=0, y index=1] {
1.0 0.007
2.0 0.015
3.0 0.024
4.0 0.033
5.0 0.042
6.0 0.051
7.0 0.06
8.0 0.069
9.0 0.078
10.0 0.087
11.0 0.097
12.0 0.107
13.0 0.117
14.0 0.127
15.0 0.137
16.0 0.147
17.0 0.158
18.0 0.169
19.0 0.18
20.0 0.191
21.0 0.203
22.0 0.215
23.0 0.227
24.0 0.239
25.0 0.252
26.0 0.265
27.0 0.278
28.0 0.291
29.0 0.304
30.0 0.317
31.0 0.331
32.0 0.345
33.0 0.359
34.0 0.373
35.0 0.387
36.0 0.401
37.0 0.416
38.0 0.431
39.0 0.446
40.0 0.461
41.0 0.476
42.0 0.491
43.0 0.506
44.0 0.521
45.0 0.537
46.0 0.553
47.0 0.569
48.0 0.585
49.0 0.601
50.0 0.618
51.0 0.635
52.0 0.652
53.0 0.669
54.0 0.687
55.0 0.705
56.0 0.723
57.0 0.742
58.0 0.761
59.0 0.78
60.0 0.799
61.0 0.818
62.0 0.837
63.0 0.856
64.0 0.875
65.0 0.895
66.0 0.915
67.0 0.935
68.0 0.955
69.0 0.976
70.0 0.997
71.0 1.018
72.0 1.039
73.0 1.06
74.0 1.081
75.0 1.102
76.0 1.123
77.0 1.145
78.0 1.167
79.0 1.189
80.0 1.212
81.0 1.236
82.0 1.26
83.0 1.284
84.0 1.308
85.0 1.332
86.0 1.357
87.0 1.383
88.0 1.409
89.0 1.435
90.0 1.461
91.0 1.488
92.0 1.515
93.0 1.542
94.0 1.569
95.0 1.596
96.0 1.624
97.0 1.653
98.0 1.682
99.0 1.712
100.0 1.742
101.0 1.772
102.0 1.802
103.0 1.832
104.0 1.862
105.0 1.892
106.0 1.923
107.0 1.954
108.0 1.985
109.0 2.016
110.0 2.048
111.0 2.08
112.0 2.112
113.0 2.144
114.0 2.176
115.0 2.209
116.0 2.242
117.0 2.276
118.0 2.311
119.0 2.346
120.0 2.381
121.0 2.416
122.0 2.451
123.0 2.487
124.0 2.523
125.0 2.559
126.0 2.597
127.0 2.635
128.0 2.674
129.0 2.713
130.0 2.753
131.0 2.794
132.0 2.837
133.0 2.88
134.0 2.923
135.0 2.967
136.0 3.012
137.0 3.058
138.0 3.106
139.0 3.154
140.0 3.203
141.0 3.252
142.0 3.302
143.0 3.352
144.0 3.402
145.0 3.453
146.0 3.506
147.0 3.559
148.0 3.613
149.0 3.668
150.0 3.724
151.0 3.78
152.0 3.838
153.0 3.9
154.0 3.965
155.0 4.031
156.0 4.097
157.0 4.163
158.0 4.23
159.0 4.297
160.0 4.365
161.0 4.433
162.0 4.501
163.0 4.569
164.0 4.637
165.0 4.707
166.0 4.78
167.0 4.854
168.0 4.929
169.0 5.005
170.0 5.087
171.0 5.172
172.0 5.262
173.0 5.356
174.0 5.455
175.0 5.556
176.0 5.657
177.0 5.759
178.0 5.869
179.0 5.979
180.0 6.091
181.0 6.205
182.0 6.321
183.0 6.438
184.0 6.557
185.0 6.676
186.0 6.798
187.0 6.92
188.0 7.042
189.0 7.167
190.0 7.296
191.0 7.427
192.0 7.573
193.0 7.722
194.0 7.872
195.0 8.026
196.0 8.196
197.0 8.369
198.0 8.548
199.0 8.744
200.0 8.94
201.0 9.136
202.0 9.338
203.0 9.545
204.0 9.753
205.0 9.963
206.0 10.186
207.0 10.416
208.0 10.651
209.0 10.889
210.0 11.127
211.0 11.367
212.0 11.611
213.0 11.861
214.0 12.112
215.0 12.375
216.0 12.643
217.0 12.913
218.0 13.198
219.0 13.491
220.0 13.813
221.0 14.173
222.0 14.535
223.0 14.906
224.0 15.355
225.0 15.806
226.0 16.273
227.0 16.742
228.0 17.215
229.0 17.716
230.0 18.233
231.0 18.787
232.0 19.361
233.0 20.183
234.0 21.092
235.0 22.139
236.0 23.382
237.0 24.68
238.0 26.134
239.0 27.876
240.0 29.827
241.0 32.818
242.0 36.471
};
\addlegendentry{Ours}

    \end{semilogyaxis}
\end{tikzpicture}} 
        \vspace{-0.27in}
        \caption{Learning curve for BigNums.}
        \label{fig:eval4-bignum}
    \end{minipage}
    \vspace{-0.1in}
\end{figure}

\section{Related Work}\label{sec:related}

In this section, we survey work that is most closely related to our proposed tactic discovery method.  \\

 \vspace{-0.1in}

\noindent
{\bf \emph{Tactic Learning.}}
Previous work, particularly \textsc{Peano}~\cite{peano}, has focused on learning tactic libraries through anti-unification~\cite{anti-unification}. This method is syntactic in nature: proofs are modeled as sequences of operations, and new tactics are derived when repeated sub-sequences are identified across multiple proofs. While effective in identifying recurring patterns, this approach is limited by its inability to capture deeper semantic relationships between proof steps or perform meaningful proof refactoring. Additionally, \textsc{Peano} operates within a restricted, custom tactic language. \\

\vspace{-0.1in}

\noindent
{\bf \emph{Proof Automation Through Tactic Prediction.}}
A significant body of recent work has used neural models for \emph{tactic prediction}: given the current proof state, predict a tactic invocation to make progress in the proof~\cite{LPAR23:Tactic_Learning_Proving_Coq, holist, holophrasm, gamepad, proverbot, coqhammer, tactitian, rangoICSE2025}. These models can either fully automate proofs or collaborate with human users~\cite{song2024towards}.
Typically, the underlying prediction models are trained on human-written proof steps in the language of interest, and then used to guide a search algorithm interacting with the theorem proving environment. More recent works also explore using LLMs for proof automation~\cite{baldur, copra, palm}, finding that broad pre-training also makes current LLMs capable of predicting tactics. Our work complements these methods by alleviating the burden on the tactic predictor: since proofs using higher-level tactics are shorter, they can be generated with fewer calls to the predictor. Systems based on LLMs that are capable of in-context learning, such as Copra~\cite{copra}, make this integration particularly convenient, since they can attempt to use custom tactics without having prior knowledge about them, as long as they are given to the LLM in their context window.  \\

\vspace{-0.1in}

\noindent
{\bf \emph{Library Learning.}}
Library learning in code-related domains aims to automatically discover reusable components in both programs and formal proofs. In the context of code reuse, researchers have explored various methods for code library learning, including anti-unification techniques~\cite{ec2, iyer-etal-2019-learning}, program synthesis algorithms~\cite{stitch, dreamcoder}, and e-graphs~\cite{babble}. For theorem proving, library learning has primarily focused on extracting reusable lemmas. Kaliszyk and Urban~\cite{kaliszyk2015learning} introduced a method for automatically extracting lemmas from the Mizar system to assist in proving additional theorems. This idea was expanded by REFACTOR~\cite{zhou2024refactor}, which applied similar lemma extraction techniques to the Metamath system. More recently, Xin et al.~\cite{xin2023lego} integrated lemma proposals from large language models (LLMs) into neural theorem proving for the Lean prover. Our work on \emph{tactic} learning provides a complementary addition to these approaches, as tactics represent imperative, untyped proof construction steps rather than specific mathematical facts. While lemma learning focuses on identifying reusable truths, tactic learning captures and generalizes common proof strategies, allowing proofs to be written more concisely and at a higher level of abstraction. \\

\vspace{-0.1in}

\noindent
{\bf \emph{Semantic Code Refactoring}}
Several approaches have been developed for semantic code refactoring across various domains. Revamp~\cite{revamp} focuses on refactoring Abstract Data Types (ADTs) using relational constraints. In other areas, researchers have explored different techniques: using program invariants to detect refactoring candidates~\cite{daikon-refactor}, employing type constraints to refactor class hierarchies~\cite{type-refactor, constraint-refactor}, and applying program analysis techniques to refactor Java generics~\cite{wildcard-refactor}. Additionally, Migrator~\cite{migrator} addresses the refactoring of database programs in response to schema updates. However, these approaches  target code refactoring rather than proof refactoring. 
\\

\vspace{-0.1in}

\noindent
{\bf \emph{Guided Enumerative Synthesis.}}
Our tactic candidate enumeration procedure shares similarities with guided enumerative synthesis techniques from program synthesis literature~\cite{lambda2, synquid, regel, alpharegex, opsynth, admissible-heuristics, dreamcoder, stitch}. While most of these works aim to synthesize complete programs based on input/output examples or other specifications, Stitch~\cite{stitch} also employs a corpus-guided top-down approach to learn library abstractions for programs. However, as noted in our introduction, Stitch primarily focuses on generalizing concrete expressions to lambda abstractions. In contrast, our work emphasizes leveraging the semantics of tactic execution to discover new usable patterns.  \\

\vspace{-0.1in}

\noindent
{\bf \emph{Graph-Based Program Abstractions.}}
Tactic Dependence Graphs (TDGs) are inspired by  graph-based program abstractions~\cite{cfa,dfa,pdg,callgraph,ptg}, such as control-flow  and data-flow graphs, commonly used in program analysis. Broadly, such abstractions represent either the program's control flow (e.g., call graph) or dependencies between data (e.g., points-to graph) and are widely employed for tasks like optimization and security analysis (e.g., malware detection, code clone identification~\cite{ccgraph, gplag, semantic-clone,yogo,apposcopy,astroid}). However, TDGs differ in their purpose and design, focusing on logical proof dependencies between tactics rather than control or data dependencies.
Notably, TDGs abstract away irrelevant syntactic differences between proofs (e.g., subgoal naming or tactic order) and concentrate on the semantic relationships between tactic applications.


\section{Discussion}\label{sec:discussion}
The proof refactoring  and tactic discovery framework presented in this paper enables multiple  applications in interactive theorem proving. 
In this section, we discuss three possible use cases of our approach in the overall proof engineering workflow.

\vspace{0.05in}
\noindent
{\bf \emph{Improving Proof Automation.}}  By encapsulating recurring proof patterns into higher-level tactics, our approach helps automated tools operate at a more abstract level. Instead of repeatedly generating low-level proof steps, proof automation tools can leverage learned tactics as higher-level building blocks. This strategy can improve scalability, since the search space at a higher level of abstraction is  smaller. This approach also facilitates a form of curriculum learning, where newly discovered tactics serve as building blocks for more advanced proofs. As shown in our evaluation in Section~\ref{sec:copra}, the tactics learned by \toolname already significantly improve the success rate of a state-of-the-art proof automation tool. Notably, this improvement is observed even though current automation tools are not explicitly trained to exploit custom tactics. As proof automation methods evolve to take better advantage of custom tactics, we expect these gains to become even more pronounced.

\vspace{0.05in}
\noindent
{\bf \emph{Interactive Tactic Suggestion.}}  Our method can also assist proof engineers by suggesting potential tactics, highlighting repeated patterns in their proofs that might otherwise go unnoticed. In this interactive mode, we expect users to take inspiration from the suggested tactics while refining them to fit their domain expertise and personal style. For instance, if a learned tactic appears too general—such as requiring an excessive number of arguments—users can specialize it with concrete parameters to make it more practical. Likewise, multiple tactics proposed by \toolname{} can be merged into a single tactic that dynamically selects which one to apply based on the structure of the proof goal. Rather than prescribing a fixed way to restructure proofs, the system provides flexible recommendations that users can adapt as needed.


\vspace{0.05in}
\noindent
{\bf \emph{Proof Refactoring.}}  A third use case of our tool is automatically or interactively refactoring existing proofs, particularly in large projects that require long-term maintenance. As definitions change or new lemmas are introduced, proof engineers often need to define new tactics and restructure existing proofs to incorporate them. However, manually refactoring proofs in this way is labor-intensive. Our method simplifies this process by identifying where a given tactic—whether user-defined or adapted from a tactic discovered by \toolname—can replace existing sequences of proof steps. 

\section{Conclusion}\label{sec:conclusion}

We introduced a new approach to tactic discovery using Tactic Dependency Graphs (TDGs), which abstract away syntactic variations while capturing the logical dependencies between tactic applications. TDGs facilitate both the learning of custom tactics and the refactoring of existing proofs into more concise, modular forms. 
We implemented this method in a tool called \toolname and evaluated it on several domains, including various program logics, arbitrary-precision arithmetic, and compiler transformations. We also compared our approach against an anti-unification-based tactic discovery method from prior work ({\sc Peano}) and demonstrated the advantages of our approach in terms of the number and quality of the learned tactics: \toolname learns around 3$\times$ as many tactics compared to \peano and achieves an compression rate of $1.35$ on the test set, reducing the size of the corpus by 26\%. We also showed that the tactics learned by \toolname are useful for improving proof automation: When {\sc Copra}, a state-of-the-art proof automation tool, is supplied with the custom tactics learned by \toolname, its success rate yields a relative increase of 172\%. 
Overall, our work shows that tactic discovery provides a promising avenue for both proof refactoring and automation.


\section*{Data-Availability Statement}

Our artifact, along with the benchmark suite used for evaluation, is available on Zenodo \cite{xin_2025_15761151}.

\begin{acks}
    We thank our anonymous reviewers for their helpful feedback and support. This work was conducted in a research group supported by the
  \grantsponsor{GS100000001}{National Science
    Foundation}{http://dx.doi.org/10.13039/100000001} under Grant
  No.~\grantnum{GS100000001}{1762299}, Grant
  No.~\grantnum{GS100000001}{1918889}, Grant
  No.~\grantnum{GS100000001}{1908304}, Grant
  No.~\grantnum{GS100000001}{1901376}, Grant
  No.~\grantnum{GS100000001}{2120696}, Grant
  No.~\grantnum{GS100000001}{2210831}, and Grant
  No.~\grantnum{GS100000001}{2319471}.
\end{acks}

\bibliography{main}

\appendix
\section{Proofs}
\begin{theorem}
\label{theorem:1}
Let $\pscript$ be a proof of $\pstate$ and let $G$ be the \tdg for $\pscript$. Then, for any $\pscript'$ which is a topological ordering of $G$,
\begin{enumerate}
    \item $\pscript' \in \mathsf{InducedProofs}(G)$,
    \item $\pscript'$ is a proof of $\pstate$, and 
    \item $\pscript$ and $\pscript'$ are the same size.
\end{enumerate}
\end{theorem}
\begin{proof}
First, it's clear that $G$ is a directed acyclic graph; since $\pi$ is invoked sequentially, edges are drawn strictly from outputs of earlier tactics to inputs of later tactics in $\pi$. (A cycle would imply a back-edge from a later tactic to an earlier tactic, which is not possible.) Hence $\pi$ is a topological ordering of $G$.

Now we take the set of all goals produced in $\pi$, Goals$(\bigcup_{t \in \pi} Y_t)$.
Consider some arbitrary goal $g$ in this set:
\begin{claim}
    Any $g \in $ Goals$(\bigcup_{t \in \pi} Y_t)$ is an input of exactly one tactic in $\pi$.
\end{claim}
We show this by proving the two following claims:
\begin{itemize}[leftmargin=*]
    \item $g$ is an input of at most one tactic in $\pi$. If this were false, we can select the first two tactics $t_i, t_j \in \pi$ that take in $g \in X_{t_i}, X_{t_j}$. When we invoke $t_i$, $g$ is necessarily removed from the proof state. Because goals are unique, then $g$ cannot be an input of $t_j$, so we derive a contradiction.
    \item $g$ is an input of at least one tactic in $\pi$. If this were false, $g$ would remain in the proof state permanently, which can be shown inductively. This contradicts the assumption that $\pi$ is a valid proof.
\end{itemize}

Next, we'll begin the proof of the main theorem.
\begin{claim}
    $\pscript' \in \mathsf{InducedProofs}(G)$.
\end{claim}
In other words, $\pscript'$ generates a graph $G'$ such that $G = G'$.
\begin{enumerate}
    \item The set of vertices in $G'$ is identical to the set of vertices in $G$. This is evident because two topological orderings biject, which is standard. That is, matching tactics $t(\tname) \in \pi$ and $t'(\tname) \in \pi'$ map directly to the same node $v(\tname) \in G$.
    \item The set of edges in $G'$ is identical to the set of edges in $G$. We show this by showing that the inputs and outputs of every matching pair are identical, since each edge $(s, t, \alpha, \beta)$ corresponds exactly to a proposition that is an output of $s$ and an input of $t$. We show that for each $t \in \pi$ corresponding to $t' \in \pi'$, $X_t = X_{t'}$ and $Y_t = Y_{t'}$:

    Because tactics are deterministic, proving this fact amounts to proving that the inputs $X_{t'}$ are the same as their counterparts in $X_t$ in the proof state immediately prior to the invocation of $t'$ during the execution of $\pi'$. We prove this inductively:

    \begin{itemize}[leftmargin=*]
        \item The first tactic necessarily has no inputs and thus produces the same outputs. This corresponds to the \texttt{Proof.} step, which initiates a proof. (For simplicity, we leave this detail out of the diagrams in the body of the paper.)
        \item For every subsequent tactic $t' = (\tname, X_{t'}, Y_{t'})$, consider an arbitrary $\alpha \in \text{In}(\tname)$. Suppose in $t$, $\alpha$ selects some proposition $x \in X_t$. We show that $\alpha$ also selects $x$ as an input to $t'$, e.g. $x \in X_{t'}$. 
    
        First, $\alpha$ selects $x \in X_t$ implies some edge $(s(\tname'), t(\tname), \beta, \alpha) \in E$ where $\alpha, \beta$ correspond to $x$ in the respective tactic invocations for $t, s \in \pi$. (That is, $x \in X_t$ and $x \in Y_s$.)
        
        $s$ corresponds to some $s' \in \pi'$, and because $\pi'$ is topologically sorted, $s'$ must already be processed. Then, by the inductive hypothesis, $Y_s = Y_{s'}$. Therefore, $\alpha$ selects $x \in X_t \rightarrow x \in Y_s \rightarrow x \in Y_{s'} \rightarrow x$ was produced before the invocation of $t'$.
    
        \begin{itemize}
            \item If $x$ is a hypothesis, then it is \textit{persistent}, e.g. it is never removed and doesn't change.
            \item If $x$ is a goal, then $t'$ is the sole tactic which takes $x$ as an input. Therefore, $x$ cannot be removed by an earlier tactic invocation, implying that $x$ remains in the proof state until $t'$ is invoked.
        \end{itemize}
    
        Therefore, $\alpha$ also selects $x$ as an input for $t'$. Additionally, since all inputs are the same for $t$ and $t'$, so are their outputs.
    \end{itemize}
\end{enumerate}

\begin{claim}
    $\pscript'$ is a proof of $\pstate$.
\end{claim}
We have now shown that all tactics take the same inputs and produce the same outputs, which also implies that that Goals$(\bigcup_{t \in \pi} Y_t)$ = Goals$(\bigcup_{t' \in \pi'} Y_{t'})$.

Now consider the \textit{final} proof state $\pstate'$ after executing $\pi'$ on $\pstate$. If there are no goals in the state, then $\pi'$ is a valid proof. If there is a goal $g$ in the final state, then $g \in$ Goals$(\bigcup_{t' \in \pi'} Y_{t'}) \rightarrow g \in $ Goals$(\bigcup_{t \in \pi} Y_t)$. By our assertion above, $g$ is an input of some tactic $s \in \pi$, so $g$ is an input of the corresponding tactic $s' \in \pi'$. However, this would imply that $s' \in \pi'$ comes strictly \textit{after} the tactic sequence $\pi'$ by definition of a topological order, which is a contradiction. Thus $\pi'$ is a proof of $\pstate$.

\begin{claim}
    $\pscript$ and $\pscript'$ are the same size.
\end{claim}
Our final claim follows trivially from the one-to-one correspondence between tactics in $\pi$ and $\pi'$.
\end{proof}

\begin{theorem} 
Let $\tactic$ be a tactic, and let $\pscript$ be a proof script representing a proof of $\pstate$. Then, $\mathsf{Refactor}(\tactic, \pscript)$ yields a proof script $\pscript'$ that is also a valid proof of $\pstate$.
\end{theorem}
\begin{proof}
We show that the following invariant holds after every iteration of the loop: after every iteration, any $\pi' \in ($InducedProofs$(G_c))$ is a valid proof of $\pstate$, and (InducedProofs$(G_c)$) is nonempty. Then the theorem statement follows trivially.

First, the invariant holds at the start of the loop, where $\pi \in ($InducedProofs$(G_c))$, by simple application of Theorem~\ref{theorem:1}.

Now, assume the invariant holds for $G_c$ at the start of the loop. The existence of some induced proof indicates that $G_c$ is initially acyclic. Similarly, $G = (V, E)$ is a directed acyclic graph by the same reasoning from ~\ref{theorem:1}, since $\tau$ can be processed as a list of tactics.

We show that $G_c$ remains acyclic after contraction. Assume $G_c'$, which is the graph $G_c$ after contraction, is cyclic. Then there are three cases:
\begin{enumerate}[leftmargin=*]
    \item $t \leadsto v(\tau.\tname) \leadsto t$ for some tactic $t \notin V$. Then there exists some path $t \leadsto a$ and $b \leadsto t$ for $a, b \in V$ before contraction; this follows from a standard graph contraction definition. However, by (1) of Definition~\ref{def:collapsible}, then $t$ ought to be in $V$, a contradiction. 
    \item $v(\tau.\tname) \rightarrow v(\tau.\tname)$, e.g. $v(\tau.\tname)$ forms a self loop. Then there exists some edge $a \rightarrow b$ for $a, b \in V$ before contraction. However, by (2) of Definition~\ref{def:collapsible}, then $(a, b, \_, \_)$ ought to be in $E$, and should therefore have been removed during contraction, making the self-loop impossible.
    \item $v(\tau.\tname)$ is not part of the cycle. This implies the existence of a cycle before contracting $G_c$, which is a contradiction.
\end{enumerate}
Therefore, $G_c$ remains acyclic after contraction, which means that (InducedProofs($G_c'$)) must be nonempty. From this we obtain an arbitrary induced proof script $\pi' = t_0;...;t_i; \tau; t_{j};...;t_n$ which is topologically ordered, of which there must be at least one.

We compare $\pi'$ to the expanded proof script 
\begin{align*}
    \pi'[\tau \backslash T] = t_0,...,t_i,T_0,...,T_N,t_j,...,t_n
\end{align*}
where $\tau = (\tname, I, O, T)$. Then $\pi'$ is a topological ordering of $G_c$, since $T$ and $\pi'$ are both topologically ordered, and expanding a (topologically) ordered list with another (topologically) ordered list results in a (topologically) ordered list.

Therefore, subject to formal-to-actual renamings, $\pi'[\tau \backslash T]$ is an induced proof of $G_c$, which holds by (1) of Theorem~\ref{theorem:1}.

By definition, the custom tactic $\tau$ has identical semantics to the sequence $T_0;...;T_N$, subject to parameter renamings. Therefore, since $\pi'[\tau \backslash T]$ is a valid proof of $\pstate$ by the loop invariant, $\pi'$ is also a valid proof of $\pstate$.
\end{proof}

\setcounter{claim}{0}
\begin{theorem}\label{thm:UBcorrect}
Let $\tc = (G, \ev)$ and $\tc'= (G', \ev')$ be two tactic candidates such that $G$ is a subgraph of $G'$ and \textsf{root}$(G) =$ \textsf{root}$(G')$. Then, $\strength(\tc', \corpus) \leq {\textsc{UpperBound}}(\tc, \corpus)$. 
\end{theorem}
\begin{proof}
Given a TDG $G$, a proof $\pi$, and a witness function $f$, let $G_f$ denote $\textsf{ApplyWitness}(G, f)$, and  $G_{e, \pi}$ denote $\textsf{MaxExtension}(G_f, \textsf{TDG}(\pi))$.
By the definition $\strength$ and Algorithm~\ref{alg:ub} $\textsc{UpperBound}$, the statement to be proven can be rewritten as:
\begin{equation}
    \sum\limits_{\pi \in \corpus}\sum\limits_{f \in \ev'[\pi]} \mathmybb{1}[{\sf IsCollapsible}(f, G', \tdg(\pi))] \cdotp (\mathsf{Size}(G') - 1) \leq \sum\limits_{\pi \in \corpus}\sum\limits_{f \in \ev[\pi]}\textsf{Size}(G_{e, \pi} - 1) \tag{1}
\end{equation}

By simplification we can prove (1) by proving (2), which can be further simplified into (3):
\begin{equation}
\sum\limits_{f \in \ev'[\pi]} (\mathsf{Size}(G') - 1) \leq \sum\limits_{f \in \ev[\pi]}(\mathsf{Size}(G_{e, \pi}) - 1), \forall \pi \in \corpus    
\tag{2}
\end{equation}

\begin{equation}
(\mathsf{Size}(G') - 1) \cdotp |\ev'[\pi]| \leq (\mathsf{Size}(G_{e, \pi}) - 1) \cdotp |\ev[\pi]|, \forall \pi \in \corpus    
\tag{3}
\end{equation}

We prove (3) by proving the following two claims.
\begin{claim}
    $\forall \pi \in \corpus, |\ev'[\pi]| \leq |\ev[\pi]|$.
\end{claim}
Let $f' \in \ev'[\pi]$ for any $\pi \in \corpus$. Note that $V_G \subseteq V_{G'}$ since $G$ is a subgraph of $G'$, hence $V_G \subseteq \textsf{Dom}(f')$. We can construct a function $f$ by duplicating the mappings of $f'$ but only including a subset of its domain such that $\textsf{Dom}(f) = V_G$ and $f(v) = f'(v), \forall v \in \textsf{Dom}(f)$. Note that $f$ is a witness function that proves $G$ is an isomorphic embedding to \textsf{TDG}($\pi$), hence $f \in \ev[\pi]$. Thus, we can conclude that for each $f' \in \ev'[\pi]$, we can construct a $f \in \ev[\pi]$ that is a witness function from $G$ to \textsf{TDG}($\pi$), which proves the claim.

\begin{claim}
    $\forall \pi \in \corpus, \ev'[\pi] \neq \bot \implies \textsf{Size}(G') \leq \textsf{Size}(G_{e, \pi})$.
\end{claim}
We prove the claim by contradiction. Assume there exists a $\pi \in \corpus$, such that $\ev'[\pi] \neq \bot$ and $\textsf{Size}(G') > \textsf{Size}(G_{e, \pi})$. Let $G'_{f}, G_{f}$ denotes the subgraphs of $\textsf{TDG}(\pi)$ that are isomorphic to $G', G$, respectively, obtained by applying any witness function $f \in \ev'[\pi]$ to $G', G$. Note $\textsf{Size}(G'_{f}) > \textsf{Size}(G_{e, \pi})$ by assumption.

Note that $G_{f}$ is a subgraph of $G'_{f}$, $G'_{f}$ is a subgraph of $\textsf{TDG}(\pi)$, and $\textsf{root}(G'_f) = f(\textsf{root}(G')) = f(\textsf{root}(G)) = \textsf{root}(G_f)$. By (4) of definition \ref{def:max-ext}, $\textsf{Size}(G'_f) \leq \textsf{Size}(G_{e, \pi})$, which contradicts to the previous conclusion on the last paragraph.

Combining \textbf{Claim 1} and \textbf{Claim 2}, we conclude (3) and complete the proof.
\end{proof}

The above theorem justifies the soundness of the pruning performed at line 8 of the {\sc LearnTactic} algorithm, and the following theorem states the correctness of the overall {\sc LearnTactic} procedure. 

\begin{theorem}
Let $\tactic^*$ be the result of calling {\sc LearnTactic} on proof corpus $\corpus$. Then, we have:
\[
\tau^* = 
\arg \max_{\tactic} \  \cp(\tactic, \corpus)
\]
\end{theorem}
\begin{proof}
We approach the proof by contradiction, and we'd like to establish two claims before we start the proof: 
\begin{claim}
Let $\tactic, \pi$ be any tactic and proof. Let $G$ denote $\tdg(\tactic)$, and $\pi'$ denote $\textsc{Refactor}(\tactic, \pi)$.
\[\mathsf{Size}(\pi') = \mathsf{Size}(\pi) - \sum\limits_{f \in \ws(G, \tdg(\pi))}\mathmybb{1}[{\sf IsCollapsible}(f, G, \tdg(\pi))] \cdotp (\mathsf{Size}(\tactic) - 1)\]
\end{claim}
As shown in Algorithm~\ref{alg:refactor} {\sc Refactor}, given a tactic $\tactic$ and a proof $\pi$, it iteratively finds collapsible embeddings of $\tactic$ in $\pi$, and contracts $\pi$ by replacing the embedding of $\tactic$'s tactic body in $\tdg(\pi)$ by $\tactic$. Therefore, for each aforementioned collapsible embedding, $\textsc{Refactor}$ reduces the size of $\pi$ by $\mathsf{Size}(\tactic) - 1$, which concludes \textbf{Claim 3}.

\begin{claim}
For any tactics $\tactic_1$, $\tactic_2$, and any corpus $\corpus$, 
\[\strength(\tc_1, \corpus) \leq \strength(\tc_2, \corpus) \implies \mathsf{CP}(\tc_1.\tactic, \corpus) \leq \mathsf{CP}(\tc_2.\tactic, \corpus).\]
\end{claim}
We can conclude \textbf{Claim 4} by proving:
\begin{equation}
\strength(\tc_1, \corpus) - \strength(\tc_2, \corpus) \leq 0 \implies \mathsf{CP}(\tc_1.\tactic, \corpus) - \mathsf{CP}(\tc_2.\tactic, \corpus) \leq 0
\tag{1}
\end{equation}
We start by rewriting the left-hand-side of the assumption. Let $G_1, G_2, G'$ denote the $\tdg$ of $\tactic_1$, $\tactic_2$, and $\pi$, respectively. 
{\footnotesize
\begin{align*}
\strength(\tc_1, \corpus) - \strength(\tc_2, \corpus) =& (\textsf{Size}(\tactic_1) - 1) \cdotp \sum\limits_{\pi \in \corpus}\sum\limits_{f \in \ws(G_1, G')}\mathmybb{1}[{\sf IsCollapsible}(f, G_1, G')] - \\
& (\textsf{Size}(\tactic_2) - 1) \cdotp \sum\limits_{\pi \in \corpus}\sum\limits_{f \in \ws(G_2, G')}\mathmybb{1}[{\sf IsCollapsible}(f, G_2, G')] \\
=& \sum\limits_{\pi \in \corpus}\Big((\textsf{Size}(\tactic_1) - 1) \cdotp \sum\limits_{f \in \ws(G_1, G')}\mathmybb{1}[{\sf IsCollapsible}(f, G_1, G')] - \\
&\hspace{0.29in} (\textsf{Size}(\tactic_2) - 1) \cdotp \sum\limits_{f \in \ws(G_2, G')}\mathmybb{1}[{\sf IsCollapsible}(f, G_2, G')]\Big)
\end{align*}}
Let's continue rewriting the left-hand-side of the conclusion by using \textbf{Claim 3}:
{\scriptsize
\begin{align*}
&\mathsf{CP}(\tc_1.\tactic, \corpus) - \mathsf{CP}(\tc_2.\tactic, \corpus) \\
&= \frac{\sum\limits_{\pi \in \corpus}\textsf{Size}(\pi)}{\sum\limits_{\pi_1' \in \corpus_1'}\textsf{Size}(\pi_1')} - \frac{\sum\limits_{\pi \in \corpus}\textsf{Size}(\pi)}{\sum\limits_{\pi_2' \in \corpus_2'}\textsf{Size}(\pi_2')}\\
&= \frac{\sum\limits_{\pi \in \corpus}\textsf{Size}(\pi)}{\sum\limits_{\pi_1' \in \corpus_1'}\textsf{Size}(\pi_1)  \sum\limits_{\pi_2' \in \corpus_2'}\textsf{Size}(\pi_2')}\left(\sum\limits_{\pi_2' \in \corpus_2'}\textsf{Size}(\pi_2') - \sum\limits_{\pi_1' \in \corpus_1'}\textsf{Size}(\pi_1)\right)\\
&= \frac{\sum\limits_{\pi \in \corpus}\textsf{Size}(\pi)}{\sum\limits_{\pi_1' \in \corpus_1'}\textsf{Size}(\pi_1)  \sum\limits_{\pi_2' \in \corpus_2'}\textsf{Size}(\pi_2')} \Bigg(\sum\limits_{\pi \in \corpus}\Big(\textsf{Size}(\pi) - \sum\limits_{f \in \ws(G_2, G')}\mathmybb{1}[{\sf IsCollapsible}(f, G_2, G')] \cdotp (\textsf{Size}(\tactic_2) - 1)\Big) - \\
&\hspace{4.2cm}\sum\limits_{\pi \in \corpus}\Big(\textsf{Size}(\pi) - \sum\limits_{f \in \ws(G_1, G')}\mathmybb{1}[{\sf IsCollapsible}(f, G_1, G')] \cdotp (\textsf{Size}(\tactic_1) - 1)\Big)\Bigg)\\
&= \frac{\sum\limits_{\pi \in \corpus}\textsf{Size}(\pi)}{\sum\limits_{\pi_1' \in \corpus_1'}\textsf{Size}(\pi_1)  \sum\limits_{\pi_2' \in \corpus_2'}\textsf{Size}(\pi_2')}\sum\limits_{\pi \in \corpus}\Bigg(\sum\limits_{f \in \ws(G_1, G')}\mathmybb{1}[{\sf IsCollapsible}(f, G_1, G')] \cdotp (\textsf{Size}(\tactic_1) - 1) - \\
&\hspace{4.75cm}\sum\limits_{f \in \ws(G_2, G')}\mathmybb{1}[{\sf IsCollapsible}(f, G_2, G')] \cdotp (\textsf{Size}(\tactic_2) - 1)\Bigg)\\
&=\frac{\sum\limits_{\pi \in \corpus}\textsf{Size}(\pi)}{\sum\limits_{\pi_1' \in \corpus_1'}\textsf{Size}(\pi_1)  \sum\limits_{\pi_2' \in \corpus_2'}\textsf{Size}(\pi_2')}\Bigg(\strength(\tc_1, \corpus) - \strength(\tc_2, \corpus)\Bigg) \leq 0
\end{align*}}
We've concluded \textbf{Claim 4}. Now let's refocus to the statement of the current theorem. Let $\tc^*$ denote the corresponding tactic candidate of $\tactic^*$. By \textbf{Claim 4}, we can conclude the theorem statement by proving:
\[\tc^* = \arg \max_{\tc} \ \strength(\tc, \corpus).\]
We start our proof by contradiction with assuming there exists a tactic candidate $\tc'$, such that $\strength(\tc^*, \corpus) < \strength(\tc', \corpus)$. However, since $\tc'$ was not returned by {\sc LearnTactic}, it must have never reached line $9$ of {\sc LearnTactic}, implying that any tactic candidate $\tc''$ expandable to $\tc'$ was pruned out by {\sc UpperBound}. Note that $\tc''.G$ is a subgraph of $\tc'.G$ as $\tc''$ expands to $\tc$. By Theorem~\ref{thm:UBcorrect}, $\strength(\tc', \corpus) \leq \textsc{UpperBound}(\tc'', \corpus)$. Additionally, since $\tc''$ was pruned out, $\textsc{UpperBound}(\tc'',$\\
{\noindent $\corpus) < \strength(\tc^*, \corpus).$ Therefore, we conclude $\strength(\tc', \corpus) < \strength(\tc^*, \corpus)$, which yields a contradiction.}
\end{proof}

\section{Setup of the \textsc{Peano} baseline}
As outlined in Section~\ref{sec:exp-setup}, we have re-implemented \textsc{Peano} to support the Rocq environment. In this section, we provide a detailed exposition of this re-implementation. 

\noindent {\bf \emph{Proof tokenization.}} 
Given a corpus of proofs, each proof is segmented into a sequence of tactics, which are analogous to ``actions'' in \textsc{Peano} terminology. This approach represents each proof as a single linear sequence of tactics, implying a sequential dependency structure. However, this linear representation may oversimplify the more graph-like dependency structure present in Rocq proofs.

\noindent {\bf \emph{Learning maximally compressive tactics.}} 
In the \textsc{Peano} framework, the maximally compressive composite tactic, $\tactic$, for a given corpus of proofs, $\corpus$, is constructed from the longest consecutive common subsequence shared among a subset of proofs, $\corpus'$, with argument anti-unification applied (detailed in the following subsection). This composite tactic is selected to maximize compression, measured by its compression power $(\mathsf{Size}(\tactic) - 1) \cdot |\corpus'|$. To identify the optimal tactic $\tactic$, we exhaustively search for the longest common subsequence between every pair of proofs in the corpus, from the results of which we ultimately choose the subsequence that offers the highest compression power across the entire corpus.

\noindent {\bf \emph{Tactic argument anti-unification.}} 
During the search for the maximally compressive custom composite tactic, we employ argument anti-unification across sequences of tactic tokens. Specifically, given two sequences of tactics that share identical tactic names but differ in their list of formal arguments, denoted by $\vec{I}_1$ and $\vec{I}_2$, we construct a generalized tactic sequence with the minimally sufficient set of input identifiers, $\vec{I}$. This $\vec{I}$ is defined such that there exists a valid mapping from each identifier $i$ in $\{i \in I | \forall I \in \vec{I}\}$ to the corresponding identifiers in both $\{i \in I | \forall I_1 \in \vec{I}_1\}$ and $\{i \in I | \forall I_2 \in \vec{I}_2\}$. $\vec{I}$ is known as the \emph{least general generalization} of $\vec{I_1}$ and  $\vec{I_2}$ \cite{peano}.

For example, consider the two tactic sequences \verb|intros H H0 H1; rewrite H0 in H1| and\\
{\noindent \verb|intros H H0 H1; rewrite H in H1|. The resulting generalized tactic, capturing the common structure, would be defined as: \verb|Ltac H H0 H1 H2 := intros H H0 H1; rewrite H2 in H1.|}

\section{Integration with \sc{Copra}}

In Section~\ref{sec:copra}, we described our experiment measuring the impact of combining our custom tactics with \textsc{Copra}, an LLM-based tool for proof automation. We used \textsc{Copra} for this experiment due to its flexibility with respect to custom tactics (e.g., tools like CoqHammer cannot leverage them), and the fact that it leverages in-context learning. Previous proof automation tools based on machine learning require first training a neural network on existing proofs \cite{proverbot}. However, these architectures cannot predict tactics that weren't seen during training, while LLMs can still attempt to use them as long as we provide enough information about the new tactics in their context window.

Our concrete implementation thus had 4 steps:

\begin{description}
    \item[Refactoring: ] We first refactored the existing proof scripts using the discovered custom tactics. We also insert the new tactic definitions immediately before the first lemma that uses that tactic. This gives us a set of valid Rocq proof scripts that we know use the new tactics. Moreover, we initialize the \textsc{Copra} agent on this refactored script, so that it will be able to invoke the custom tactics.
    \item[Indexing available tactics: ] For each lemma $L$ that uses custom tactics, we  index offline which custom tactics have been declared before the statement of $L$. These are all the available tactics that we must describe to \textsc{Copra} when it tries to prove $L$.
    \item[Example extraction: ] We run the refactored proof script through the Coq SerAPI, and log all of the proof states along with the tactic invocations that were observed. From these, we can extract usage examples for each of the custom tactics. We save these examples in the index, along with what lemma they came from. We found even one example to help significantly: given only the tactic definitions, GPT-4o would generally avoid using the new tactics, or use them with invalid arguments.
    \item[Prompt hacking: ] Finally, at test time, when \textsc{Copra} is proving a lemma, we modify the prompt construction that takes place in its ``DFS Agent'' to first (1) look up the current lemma in the index, and (2) if there are custom tactics declared before the current lemma, add both the tactics' declarations and one usage example per tactic, taken from a different lemma. This gives \textsc{Copra} some indication of when this tactic might be useful.    
\end{description}

We give one representative prompt on the CoqArt benchmark in Figure~\ref{fig:copra-prompt}. The prompt starts with a common system message explaining how the \textsc{Copra} DFS agent works, which we abbreviate. Then, we describe each available tactic along with one usage example for each.

\begin{figure}
\begin{mdframed}
\footnotesize
    \begin{verbatim}
<<system>>
You are a proficient formal theorem-proving agent in Coq. You can predict the next proof step 
given the current proof state, relevant definitions, and some possible useful lemmas/theorems. 
The proof state is described in the following format:
1. All the goals are described under `[GOALS]` keyword.

[...]

<<user>>
In your proof, you can use the following custom tactics that were defined for this particular 
domain:

### Available domain-specific tactics ###

Tactic definition: Ltac custom2 H0 H1 :=  intros H0 H1; [assumption |..].
Real example of this tactic being used in context:
    Goal: forall (n : nat) (_ : lt (Init.Nat.mul 7 5) n), le (Init.Nat.mul 6 6) n
    Action: custom2 n H .

Tactic definition: Ltac custom0 H0 :=  apply H0; [assumption |..].
Real example of this tactic being used in context:
    Goal: Q
    Relevant Context: H : forall _ : P, Q
    Action: custom0 H .

Tactic definition: Ltac custom1 H0 H1 :=  unfold H0; [apply H1 |..].
Real example of this tactic being used in context:
    Goal: lt 8 9
    Action: custom1 lt le_n .

When used properly, these tactics can save you several steps in your proof, and increase your 
chances of success, so use them whenever appropriate.

Goals to prove:
[GOALS]
[GOAL] 1
forall (A : Type) (P Q : A -> Prop),
(exists y, P y) -> (forall x : A, P x -> Q x) -> exists y, Q y
[END]
    \end{verbatim}
\end{mdframed}
    \caption{Example of a \textsc{Copra} prompt we generate on the CoqArt benchmark when learned custom tactics are available.}
    \label{fig:copra-prompt}
\end{figure}

\newpage
\section{Example Extracted Custom Tactics}
In this appendix, we provide representative examples of custom tactics extracted by \toolname, along with descriptions of their functionality and appropriate use cases in Figure~\ref{fig:sample-custom-tactics}.\\

\begin{figure}
\begin{mdframed}
\small
\begin{verbatim}
(* Introduces H0 and simplifies it. *)
(* Best applied when H0 contains complex expressions that could be reduced before 
   further processing. *)
Ltac simplifiedIntro H0 := intros H0; [simpl in H0].

(* Inverts H0, then immediately tries contradiction. *)
(* Best applied when H0 leads to an immediate contradiction after inversion *)
Ltac invertContradict H0 := inv H0; [contradiction | .. ].

(* Solves a disjunctive goal by providing a variable to complete the left case of 
   the goal. *)
(* Best applied when the goal is of the form (exists x, P x), is a disjunction and
   left is the correct case. *)
Ltac leftWitness H0 := left; [exists H0; [auto]].

(* Splits the goal, and solves one with apply and the other with auto. *)
(* Best applied when the goal is a conjunction, and H0 can be applied to solve the
   first goal. *)
Ltac applyAndAuto H0 := split; [apply H0 | auto ].

(* Introduces H0, H1, and eliminates H1 by case analysis, solving each case with 
   auto. *)
(* Best applied when H1 is a disjunction or an existential needing case splitting. *)
Ltac elimAuto H0 H1 := intros H0 H1; [elim H1; [auto | auto | .. ]].

(* Introduces an existential witness, splits into subgoals, and sovles each subgoal
   with tauto. *)
(* Best applied when proving existentially quantified goals of the form (exists x, 
   P x) where P x can be split into simpler subgoals. *)
Ltac splitSolve H0 := exists H0; [split; [tauto | tauto]].

(* Introduces and case splits a variable, and introduces more variables in each
   subgoal. *)
(* Best applied when H0 is a inductive type variable or has multiple contructors. *)
Ltac caseIntro H0 H1 H2 := intro H0; [case H0; [intro H2 | intro H1 | .. ]].

(* Uses induction on the first hypothesis, simplifying each step, and recursively
   applies induction. *)
(* Best applied when proving properties over recursively defined structures (e.g.,
   lists, trees, etc). *)
Ltac nestedInduct1 := induction 1; [simpl; [try tauto] | simpl; [induction 1] | .. ].

(* Introduces an existential witness and splits into nested goals. *)
(* Best applied when the goal is of the form (exists x, P x) and is structured for 
   multiple levels of splitting. *)
Ltac nestedSplitSolve H0 := exists H0; [split; [auto | split; [tauto | tauto]]].
\end{verbatim}
\end{mdframed}
\caption{Example of a \textsc{Copra} prompt we generate on the CoqArt benchmark when learned custom tactics are available.}
\label{fig:sample-custom-tactics}
\end{figure}

\end{document}